\documentclass[11pt]{article}
\usepackage{cavity}

\newcommand{\defeq}{\stackrel{\mathrm{def}}=}

\newcommand{\sidtemp}[1]{}

\newcommand{\siqitemp}[1]{}

\newcommand{\prasadtemp}[1]{}

\newcommand{\IntroNull}{\mathcal{D}_{\text{null}}}
\newcommand{\IntroPlanted}{\mathcal{D}_{\text{planted}}}

\newcommand{\Tr}{\mathsf{Tr}}
\newcommand{\Model}{\mathrm{M}}
\newcommand{\Null}{\Model^{\times}}
\newcommand{\Planted}{\Model}
\newcommand{\TypeDist}{\mathcal{T}}
\newcommand{\ColDist}{\mathbb{P}}
\newcommand{\arity}{\mathrm{a}}

\newcommand{\FacType}{\theta}
\newcommand{\Modeln}{\Model_{n}}
\newcommand{\Ham}{\mathbf{H}}
\newcommand{\Trans}{\mathbf{M}}

\newcommand{\Dcolor}{\mathbf{D}}
\newcommand{\Facm}{\mathbf{\Psi}}
\newcommand{\Class}{\mathrm{Cl}}
\newcommand{\Cl}{\Class}
\newcommand{\Poi}{\mathrm{Poisson}}
\newcommand{\MatCol}{\mathrm{Bl}}
\newcommand{\numcols}{q}

\newcommand{\BlankCol}{\MatCol}
\newcommand{\Adj}{A}
\newcommand{\CentAdj}{\ul{A}}
\newcommand{\CompG}{\calK}
\newcommand{\Bip}{\mathrm{Bip}}
\newcommand{\NB}{\mathrm{NB}}

\newcommand{\FacDeg}{\phi}
\newcommand{\FacClass}{\chi}
\newcommand{\GModPar}{\Model^{\times}}
\newcommand{\Lkgs}{\mathrm{Lkgs}}
\newcommand{\nobikes}{\bcalE}
\newcommand{\Mean}{\mu}
\newcommand{\Clos}{\mathrm{Clos}}
\newcommand{\Exc}{\mathrm{Exc}}
\newcommand{\Profl}{\Delta}
\newcommand{\Sh}{\mathrm{Sh}}
\newcommand{\Shps}{\mathrm{Shps}}
\newcommand{\Trs}{\mathrm{Trs}}
\newcommand{\Seg}{\mathrm{Seg}}
\newcommand{\UTrs}{\mathrm{UTrs}}
\newcommand{\FTrs}{\mathrm{FTrs}}
\newcommand{\Lm}{\mathrm{Lm}}
\newcommand{\Sing}{S}
\newcommand{\Dup}{D}
\newcommand{\LSh}{\mathrm{LSh}}
\newcommand{\RType}{\FacType}
\newcommand{\Index}{\mathrm{i}}
\renewcommand{\index}{\mathrm{i}}
\newcommand{\goin}{\mathrm{in}}
\newcommand{\goout}{\mathrm{out}}
\newcommand{\innerp}{\mathbf{H}}

\newcommand{\ldist}{\mu}
\newcommand{\minit}{\overline{m}}
\newcommand{\ACFacDeg}[1]{\ol{\phi_{#1}}}
\newcommand{\prior}{\mathbb{P}}
\newcommand{\alg}{\mathsf{A}}
\newcommand{\varvar}[2]{\mathrm{num}_{#1}(#2)}

\newcommand{\taup}{\tau'}

\newcommand{\Tree}{\bbT}

\newcommand{\bTrans}{\overline{\Trans}}

\newcommand{\Diag}{\text{Diag}}
\newcommand{\tot}[1]{\xrightarrow{#1}}
\newcommand{\lab}{\eta}
\newcommand{\boltz}{\mu}
\newcommand{\fix}{\mathrm{fp}}
\newcommand{\support}{\text{support}}

\newcommand{\BPM}[3]{\bTrans_{#1,#2\mid #3}}
\newcommand{\Erdos}{Erd\H{o}s}
\newcommand{\Renyi}{R\'enyi}
\newcommand{\wgr}{\mathrm{wgr}}
\newcommand{\mwgr}{\mathrm{mwgr}}

\newcommand{\varfac}[2]{\deg_{#2}({#1})}

\newcommand{\varfacbin}[2]{b_{#1}^{#2}}

\newcommand{\statvec}{g}
\newcommand{\Path}{p}
\newcommand{\Pathp}{\Tilde{p}}

\newcommand{\expwlk}{\mathbf{w}}

\newcommand{\Chi}{\chi}
\newcommand{\labp}{\Tilde{\eta}}
\newcommand{\SA}{\mathrm{SA}}
\newcommand{\SACentAdj}{\ol{A}}
\newcommand{\Pathu}{p^{\cup}}

\newcommand{\CentBdj}{\ul{B}}
\newcommand{\CentRdj}{\ul{R}}
\newcommand{\nb}{\circ}
\newcommand{\bchat}{\ul{\chi}}
\newcommand{\bmu}{\mathbf{\mu}}
\newcommand{\bfe}{\mathbf{e}}
\newcommand{\myrho}{\left(\sqrt{\lambda_L}\right)}

\newcommand{\ualpha}{\ul{\alpha}}
\newcommand{\ubeta}{\ul{\beta}}

\newcommand{\iprod}[1]{\langle#1\rangle}
\newcommand{\Iprod}[1]{\left\langle#1\right\rangle}

\begin{document}

\title{On statistical inference when fixed points of belief propagation are unstable}

\author{Siqi Liu\thanks{EECS Department, University of California Berkeley.  \texttt{sliu18@berkeley.edu}. Supported in part by the Berkeley Haas Blockchain Initiative and a donation from the Ethereum Foundation.} \and Sidhanth Mohanty\thanks{EECS Department, University of California at Berkeley.  \texttt{sidhanthm@berkeley.edu}.  Supported by Google PhD Fellowship.} \and Prasad Raghavendra\thanks{EECS Department, University of California Berkeley.  \texttt{prasad@cs.berkeley.edu}. Research supported by grants NSF 2007676 and NSF 2023505.}}
\date{\today}

\maketitle

\begin{abstract}
	
	 Many statistical inference problems correspond to recovering the values of a set of hidden variables from sparse observations on them.  For instance, in a planted constraint satisfaction problem such as planted 3-SAT, the clauses are {\it sparse observations} from which the hidden assignment is to be recovered.  In the problem of community detection in a stochastic block model, the community labels are hidden variables that are to be recovered from the edges of the graph.

	Inspired by ideas from statistical physics, the presence of a stable fixed point for belief propogation has been widely conjectured to characterize the computational tractability of these problems.  For community detection in stochastic block models, many of these predictions have been rigorously confirmed.

	In this work, we consider a general model of statistical inference problems that includes both community detection in stochastic block models, and all planted constraint satisfaction problems as special cases.  We carry out the cavity method calculations from statistical physics to compute the regime of parameters where detection and recovery should be algorithmically tractable.  At precisely the predicted tractable regime, we give:
	\begin{enumerate}[(i)]
		\item a general polynomial-time algorithm for the problem of \emph{detection}: distinguishing an input with a planted signal from one without;
		\item a general polynomial-time algorithm for the problem of \emph{recovery}: outputting a vector that correlates with the hidden assignment significantly better than a random guess would.
	\end{enumerate}
	Analogous to the spectral algorithm for community detection \cite{krzakala2013spectral,BLM15}, the detection and recovery algorithms are based on the spectra of a matrix that arises as the derivatives of the belief propagation update rule.  To devise a spectral algorithm in our general model, we obtain bounds on the spectral norms of certain families of random matrices with correlated and matrix valued entries.  We then demonstrate how eigenvectors of various powers of the matrix can be used to partially recover the hidden variables.

\end{abstract}

\thispagestyle{empty}
\newpage

\thispagestyle{empty}
\tableofcontents
\thispagestyle{empty}
\setcounter{page}{0}
\newpage

\newcommand{\ER}{Erd\H{o}s-Renyi\ }
\newcommand{\pin}{p_{in}}
\newcommand{\pout}{p_{out}}
\newcommand{\col}{[q]}
\newcommand{\types}{T}

\section{Introduction}

In the {\sc Planted-$q$-Coloring} problem, a hidden coloring $\bc: [n] \to \{1,\ldots,q\}$ is sampled from the uniform distribution over $[q]^n$.  A random graph $G = ([n], \bE)$ is drawn from the \ER distribution conditioned on $\bc$ being a legitimate coloring.  So every edge $(i,j)$ is included in the graph with probability  $\frac{d}{n} \cdot 1[\bc(i) \neq \bc(j)]$ independently at random.
Given the edges $\bE$ as input, the goal of an inference algorithm is to recover (even partially) the hidden coloring $\bc$.

{\sc Planted-$q$-Coloring} is the archetypal example of a broad class of statistical inference problems where the goal is to recover a set of hidden variables from sparse observations on it (see \cite{andrea2008estimating}).
A large number of inference problems ranging from decoding LDPC codes to community detection in random graphs fall into this broad framework.
Broadly speaking, the setup in these inference problems is as follows.  A set of {\it hidden variables} $\{\bc(1),\ldots \bc(n)\}$ are drawn from a known prior product distribution $\bbP_{\bc}$.   A sequence of {\it observations} (a.k.a. hyperedges) $\bE$ on these hidden variables are revealed to the algorithm.  Each hyperedge $(i_1, \ldots, i_k)$ is included with probability $\frac{1}{n^{k-1}} \cdot \Phi(\bc(i_1),\ldots, \bc(i_k))$ for some constant $\Phi(\bc(i_1),\ldots, \bc(i_k))$ that depends on the values of hidden variables $\bc(i_1),\ldots,\bc(i_k)$.
Thus the inference algorithm receives $\Theta(n)$ observations with high probability and its goal is to partially recover the values of the hidden coloring.

The key computational task 
is to recover the values of the hidden variables.  %
In a sparse setup where the number of observations is linear, it is typically impossible to recover the hidden variables exactly.  Therefore, one settles for the relaxed goal of {\it weak} recovery where the algorithm is required to produce an assignment which correlates better than random with hidden variables.

It is often useful to also define a related decision problem of "detection" .  
Here, the algorithm is required to distinguish between a set of observations consistent with a single fixed assignment to hidden variables (planted distribution) or a set of observations each sampled independently by drawing a new assignment to the hidden variables (null distribution).

%


In this work, we will be considering a more general model that will permit constantly many {\it types} of variables and observations.  The prior distribution of each variable depends on its type, and the probability of  sampling an observation depends on the types and values of variables involved.  We defer the formal description of our general model to \pref{sec:obs-model}, but instead present a few examples of these problems.

\begin{example}(Stochastic Block Models)
A natural generalization of the {\sc Planted-$q$-Coloring} problem is the stochastic block model (SBM).
The stochastic block model is defined by a parameter $\numcols$ (the number of labels), a distribution $\ColDist_{\bc}$ over $[\numcols]$ (the expected fraction of vertices with a specific label), and a matrix $P \in \R^{[\numcols]\times[\numcols]}$ such that $P[c,d]$ gives the probability of an edge between two vertices with labels $c$ and $d$. In the community detection problem, a hidden labelling $\bc:[n]\to\{1,\ldots,\numcols\}$ is sampled from the product distribution $\ColDist_{\bc}^{n}$. Given $\bc$, a random graph $G = ([n],\bE)$ is drawn by including each edge $(u,v)$ independently with probability $P[\bc(u),\bc(v)]$ depending on the labels of the endpoints. The goal of the problem is to recover the labelling $\bc$ from the graph $G$.
\end{example}

\begin{example} (Planted CSPs) In a planted CSP over a domain $[q]$, an assignment $x \in [q]^n$ is chosen at random and clauses are sampled conditioned on being satisfied by the planted assignment $x$.
Depending on the predicate used, one obtains different planted CSPs such as Planted NAE-$k$-SAT and Planted $k$-SAT. 
\end{example}

Many more examples of problems that fit our framework will be presented in the rest of the paper.  Alternatively, this class of problems can be viewed as {\it ``Bayesian CSPs''}.  Traditionally, a constraint satisfaction problem involves variables taking values over finite domain and a set of local constraints on them.  The goal is to find an assignment that satisfies either all the constraints (exact CSPs) or the largest fraction of constraints (approximate CSP).  The key difference in this setup is that there is a prior distribution associated with  assignment on the variables and the constraints.

Constraint satisfaction problems (CSP) lie at the bedrock of worst-case complexity theory tracing back all the way to SAT and NP-completeness and by now there is a rich and comprehensive theory that correctly predicts the computational complexity of the traditional CSPs, with (i) the CSP dichotomy conjecture \cite{Sch78,Zhuk20} for exact CSPs, which cleanly classifies a constraint satisfaction problem as polynomial-time solvable or $\NP$-hard depending on whether a pair of solutions could be combined to form a third solution via a function called a polymorphism, and (ii) the Unique Games Conjecture for approximate CSPs, which characterizes the best approximation ratio possible in polynomial time with an integrality gap of a semidefinite program \cite{Kho02,KKMO07,Rag08}.  There is also a well understood picture of the complexity of refutation of random CSPs
from the lens of the Sum-of-Squares semidefinite programming hierarchy \cite{AOW15,RRS17,KMOW17}.
On the other hand, our understanding of the complexity of Bayesian CSPs is still in its nascent stages.
Bayesian CSPs are a rich and natural class of average case problems, and understanding their complexity would be a good test-bed for average case complexity theory.
Indeed, Goldreich's pseudorandom generator \cite{Gol11} is precisely based on harnessing the computational intractability of certain Bayesian CSPs.

A naive exponential-time algorithm for the problem would be to use the Bayes rule to compute/sample from the conditional distribution $\bc| \bE$.  
The fundamental question here is to understand the limits of efficient algorithms for this class of statistical inference problems.
Furthermore, both exact and approximate versions of traditional CSPs exhibit abrupt transitions wherein the computational complexity of the problem changes from polynomial to exponential.
It is a compelling question whether Bayesian CSPs also exhibit similar abrupt transitions in computational complexity, and whether there exist broadly applicable optimal algorithms for them.

\subsection{Belief Propogation and Cavity Method}

A natural candidate for an optimal algorithm for Bayesian CSPs (especially in the sparse case) is belief propogation (BP).
BP is often hypothesized to be theoretically optimal, and is also very efficient in practice.
There is a vast body of literature on belief propogation (BP) drawing ideas from statistical physics (see \cite[Chapter 14]{montanaribook} and \cite{ZK16} for a comprehensive treatment).
It is often very difficult to analyze BP as a standalone algorithm and we are quite far from demonstrating its optimality among polynomial-time algorithms.
However, there has been a growing body of work in the past decade which suggest a very general and precise theory to predict the computational complexity of Bayesian CSPs.

To the best of our knowledge, it appears that the work of Krzakala and Zdeborova \cite{KrzakalaZ09} is the first to hypothesize a precise computational phase transition for planted problems based on ideas from statistical physics.
Specifically, Krzakala and Zdeborova \cite{KrzakalaZ09} hypothesized that for a broad class of planted distributions, the problem of distinguishing the planted vs null distributions becomes computationally intractable at a well-defined threshold.
In the case of community detection, this threshold coincides with the so-called Kesten-Stigum threshold.  More broadly, in this work, we will often refer to this threshold of intractability for Bayesian CSPs as {\it the stable fixed point barrier} for reasons that will be soon clear.

Building on the ideas from \cite{KrzakalaZ09}, \cite{decelle2011inference,DKMZ} made a fascinating set of conjectures on community detection. For example, they conjectured that the $k$-coloring problem is easy exactly when the average degree of a vertex in the model satisfies $d > k^2$. Their conjectures fuelled a flurry of work, leading to algorithms that match the conjectured computational thresholds \cite{mossel2018proof,M14,BLM15,AS15}.

The {\it stable fixed point barrier} suggested by \cite{KrzakalaZ09,DKMZ} is applicable beyond the setting of community detection.
For instance, Krzakala and Zdeborova point out that this stable fixed point barrier is shared by problems such as hyper-graph bicoloring and locked CSPs.  Here locked CSPs are those wherein every pair of assignments to a predicate have Hamming distance at least $2$ (analogous to pairwise-independence leading to approximation resistance \cite{AM09}).
More broadly, there is a heuristic cavity method calculation to pinpoint the location of the stable fixed point barrier in general (see \pref{sec:stable-barrier} to \pref{sec:neighorhood-of-variable}).

To illustrate the rich and precise predictions of this heuristic calculation yields, we will show three examples here.  

\begin{example}
First, consider the problem of planted NAE3SAT wherein there is a uniformly random assignment in $\{0,1\}^n$ and Not-All-Equal clauses on $3$ variables are sampled so that a $\rho$-fraction of them are satisfied.  As one varies the average constraint-degree of a variable $d$ and the approximation $\rho$, there is an explicit prediction of the region of parameters where the distinguishing/recovery problem is computationally tractable (blue region in \pref{fig:NAE3SAT}).  Interestingly, the spectral (and basic SDP) refutation threshold for regular NAE3SAT was determined to be $13.5$ in \cite{DMOSS19} and similar techniques point to the threshold being $12.5$ for an ``\Erdos-\Renyi'' version of the model.  However, our results imply a distinguishing algorithm between random and planted NAE3SAT at a much smaller degree of $4.5$, which suggests this planting is not ``quiet'' and raises the question of what a quiet planted distribution is.

\myfig{0.3}{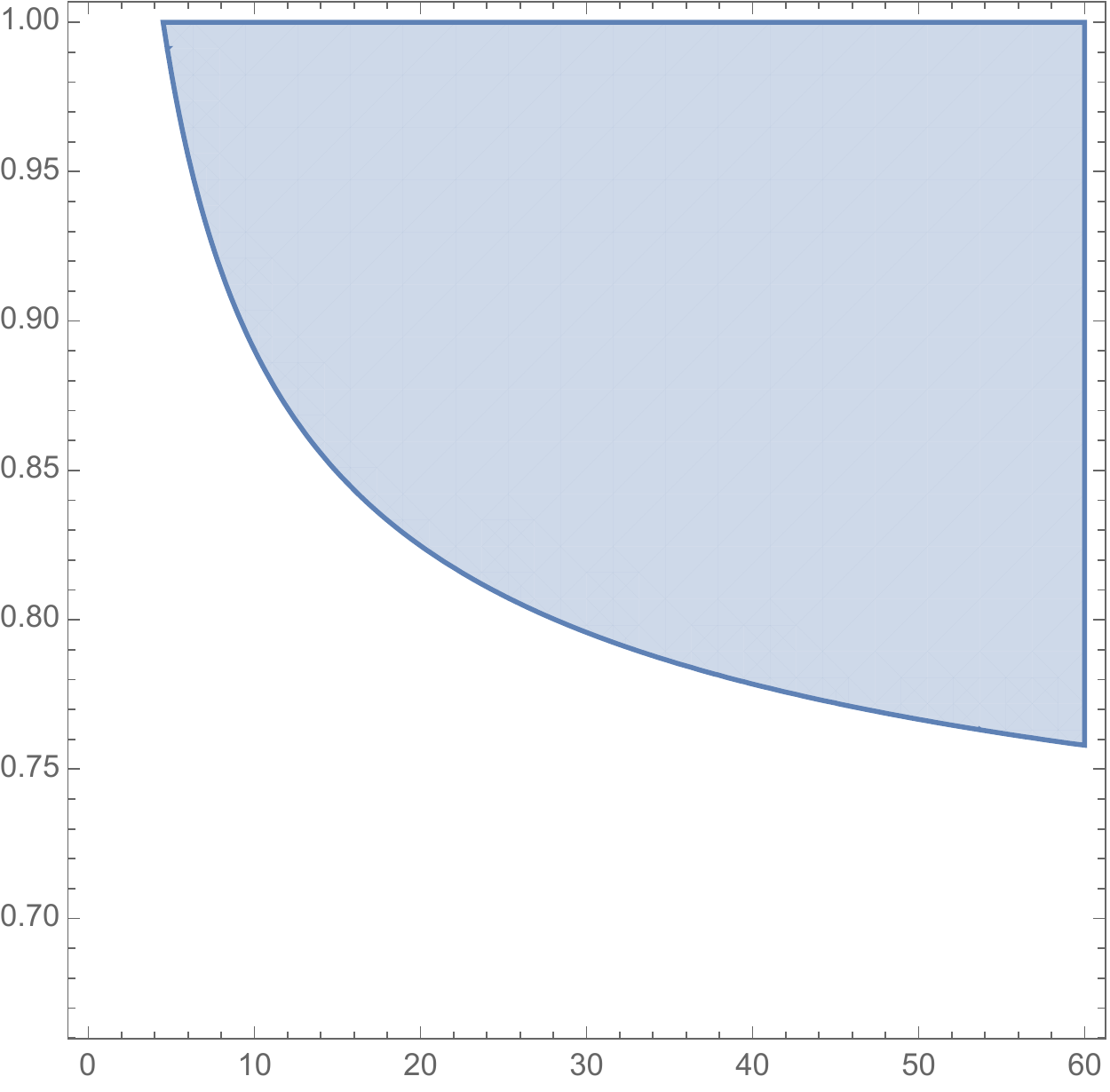}{Easy region for planted $\naethree$ shaded in blue.  Average degree on $x$-axis, fraction of clauses satisfied on $y$-axis.}{fig:NAE3SAT}
\end{example}
 
\begin{example}
	Next, we turn our attention to mixed planted CSPs.  For concreteness, we consider one particular example: planted NAE-$(3,5)$-SAT.  In this example, the variables are given a uniformly random assignment in $\{0,1\}^n$ and Not-All-Equals clauses are sampled to be on $3$ variables with probability $p$ and on $5$ variables with probability $1-p$.  As one varies the constraint-degree of a variable $d$ and the proportion of NAE3SAT clauses $p$, we can plot a precise region of parameters where the distinguishing/recovery problem is computationally tractable (blue region in \pref{fig:NAE35SAT}).
\end{example}

\myfig{0.3}{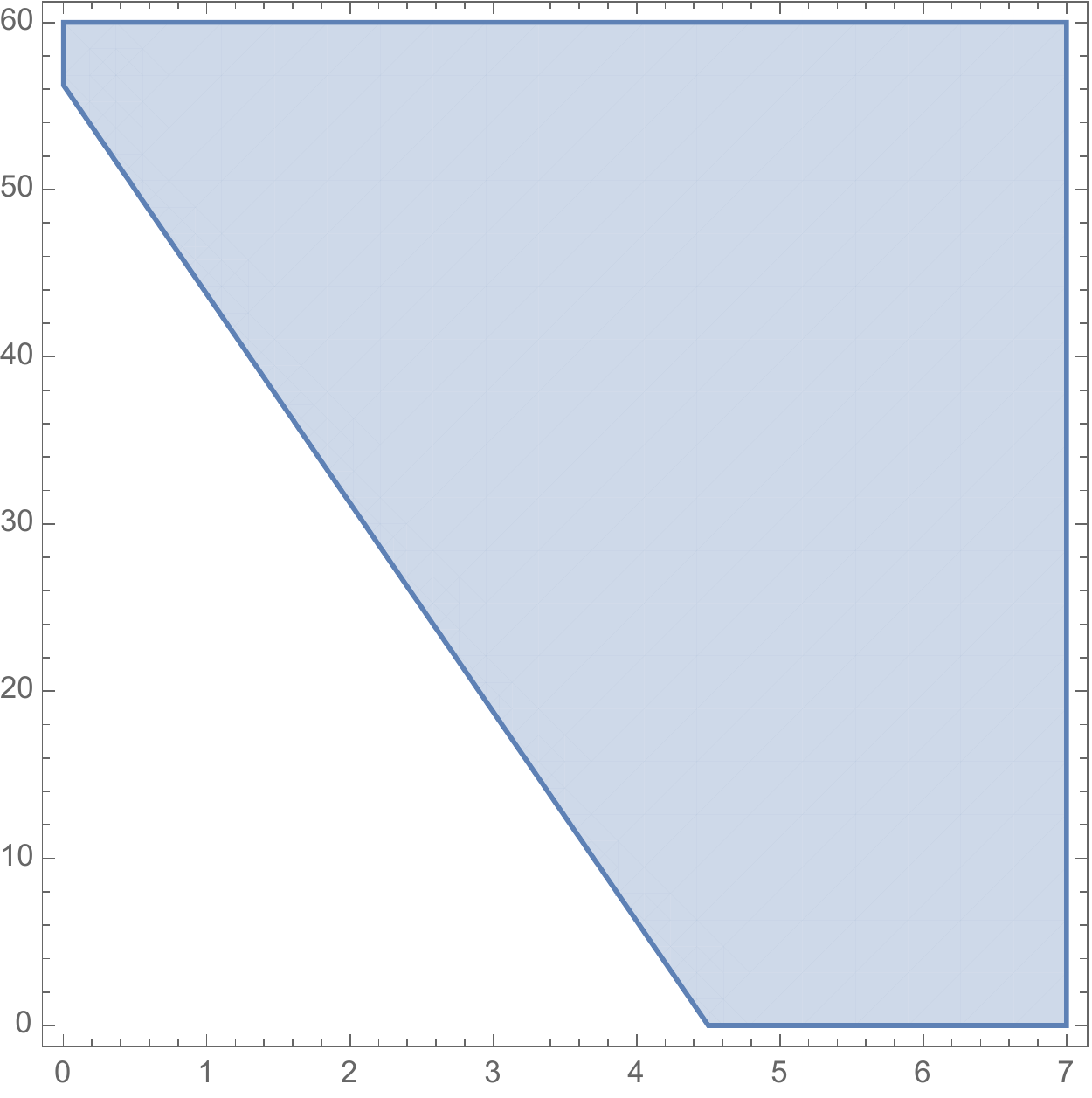}{Easy region for planted NAE-$(3,5)$-SAT shaded in blue.  Average NAE3-degree of vertex on $x$-axis, average NAE5-degree of vertex on $y$-axis.}{fig:NAE35SAT}

\begin{example}
Consider the following version of $4$-community stochastic block model with communities labeled $(0,0)$, $(0,1)$, $(1,0)$ and $(1,1)$ and $3$ parameters $d_0$, $d_1$ and $d_2$.  For a pair of vertices $u$ and $v$ from communities $x$ and $y$ we place an edge between $u$ and $v$ with probability $\frac{d_{\mathrm{dist}(x,y)}}{n}$ where $\mathrm{dist}(x,y)$ is the Hamming distance between $x$ and $y$.  For an additional twist, let us suppose that the first coordinate of the community that every vertex belongs to is also revealed to the algorithm.  What is the region of parameters $d_0,d_1,d_2$ for which an efficient algorithm can partially recover the second coordinate of the community labels?  See \pref{fig:partially-rev-4-com} for the hypothesized transition.
\begin{figure}[h]
	\begin{subfigure}{.3\textwidth}
		\centering
		\includegraphics[width=.9\linewidth]{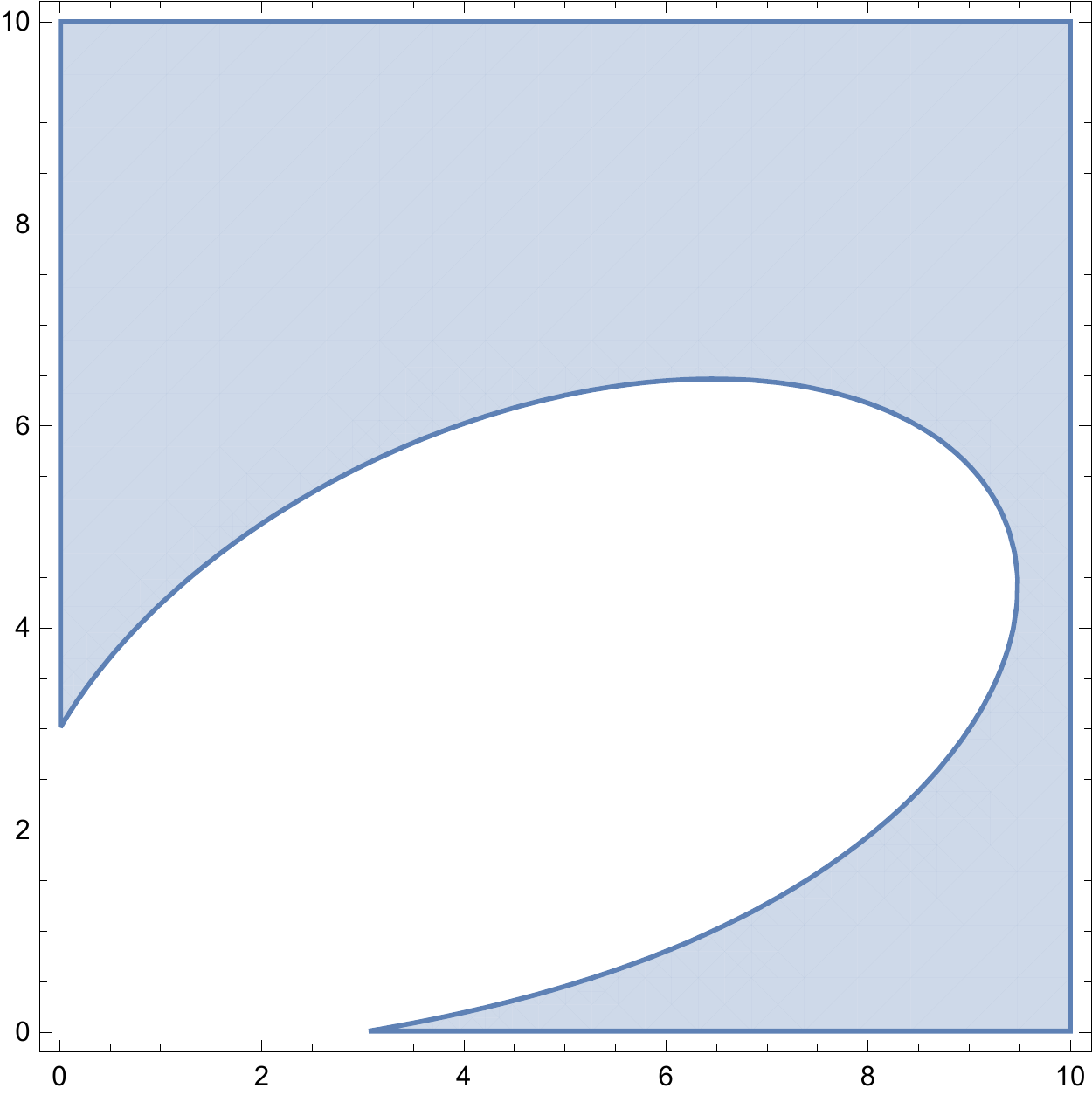}
		\caption{$d_2=1$}
	\end{subfigure}
	\begin{subfigure}{.3\textwidth}
		\centering
		\includegraphics[width=.9\linewidth]{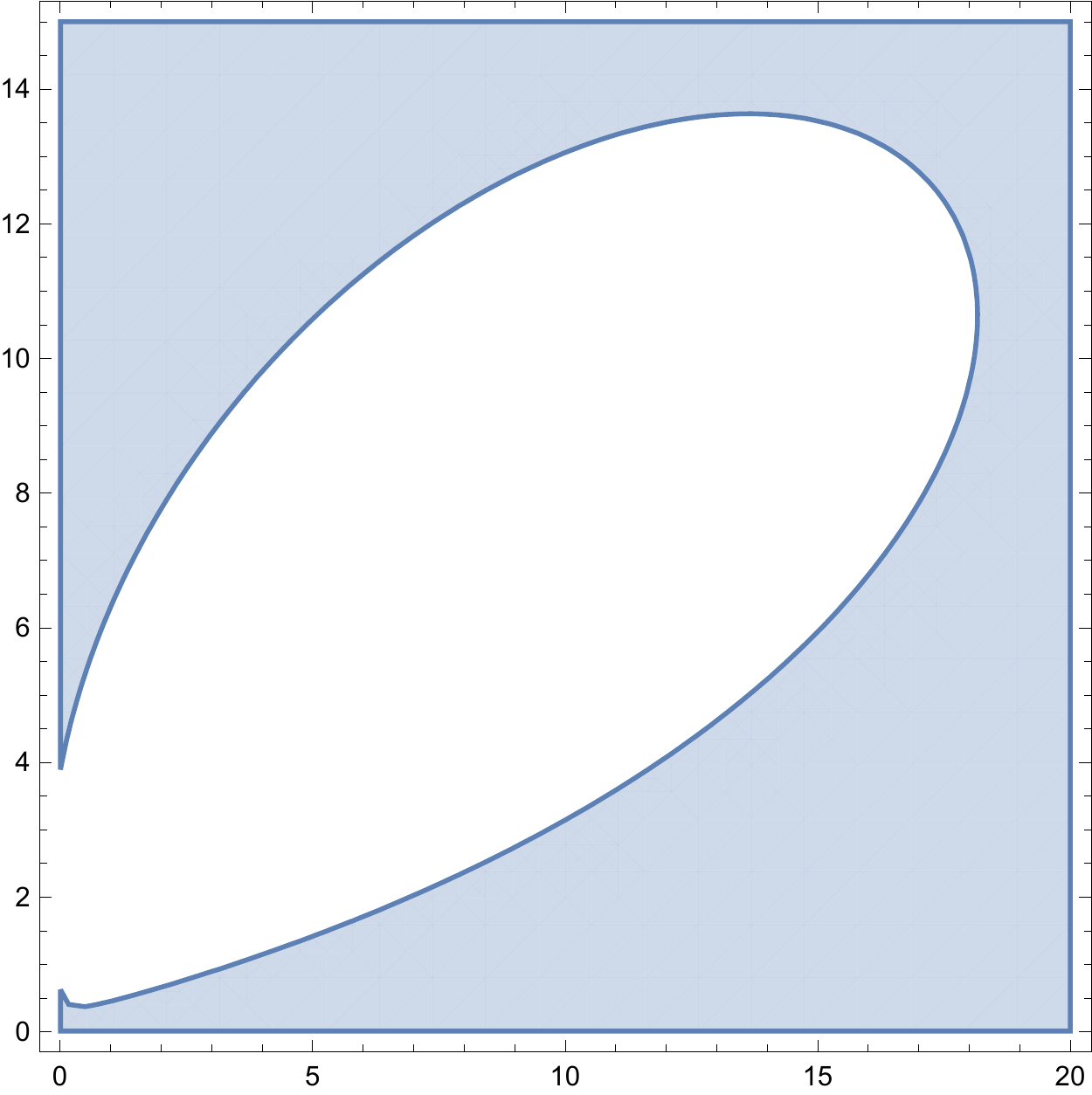}
		\caption{$d_2=5$}
	\end{subfigure}
	\begin{subfigure}{.3\textwidth}
		\centering
		\includegraphics[width=.9\linewidth]{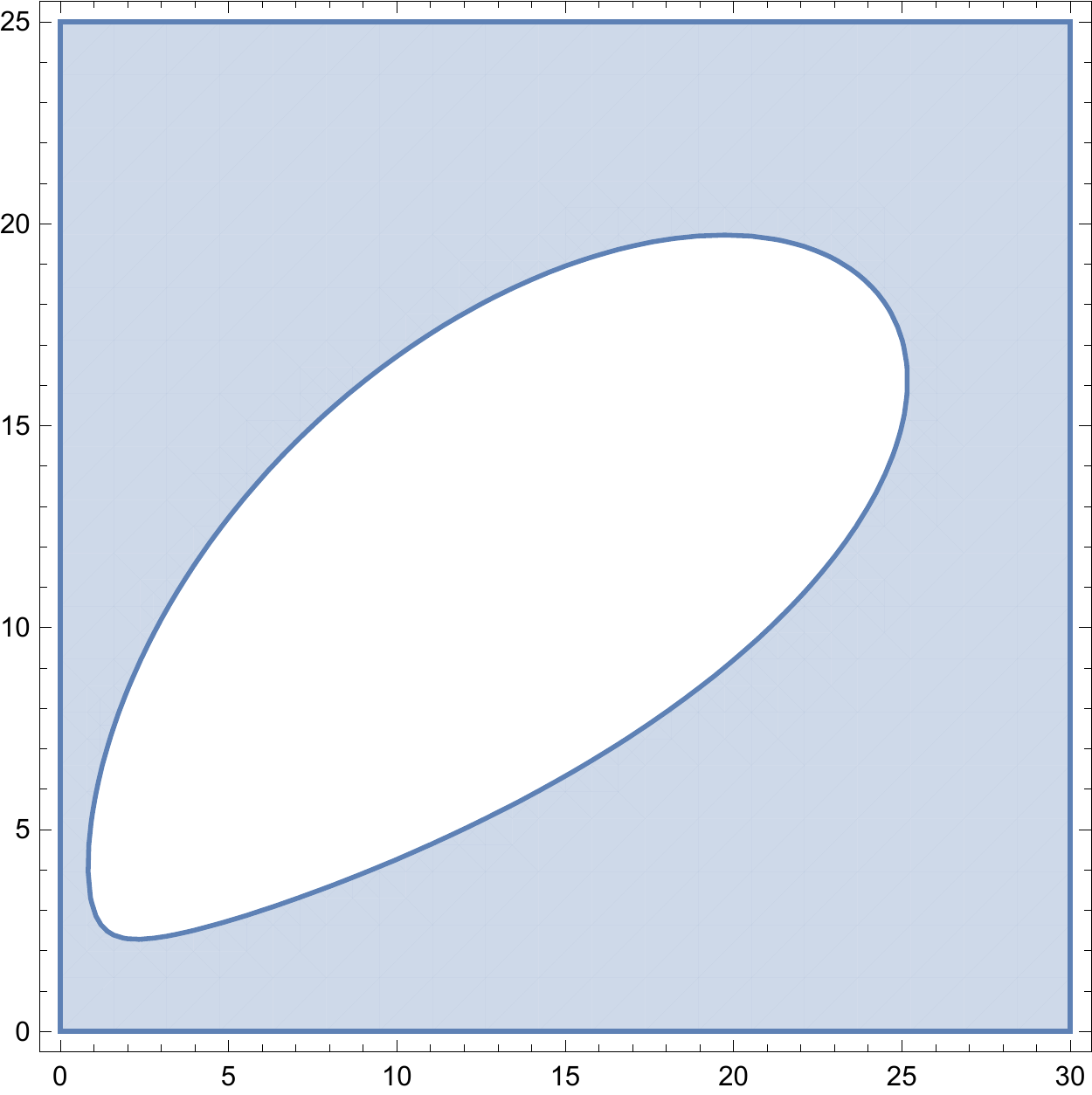}
		\caption{$d_2=9$}
	\end{subfigure}
	\caption{Easy regions for $(d_0,d_1)$ for variety of settings of $d_2$.}
	\label{fig:partially-rev-4-com}
\end{figure}
\end{example}
Unfortunately, we are still far from establishing the veracity of these heuristic predictions.  For most of these problems, BP has not been proven to succeed in the blue region of parameters, nor is any other polynomial time algorithm known.  %
There is no roadmap to establishing intractability of these problems when the parameters are chosen in the white region.

Our main result takes a step towards establishing these predictions by giving a spectral algorithm to partially recover the hidden variables whenever the parameters are in the blue region.
Specifically, we devise a spectral algorithm that uses a linearization of BP, an approach that has been succesfully carried out for the case of community detection in \cite{krzakala2013spectral,BLM15}.
%

%
%
%
%

\subsection{Stable Fixed Point Barrier}
Belief propogation (BP) aims to estimate the marginals of the hidden variables, in our case $\bc(v)$ for $v \in [n]$.
To visualize BP, it will be useful to consider the bipartite graph $\mathcal{H}$ with variables $[n]$ on one side and the factors (a.k.a. observations) $\bE$ on the other. There is an edge between a variable $v$ and an observation $e \in \bE$ if $v \in e$. 
The execution of BP is divided into rounds where in each round, the variable nodes send messages to factor nodes or vice versa.

Let $m^{v \to e}$ denote the message sent by a variable $v$ to a factor $e \in \bE$ and let $m^{e \to v}$ denote the message from a factor $e \in \bE$ to a variable $v$.
All messages exchanged are distributions over the domain $[\numcols]$, i.e., $m^{v \to e} = (m^{v \to e}_1,\ldots, m^{v \to e}_{\numcols})$ and similarly $m^{e \to v} = (m^{e \to v}_1,\ldots, m^{e \to v}_{\numcols})$.
Intuitively speaking, $m^{v\to e}_c$ is an estimate of the marginal probability that $v$ is assigned the color $c$ when the factor $e$ is absent, and $m^{e\to u}_c$ is an estimate of the marginal probability that $u$ has color $c$ when all other factors involving $u$ are absent.

The general schema of a BP algorithm is to start BP with some intialization of the messages
\[
	\{ m^{v \to e}[0], m^{e \to v}[0]\}_{v \in [n],e \in \bE}
\]
and iteratively update the messages as specified by the functions $\Upsilon$, until the messages stabilize into a  fixed point, i.e., a set of messages $\{ \hat{m}^{v \to e}, \hat{m}^{e \to v}\}$ so that,
\begin{align*}
\hat{m}^{v \to e} & = \Upsilon_{v \to e}\left( \{ \hat{m}^{f \to v} \mid f \in \partial v \backslash e \} \right) \\
\hat{m}^{e \to v} & = \Upsilon_{e \to v}\left( \{ \hat{m}^{u \to e}  \mid u \in \partial e \backslash v \} \right)
\end{align*}

There is a canonical starting point $\minit$ for the BP iterations where the messages $m^{e \to v}$ correspond to uniform distribution over the possible values $[q]$.  
%
Conjecturally, this canonical initialization $\minit$ plays a critical role in characterizing the computational complexity of inferring the hidden variables in model $\Model$.
There appear to be three possible cases with regards to this canonical initialization.

\paragraph{Case 1: $\minit$ is not a fixed point} Suppose $\minit$ is not a fixed point for the BP iteration over the model $\Model$, then BP iteration can be expected to make progress, thereby yielding a weak recovery of hidden variables.

In fact, we will present a self-contained algorithm that weakly-recovers the hidden coloring in this case. Formally, we will show the following in \pref{app:easy-cases}:
\begin{lemma}\label{lem:easy-case}
	If $\minit$ is not a fixed point for the BP iteration on model $\Model$, then there is a polynomial time algorithm $\calA$ and an $\epsilon > 0$ such that
	\begin{enumerate}
		\item if $(\bE,\btau)\sim\Model$: $\calA$ outputs a coloring that beats the correlation random guessing achieves with the hidden coloring by $\eps$,
		\item $\calA$ solves the $\Model$ vs.\ $\Model^{\times}$ (the null distribution) distinguishing problem with high probability.
	\end{enumerate}
\end{lemma}
In light of the above lemma, it is natural to restrict our attention to the case where $\minit$ is a fixed point for the BP iteration.  

\paragraph{Case 2: $\minit$ is an unstable fixed point}

$\minit$ is an {\it unstable fixed point} if arbitrarily small perturbations of $\minit$ will lead to the BP iteration moving away from the fixed point $\minit$.
%
This case was marked by the blue region in \pref{fig:NAE3SAT} and \pref{fig:partially-rev-4-com}.  In this case, our main algorithmic result is a spectral algorithm to recover a coloring $\bc'$ that beats the correlation random guessing achieves with the hidden coloring.  Alternatively, the spectral algorithm can be used to distinguish between the planted and the null distributions.

\paragraph{Case 3: $\minit$ is a stable fixed point}

$\minit$ is a {\it stable fixed point} if there exists a neighborhood $U$ around $\minit$ such that for any initialization $\hat{m} \in U$, BP iteration converges to the canonical fixed point $\minit$.
In this case, the canonical fixed point $\minit$ clearly highlights a potential failure of BP algorithm.  The hypothesis of Krzakala and Zdeborova \cite{KrzakalaZ09} asserts that existence of this stable fixed point marks the onset of computational intractability in general. 

\subsection{Related Work}

Ideas from statistical physics have long been brought to bear on inference problems.  
We refer the reader to \cite{nishimori2001statistical,mezard2009information,zdeborova2016statistical,ricci2019typology} for an introduction to the phase transitions that mark changes in statistical and computational properties of these problems.

\paragraph{Planted models}

Special cases of the planted model we consider have appeared extensively in literature.  
The conditional probability of the hidden vector given the noisy observations takes the form of
a graphical model, i.e. factorizes according to an hypergraph whose nodes correspond to variables
and hyperedges correspond to noisy observations. Such graphical models have been studied by
many authors in machine learning \cite{lafferty2001conditional} under the name of {\it conditional random fields}.
We highlight a few among the extensive body of literature on information-theoretic and structural properties of these planted models.  Montanari \cite{andrea2008estimating} characterized the posterior marginals in terms of fixed points of the associated density evolution operator.  Subsequently, Abbe and Montanari \cite{abbe2013conditional} show concentration for the conditional entropy per hidden variable given the observations.
More recently, Coja-Oghlan et. al. \cite{coja2020inference} study the information theoretic limits to recovery and confirm a conjectured formula for the mutual information between the observations and the planted assignment.

\paragraph{Spectral algorithms via non-backtracking operator}
The idea of using the spectra of non-backtracking matrix for recovery in planted problems can be traced back to the seminal work of Krzakala et al. \cite{krzakala2013spectral} in the context of community detection.  While this work provided heuristic arguments supporting the correctness of the algorithm, it was rigorously established in the work of Bordenave et. al. \cite{BLM15}.
Subsequently, \cite{saade2015spectral} devised spectral algorithms for solving the recovery problem in the censored block model, a variant of community detection wherein the edges are weighted and the weights carry the information about the community labels, but the edges don't.  Building on the result of \cite{BLM15}, this work shows that the eigenvectors of non-backtracking matrix can be used to partially recover the communities, right up to the threshold.
Finally, Angelini et al. \cite{angelini2015spectral} consider a model of sparse hypergraphs that includes planted CSPs as a special case.  The paper proposes a spectral algorithm based on a generalization of a non-backtracking matrix to hypergraphs, and gives a heuristic argument that the algorithm solves detection whenever belief propogation succeeds.  Unlike our work, the algorithm proposed in \cite{angelini2015spectral} uses an unweighted non-backtracking matrix that is independent of the prior probabilities.  While it is a desirable feature that the algorithm is {\it non-parametric}, i.e., does not rely on the knowledge of prior distributions generating the instance, it is unclear if such a non-parametric algorithm can achieve detection up to the threshold in general.

Apart from recovery in planted models, the non-backtracking operator and the closely related Bethe-Hessian matrix have also been applied towards computing upper bounds for the log-partition function in ferro-magnetic Ising models on general graphs \cite{saade2017spectral}.

\paragraph{Quiet Planting}

Planted distributions that are indistinguishable from their random counterparts are often referred to as "quiet planting", though the terminology is not often consistent on whether the distributions are computationally or statistically indistinguishable.   

A quiet planting that is statistically indistinguishable from random was used as a technical tool to study random instances in \cite{AchlioptasCO08}.  Krzakala and Zdeborova \cite{KrzakalaZ09} studied the existence of quiet plantings for graph coloring problem and were the first to hypothesize that under the Kesten-Stigum threshold, the planted ensembles are a computationally indistinguishable from random.
Subsequently, the authors \cite{zdeborova2011quiet} considered planted distributions for {\it locked CSPs}, wherein every pair of assignments to a predicate have Hamming distance at least $2$ and showed that problem is easy above a threshold that coincides with the Kesten-Stigum threshold and hypothesize that non-trivial recovery is computationally hard under it.
Finally, a statistically quiet planting for the random $k$-SAT problem has been proposed in \cite{krzakala2012reweighted}.

\paragraph{Community Detection}

Extensive work on community detection for stochastic block models has led to the confirmation \cite{mossel2018proof,M14,BLM15,AS15} of conjectures of Decelle et. al. \cite{decelle2011inference,DKMZ}. 
As predicted, existing algorithms \cite{M14,BLM15,AS15} can partially recover community labels up to the Kesten-Stigum threshold, but no lower.
For $q = 2$ communities, the Kesten-Stigum threshold also matches the information theoretic threshold beyond which recovery is impossible.
However, for $q \geq 3$ communities, the problem is believed to exhibit a statistical-vs-computational gap, in that there is a range of parameters where partial recovery is possible but is computationally intractable. 
The presence of a gap between the Kesten-Stigum threshold and the information theoretic threshold for all $q > 5$ was established in \cite{sly2009reconstruction}. 
More recently, Ricci-Tersenghi et al.\ \cite{ricci2019typology} predicted the existence of such a gap for $q = 4$ communities for some degree distributions, and also identifies a threshold beyond which there is a hard phase in asymmetric SBM.   
Furthermore, this work predicts the existence of hybrid-hard phases where it is computationally easy to reach a non-trivial inference accuracy, but computationally hard to match the information theoretically optimal one.  Specifically, there are stable fixed points for BP that are not the trivial fixed point, but also don't correspond to optimal recovery.

\paragraph{Spectral norm bounds}
Technically, our work draws on ideas from Bordenave, Lelarge and Massoulie \cite{BLM15} who established spectral norm bounds for non-backtracking matrices associated with Erd\H{o}s-Renyi random graphs.
Closer to our own setup, Stephan et al.\ \cite{stephan2020non} show eigenvalue bounds for the non-backtracking matrices of random graphs that have independent and bounded edge weights, and bounded model complexity (measured by the rank of the expected adjacency matrix). However, in our model the edges have correlated matrix weights instead of independent scalar weights, so their eigenvalue bounds do not generally apply to our model.
Another work we draw several ideas from is that of Bordenave and Collins \cite{BC19}, who prove that the spectra of a wide family of random graphs, namely those arising from matrix-weighted noncommutative polynomials of random permutation matrices (see \cite{OW20} for a comprehensive characterization and examples in this family), are roughly contained within the spectrum of an appropriately defined infinite graph.  The key techniques useful in our work are the ones they employ to bound the spectral norms of the non-backtracking matrices of random regular graphs whose the edges are endowed with varying matrix weights.

\subsection{Technical Overview}

We define a general model for sparse observations on a hidden vector, and carry out the cavity method calculations in full generality following \cite{DKMZ}. We obtain a criterion for computational tractability of the recovery and detection problems on this model, and provide spectral algorithms for recovery (\pref{thm:informal-weak-recovery}) and detection (\pref{thm:informal-distinguishing}) in the tractable regime.  The key technical ingredient in our work is tight eigenvalue bounds for nonbacktracking matrices of sparse random hypergraphs with (possibly varying) matrix-valued edge weights (\pref{thm:main-mat-conc-informal}).

In this section, we will attempt a brief technical outline of our result specialized to the case of distinguishing a random NAE$3$SAT instance from one with a hidden satisfying assignment.  Concretely, consider the problem distinguishing $\IntroNull$ from $\IntroPlanted$ where:\footnote{Strictly speaking, this model and the distribution over NAE$3$SAT instances our generic model yields differ slightly.  Nevertheless they are contiguous and so the phenomena in one carry to the other.}
\begin{itemize}
\item An instance $\bcalI\sim\IntroNull$ is obtained by sampling each triple of distinct vertices $(u,v,w)$ in $[n]^3$ independently with probability $\frac{d}{3!n^2}$ and then placing uniformly random negations $(\bsigma_u,\bsigma_v,\bsigma_w)$ on each variable.
\item An instance $\bcalI\sim\IntroPlanted$ is sampled in a two-stage process: (1) sample a hidden assignment $\bx\sim\{\pm1\}^n$, (2) sample each triple of distinct vertices $(u,v,w)$ in $[n]^3$ independently with probability $\frac{d}{3!n^2}$ and place uniformly random negations $(\bsigma_u,\bsigma_v,\bsigma_w)$ \emph{conditioned} on $\mathrm{NotAllEquals}(\bsigma_u \bx_u, \bsigma_v \bx_v, \bsigma_w \bx_w)=1$.
\end{itemize}

First, let us map out the statistical physics prediction of the smallest value of $d$ at which the problem becomes computationally tractable.
In particular, we need to work out the value of $d$ for which the trivial fixed point for belief propogation is unstable.  
To this end, one emulates the cavity method heuristic calculations analogous to the one carried out in \cite{DKMZ} for stochastic block models.
Oversimplifying for the sake of presentation, the cavity method heuristic amounts to carrying out the calculation by treating the neighborhood of each variable to be an infinite tree (see \pref{sec:cavity-method} for more details).

Concretely, the setup in the cavity method calculation is as follows.  The neighborhood of a variable $v$ in the NAE$3$SAT instance is modelled as an infinite tree with alternating layers of variable and NAE$3$SAT constraint nodes.  The tree is generated by a Galton-Watson process where each variable $v$ picks a degree $d_w \sim \mathrm{Poisson}(d)$ from the Poisson distribution, and has $d_w$ NAE$3$SAT constraint nodes as children, and each constraint node has exactly $2$ children.  For each path of length $2$, $u \rightarrow C \rightarrow w$ from a variable $u$ to its constraint node $C$ followed by another variable $w$ in the constraint, there is an associated constant sized matrix $M_{u C w}$ depending on the prior distribution.  For the case of NAE$3$SAT, all of the matrices $M_{u C w}$ are given by $M_{u C w} = \sigma_u \sigma_w M$ where
\[
	M \coloneqq
	\begin{bmatrix}
		-1/6 & 1/6 \\
		1/6 & -1/6
	\end{bmatrix}.
\]
For any depth $t$, consider the following quantity $\rho_t$ where the expectation is over the choice of the infinite tree $\mathcal{T}$,
\[ \rho_t(d) = \E_{\textrm{tree } \mathcal{T}} \left[ \sum_{\textrm{paths} u_0 = v \to C_0 \to u_1 \to C_1 \to u_2 \to \cdots \to u_t} \mathrm{Tr}\left( \left(\prod_{i=0}^{t-1} M_{u_i C_i u_{i+1}}\right) \left(\prod_{i=0}^{t-1} M_{u_i C_i u_{i+1}}\right)^* \right) \right]  \]
The threshold $d^*$ predicted by the cavity method is precisely the smallest value of $d$ for which $\lim_{t \to \infty} \rho_t(d) > 1$.

This characterization of $d^*$ is a little unwieldy in that it is not immediate that the value of the threshold $d^*$ is decidable.  Fortunately, through ideas from the work of Bordenave and Collins \cite{BC19}, the above characterization can be equivalently written in terms of the spectral radius of an associated finite matrix. 
Specifically, for NAE$3$SAT, $d^*$ is the smallest $d$ for which the spectral radius of $L$ exceeds $1$ where:
\[
	L =
	d\cdot\begin{bmatrix}
		1/18 & -1/18 & -1/18 & 1/18 \\
		-1/18 & 1/18 & 1/18 & -1/18 \\
		-1/18 & 1/18 & 1/18 & -1/18 \\
		1/18 & -1/18 & -1/18 & 1/18 \\
	\end{bmatrix}.
\]
(see \pref{sec:from-alpha-to-lambda} for an overview of how to construct $L$ in general, and \pref{sec:cavity-method} for details).
Hence, for NAE$3$SAT, the problem is hypothesized to become algorithmically tractable once $d>4.5$.  Our main results are algorithms for distinguishing the null and planted distributions, and for partially recovering a hidden assignment in the general model we consider whenever the spectral radius $\rho(L)$ of the matrix $L$ corresponding to the model exceeds $1$.  In the case of NAE3SAT, we prove:
\begin{theorem}
	When $d > 4.5$, given $\bcalI\sim\IntroNull$ or $\IntroPlanted$:
	\begin{enumerate}
		\item There is an efficient algorithm to distinguish $\IntroNull$ from $\IntroPlanted$ with probability $1-o(1)$.
		\item There is an efficient algorithm to produce $\Theta(1)$ unit vectors $V$ where $\langle v,x\rangle\ge\Omega(\sqrt{n})$ for some $v\in V$.
	\end{enumerate}
\end{theorem}


%

We now describe the distinguishing algorithm, which is spectral in nature, and briefly survey the techniques to analyze the matrix involved.  The matrix we employ is a power of the so-called non-backtracking matrix obtained by linearizing belief propogation.  
For each clause $C$ and pair of variables $u,v$ in the clause signed by $\sigma_u,\sigma_v$ we define matrix $M_{uCv}\coloneqq \sigma_u \sigma_v M$.  The $s$-th nonbacktracking power matrix is a $n\times n$ block matrix where each block is $2\times 2$:
\[
	A^{(s)}[a,b] \coloneqq \sum_{uC_1u_1C_2u_2\dots u_{s-1}C_s v} M_{uC_1u_1} M_{u_1C_2u_2}\cdots M_{u_{s-1}C_s v}.
\]
The algorithm is then fairly simple:
\begin{itemize}
	\item Let $s = \floor{\sqrt{\log n}}$, and let $\kappa$ be strictly between $\sqrt{\rho(L)}$ and $\rho(L)$.
	\item If $\|A^{(s)}\| < \kappa^s$ output $\IntroNull$, otherwise output $\IntroPlanted$.
\end{itemize}
In order to prove that the algorithm is correct, there are two key technical steps: (1) to prove that in the planted model, the operator norm is large, (2) to prove that in the null model, the operator norm is bounded.

The key insight in proving (1) is that the large operator norm of $A^{(s)}$ arises from the hidden assignment to the planted instance $\bcalI$ itself.  In particular, denoting $\by\coloneqq \bx\otimes\begin{bmatrix}1 \\ -1\end{bmatrix}$  we prove:
\begin{lemma} \label{lem:corr}
	With probability $1-o(1)$:
	\[
		\frac{\langle \by, A^{(s)}\by\rangle}{\|\by\|^2} \ge \Omega(\rho(L)^s).
	\]
\end{lemma}
This is proved in full generality in \pref{sec:statistics}.

The main technical difficulty is in proving (2) in the general model, as is done in \pref{sec:eig-bounds}.  We prove:
\begin{lemma}
	With probability $1-o(1)$:
	\[
		\|A^{(s)}\| \le \left((1+o(1))\sqrt{\rho(L)}\right)^s.
	\]
\end{lemma}
Our proof is largely inspired by the works of \cite{BLM15,BC19,stephan2020non}.  On one hand, \cite{BLM15} and \cite{stephan2020non} show tight eigenvalue bounds for the nonbacktracking matrices of sparse (possibly inhomogenous) \Erdos-\Renyi~graphs with scalar edge weights. The proof exploits the commutativity of scalar products, i.e. the product of edge weights along a walk is invariant under reordering.  However, the graphs we consider have matrix-valued weights, which in general don't commute under multiplication.  Therefore the product of edge weights changes depending on the order of multiplication.  On the other hand, random regular graphs with matrix-valued edge weights is handled in the work of \cite{BC19}.  However, the proof in \cite{BC19} heavily exploits the regularity of the model --- each vertex has exactly $d$ adjacent edges and these edges have the exact same set of matrix weights. This leads to every vertex having isomorphic neighborhoods, and simplifies the analysis, which does not occur in our setting due to the lack of regularity.   Our situation is further complicated by the fact that due to hyperedges of size greater than $2$ even the random matrix weights in different blocks are not independent, which introduces mild correlations.  The proof follows the general framework of the trace method and gives a more fine-grained analysis for nonbacktracking walks based on their shapes.

While the spectral radius of non-backtracking powers $A^{(s)}$ serve as a distinguisher, recovering the hidden assignment from the eigenvectors is little more subtle.  In particular, this requires proving a converse of \pref{lem:corr} that the every vector $v$ for which $\langle v, A^{(s)} v\rangle$ is large, is actually correlated with the planted assignment.
Instead, we bypass this issue by collating information from eigenvectors of $A^{(s)}$ for a range of values of $s$ 
(see \pref{sec:recovery} for details).

\subsection{Discussion and Future Work}
In this work, we have shown that for a very general class of planted problems, the problem is computationally tractable whenever the trivial fixed point is unstable.  
This establishes the algorithmic side of the predictions of Krzakala and Zdeborova \cite{KrzakalaZ09} for all these problems.
Several compelling open questions remain, we list a few here.

\paragraph{Reductions.}
From the standpoint of average case complexity, the main open question is to establish or refute the stable fixed point barrier.
Given that all Bayesian CSPs have a uniform onset of intractability as specified by the stable fixed point barrier, perhaps these problems are reducible to one another.
Traditional CSPs are very amenable to reductions, it is compelling to see if there are reductions between Bayesian CSPs, and stable fixed point barrier can be obtained as a consequence of the intractability of a single Bayesian CSP.  The main challenge here is in coming up with reductions between problems that are \emph{distribution-preserving} and we speculate that the ideas in \cite{BBH18,BB20}, which are examples of recent successes in reductions between average case problems, might be useful.

\paragraph{Hardness evidence in restricted computational models.}
Evidence on the stable fixed point barrier would also be very interesting.  \cite{hopkins2017efficient} showed that an algorithm based on low-degree polynomials solves the distinguishing problem in community detection up to Kesten-Stigum threshold, and also proves matching hardness in that low-degree polynomials fail to solve the problem under the Kesten-Stigum threshold.  Recent work introduced the local statistics SDP hierarchy \cite{BMR19} and showed the same algorithmic result for this class of algorithms and proved a negative result for the degree-$2$ SOS version of this algorithm. 
It will be useful to show that low-degree method and local statistics SDP hierarchy fail to solve the detection problem in the general model we consider in the presence of a stable fixed point.
It will also be interesting to see if conditional hardness results for the problem can be obtained in other models such as statistical query algorithms \cite{FGRVY17}.

Another direction in the spirit of the recent work of \cite{BBHLS20} which establishes an equivalence between the predictions of statistical query algorithms and the low-degree polynomials method would be to formally establish the equivalence of the predictions for the stable fixed point based on the cavity method with the other restricted models of computation such as the ones mentioned above.

\paragraph{Goldreich's PRG for 1-wise independent, balanced, local predicates}
Goldreich proposed a construction of pseudorandom generators from random CSPs with balanced local predicates \cite{G00}. The generator mapping $n$ variables to $\{0,1\}^m$ is constructed as follows: let $E_1,\dots,E_m$ be a randomly chosen set of constraints on $n$ input variables, then the $i$-th bit of the output string indicates whether $E_i$ is satisfied by the input or not. 

The constraints in Goldreich's generator can be sampled from the null distribution of the model that we study. Then on any input $\bc$ (analogous to the hidden variables in the model), the output of the generator together with the constraints can be viewed as observations from the model's planted distribution. Roughly speaking, we say that this generator produces pseudorandom strings if and only if the detection problem for this model is intractable. 

For any random CSP with $1$-wise independent predicates, the cavity method yields a concrete predicate density threshold above which the detection problem should be tractable (indeed the threshold is always of order $O(n)$). Our distinguishing algorithm confirms tractability in this regime, and therefore provides a concrete linear upper bound on the stretch of the Goldreich's PRG constructed from the random CSP. Indeed, the upper bounds would be tight if the stable fixed point barrier hypothesis holds.

\paragraph{$\NP$ problem}
If the stable fixed point barrier hypothesis holds, then these Bayesian CSPs are excellent examples of average-case hard problems that are easy to sample. Their intractability can be harnessed to build cryptographic and pseudorandom primitives whose security depends on the existance of average-case hard problems.
To this end, it is important that the underlying intractable problem is in $\NP$, i.e., given the true hidden assignment, an efficient algorithm must be able to recognize it.
Formally, this motivates the following $\NP$-version of the problem:
\begin{problem} ($\NP$ version)
	Devise an efficient verification algorithm $\calA$ that,
	given observations $\bE$ from the planted model $\Model$ and a candidate assignment $\bc: [n] \to [q]$, has the following property:
	\begin{itemize}
		\item If $(\bc, \bE)$ are generated from the model $\Model$, the algorithm $\calA$ accepts $(\bc, \bE)$ with high probability.
		\item If the observations $\bE$ are generated from the null model $\Model^{\null}$, then for every assignment $\bc :[n] \to [q]$, the algorithm rejects $(\bc,\bE)$ with high probability.
\end{itemize}
\end{problem}

\paragraph{Dense models}
The focus of this paper has been the sparse settings, where in the underlying variable-observation graph is constant degree on average.
Stability of trivial fixed point is also hypothesized to indicate computational intractability in dense problems such as spiked Wigner matrix (see \cite{MV17} for some rigorous results).
In this setting, it is the stability of fixed points of the \emph{approximate message passing} (AMP) algorithm.  
A natural open question is whether the spectral algorithm based on linearizing AMP can be shown to generically hold in the dense setting. 

\paragraph{Optimal Recovery}
Finally, in the region where weak recovery is possible, BP is conjectured to achieve the optimal recovery rate, i.e., achieve the maximum possible correlation with the hidden communities.  While spectral algorithms provably achieve weak recovery, there has only been partial progress on the problem of achieving the optimal recovery rate \cite{MNS14opt} --- in particular, optimal recovery even in the $2$-community block model close to the Kesten-Stigum threshold is open.
In analogy with traditional CSPs, stable fixed point barrier marks the onset of {\it "approximation resistance"} for some problems, while the recovery rate corresponds to the approximation ratio.

\newcommand{\dom}{[q]}

\section{Preliminaries}

\subsection{Observation Model}	\label{sec:obs-model}

We will now formally define the {\it observation model} that is used throughout this work.  
The basic setup consists of a set of hidden variables $\bc(1),\ldots,\bc(n)$ taking values over a finite domain $\dom = \{1,\ldots, q\}$.   Borrowing terminology from the {\sc Planted Coloring} problem, we will refer to $[n] = \{1,\ldots,n\}$ as the set of variables, $\dom$ as the set of colors and $\bc: [n] \to [q]$ as the {\it hidden coloring}.

The hidden coloring $\bc: [n] \to [q]$ is drawn from a prior distribution $\prior_{\bc}$.  A sequence of hyperedges $\bE$ on the vertex set $[n]$ are drawn, and we will refer to these hyperedges as {\it observations}.  More precisely, an observation is effectively a hyperedge $e \in \bE$ with a type $\btau(e)$.  

\begin{definition} \label{def:model}
An observation model $\Model=([\numcols], T,\TypeDist,\{\ColDist_{\tau}\}_{\tau\in T},\Phi)$ describes a distribution on $n$-vertex hypergraphs $\Model_n$ for every $n\ge 1$ and is specified by,

\begin{itemize}
	\item {\bf (Variable Types $T$).} A set of types $T$ for the hidden variables and a distribution $\TypeDist$ over them.

		Each variable is assigned a random type sampled from $\TypeDist$ and is described by $\btau : [n] \to T$; in aggregate there are $\approx\TypeDist(\tau) \cdot n$ variables of type $\tau$. 

	\item {\bf (Prior Distributions $\{\ColDist_\tau\}_{\tau \in T}$).} For each variable of type $\tau \in T$, a prior distribution $\ColDist_\tau$.

		The prior distribution of hidden coloring $\bc: [n] \to \dom$ is the product distribution,

			\[\prior_{\bc} = \ColDist_{\btau(1)} \times \ColDist_{\btau(2)} \ldots \times \ColDist_{\btau(n)} \]

	\item {\bf (Observation Types $\Phi$).}  Set of observation types $\Phi = \{\phi_1,\ldots, \phi_F\}$.  The arity of a type $i$ observation is denoted by $\arity(i)$.

		Each observation on the variables is a hyperedge with a type from $\Phi$.  Specifically, the set of all observations is a set of hyperedges $\bE$ partitioned as $\bE = \cup_{i \in [F]} \bE_i$ where $\bE_i$ is a set of $\arity(i)$-tuples of distinct elements in $[n]$. 

	\item {\bf (Observation Distributions).}  For each observation type $\phi_i \in \Phi$, we have a bounded function $\phi_i: T^{\arity(i)} \times [q]^{\arity(i)} \to \R^+$.

	For every $\arity(i)$-tuple $(v_1,\ldots,v_{\arity(i)})$ of distinct elements in $[n]$, the observation $\phi_i(v_1,\ldots,v_{\arity(i)})$ is included independently with probability 

		\[ \Pr\left[ \phi_i(v_1,\ldots,v_{\arity(i)}) \in \bE_i \right] \defeq \frac{\phi_i\left( (\btau(v_1),\bc(v_1)),\ldots,(\btau(v_{\arity(i)}), \bc(v_{\arity(i)}))\right)}{n^{\arity(i) - 1}}\]

	Notice that the probability of drawing an observation $\phi_i (v_{1},\ldots,v_{\arity(i)})$ depends both on the types of the variables and their colors.
\end{itemize}
\end{definition}
We refer the reader to the work of Montanari \cite{andrea2008estimating}, where this model has been previously used for a wealth of concrete examples captured in this framework.  Here we will exhibit a few examples.

\begin{example}
(Stochastic Block Model in semi-supervised setting)

In this variant of community detection, a graph $G= (V,E)$ is drawn from a $[q]$-community SBM and in addition an $\alpha$-fraction of the vertex labels are revealed.  \cite{ZMZ14} study the KS threshold in this model using the cavity method.

To encode this problem into our framework, we will have vertex types $T = [q] \cup \{ \perp \} $ wherein the type of a vertex $v$ is $\btau(v) \in [q]$ if the label of $v$ is revealed, and $\btau(v) = \perp$ if it is unrevealed.

We have a single observation type namely the edges of the SBM, and the probability of an edge $(u,v)$ is clearly $\phi\left((\btau(u),\bc(u)),(\btau(v), \bc(v))\right)/n$ for a function $\phi$ depending on the types and colors of two vertices.

\prasadtemp{what more to add here?}
\end{example}

More generally, the model can encode variants of SBM wherein there is additional attributes revealed about the vertices or edges or both.  For example, SBM with labelled edges \cite{HLM12} are subsumed by different types of observations, while SBM with vertex features \cite{DSMM18} are captured by vertex types.  

Further, the model can also be used to express geometric SBM \cite{SASB18} in restricted cases.  In a geometric SBM, the vertices are distributed on a compact metric space like the sphere, and the probability of including an edge between vertices $u,v$ is a function of the distance between the two.  If the metric space is compact, say a sphere in a constant dimensional space,  then one can use an $\epsilon$-net of the compact set as a finite set of vertex types to model the SBM in our framework.

\prasadtemp{ add another example?}

\subsubsection{Miscellaneous simplifying notation}\label{sec:notations}
\paragraph{Class function $\Cl$:}
For notational convenience, we will make a modification to our \pref{def:model} that does not affect the generality of our results.
We will enforce that each observation type $\phi_i$ have a fixed tuple of variable types on which it applies.  Formally, each observation type $\phi_i$ has an associated {\it class} type $\Cl(i)  \in T^{\arity(i)}$ such that all occurrences of the observation $\phi_i$ have input variable types given by $\Cl(i)$.
It is clear that this restriction is a special case of \pref{def:model} with the additional restriction that,
\[ \phi_i\left( (\tau_1,c_1),\ldots, (\tau_{\arity(i)}, c_{\arity(i)})\right) = 0 \text{ if } (\tau_1,\ldots, \tau_{\arity(i)}) \neq \Cl(i)\]

Conversely, given a general model $\Model$ as per \pref{def:model}, for each observation type $\phi_i$ and each tuple $\btau = (\tau_1,\ldots \tau_{\arity(i)})$, introduce an observation type $\phi'_{i, \btau}$ that is identical to $\phi_i$, but restricted to variable types $\btau$, i.e., set $\Cl(i) = \btau$.  It is easy to see that this transformation creates a model $\Model'$ that is equivalent to $\Model$.  Without loss of generality we will henceforth use $\phi_i(c_1,\dots,c_{\arity(i)})$ to denote $\phi_i((\tau_1,c_1), \dots, (\tau_{\arity(i)}, c_i))$ where $\btau = \Cl(i)$.

\paragraph{Average factor density:} We will use $\ACFacDeg{i}$ to denote the average density of a factor:
\[
	\ACFacDeg{i} \coloneqq \sum_{(c_1,\dots,c_{\arity(i)}) \in [q]^{\arity(i)}} \left(\prod_{k=1}^{\arity(i)}\ColDist_{\Cl(i)_k}\right)\cdot \phi_i(c_1,\dots,c_{\arity(i)}). 
\]

\paragraph{Bipartite view:}  Given a collection of sampled observations $\bE = \cup_{i=1}^F \bE_i$, we associate a bipartite graph $\bG$ where the left vertex set is given by the variables $[n]$ and the right vertex set is given by the collection of all $(i,\gamma)$ for $\gamma$ in $\bE_i$.

\paragraph{Index function:} For $e = (i,(v_1,\dots,v_{\arity(i)}))$ we define $\index_e(v_s)$ as $s$ and $e[s]$ as $v_s$.  When $e$ is clear from context we will drop the $e$ and just use $\index(v_s)$.

\begin{definition}	\label{def:null-model}
	For an observation model $\Model=([q], T,\TypeDist,\{\ColDist_{\tau}\}_{\tau\in T},\Phi)$ the corresponding \emph{null model} $\Model^{\times}$ is the observation distribution where for every $\arity(i)$-tuple $(v_1,\ldots,v_{\arity(i)})$ a hidden coloring $\bc$ is sampled independently, and the observation $\phi_i(v_1,\ldots,v_{\arity(i)})$ is included with probability 
	\[
		\Pr\left[ \phi_i(v_1,\ldots,v_{\arity(i)}) \in \bE_i \right] \defeq \frac{\phi_i\left( (\btau(v_1),\bc(v_1)),\ldots,(\btau(v_{\arity}(i)), \bc(v_{\arity(i)}))\right)}{n^{\arity(i) - 1}}.
	\]
	Equivalently, in the null model for every $(v_1,\dots,v_{\arity(i)})$ the observation $\phi_i(v_1,\ldots,v_{\arity(i)})$ is included independently with probability:
	\[
		\Pr\left[ \phi_i(v_1,\ldots,v_{\arity(i)}) \in \bE_i \right] \defeq \frac{\ACFacDeg{i}}{n^{\arity(i) - 1}}.
	\]
\end{definition}

\begin{remark}
	For a model $\Model$ we will refer to it as the \emph{planted model} and we will refer to $\Null$ as the \emph{null model}.  Two computational problems we are interested in are distinguishing whether a sample is drawn from $\Model$ or $\Null$, and inferring the hidden coloring for a sample drawn from $\Model$.
\end{remark}

\subsection{Bayesian Inference}

Given the variable types $\btau : [n] \to T$ and the observations $\bE$, the canonical algorithm to infer the hidden coloring $\bc$ is to use the Bayes rule to compute the conditional distribution $\mathbb{P}_{\bc| \bE}$.
Formally, the probability that a model $\Model=([\numcols], T,\TypeDist,\{\ColDist_{\tau}\}_{\tau\in T},\Phi)$ generates a hidden coloring $\bc$ and observations $\bE$ is 
\begin{align*}
\Pr[\bE_1,\dots,\bE_F, \bc \mid \btau ] = & \Pr[\bc \mid \btau] \cdot \Pr[\bE_1,\ldots,\bE_F \mid \bc, \btau] \\
= &\left(\prod_{v\in[n]}\ColDist_{\btau(v)}(\bc(v)) \right)\cdot \\
& \prod_{i\in[F]}\left(\prod_{(v_j)_j\in [n]^{\arity(i)}} 
\left[\frac{\phi_i(\bc(v_1),\dots,\bc(v_{\arity(i)}))}{n^{\arity(i)-1}}\right]^{\mathbf{1}_{(v_j)_j\in\bE_i}}
\left[1 - \frac{\phi_i(\bc(v_1),\dots,\bc(v_{\arity(i)}))}{n^{\arity(i)-1}}\right]^{\mathbf{1}_{(v_j)_j\not\in\bE_i}}\right)\enspace.
\end{align*} 
By applying Bayes rule, 

\[\Pr[\bc \mid \bE_1,\dots,\bE_F,\btau] = \frac{\Pr[\bE_1,\dots,\bE_F,\bc \mid \btau]}{\sum_{\bc^*}\Pr[(\bE_1,\dots,\bE_F,\bc^*\mid\btau]} \enspace.\]
Ignoring the normalizing constant, we can write 
\[ \Pr[\bc \mid \bE_1,\dots,\bE_F,\btau] \propto e^{-H(c | \bE_1,\ldots, \bE_F, \btau)} \ ,\]
where
\begin{multline*}
\Ham(\bc \mid \bE_1,\dots,\bE_F,\btau) 
= - \sum_{i\in[F]}\sum_{(v_j)_j\in [n]^{\arity(i)}} \left[\mathbf{1}_{(v_j)_j\in\bE_i}\log{\left(\frac{\phi_i(\bc(v_1),\dots,\bc(v_{\arity(i)}))}{n^{\arity(i)-1}}\right)} + \mathbf{1}_{(v_j)_j\not\in\bE_i}\log{\left(1- \frac{\phi_i(\bc(v_1),\dots,\bc(v_{\arity(i)}))}{n^{\arity(i)-1}}\right)} \right] \\
 -\sum_{v\in[n]} \log{\ColDist_{\btau(v)}(\bc(v))} .
\end{multline*}
The function $\Ham(\bc| \bE_1,\ldots,\bE_F, \btau)$ is referred to as the Hamiltonian, and the distribution is the Boltzmann distribution with Hamiltonian $\Ham$ and inverse temperature $\beta = 1$.  

Since in our setting, the hypergraph is sparse, i.e., $\ACFacDeg{i} = O(1)$, the terms 
\begin{equation}\label{eq:simplify}
\log{\left(1- \frac{\phi_i(\bc(v_1),\dots,\bc(v_{\arity(i)}))}{n^{\arity(i)-1}}\right)} \approx 0 
\end{equation}
for all $(v_j)_j\not\in\bE_i$.   So the these terms can be dropped to simplify the Hamiltonian to 

\begin{equation}\label{eq:Hamiltonian}
\Ham(\bc\mid\bE_1,\dots,\bE_F,\btau) 
= - \sum_{i\in[F]}\sum_{(v_j)_j\in\bE_i} \log{\left(\phi_i(\bc(v_1),\dots,\bc(v_{\arity(i)}))\right)} -\sum_{v\in[n]} \log{\ColDist_{\btau(v)}(\bc(v))}.
\end{equation}  

The Hamiltonian $\Ham$ is a sum of {\it local} terms each depending on a constant number of variables.  The observations $\bE$ and the variables $\bc$ together form what is termed as {\it factor graphs} (see \cite{montanaribook}), where each observation is a {\it factor} of the Boltzmann distribution.  
Recall that the Boltzmann distribution is given by
\[
	\Pr[\bc \mid (\bE_1,\dots,\bE_F,\btau)] = \frac{e^{-\Ham(\bc \mid \bE_1,\dots,\bE_F,\btau)}}{\sum_{\bc^*}e^{-\Ham(\bc^* \mid (\bE_1,\dots,\bE_F,\btau)}}
\]
The normalization term in the denominator is called the partition function of the distribution and is denoted $Z(\Model)$.
Notice that a naive algorithm to infer the hidden coloring via the Bayes rule as described above would take exponential time.

\subsection{Belief Propogation}
The algorithm of choice to infer the hidden variables in a sparse factor model would be belief propogation.
We refer the reader to \cite{montanaribook} for a detailed exposition of belief propogation, and restrict ourselves to a broad outline.

Belief propogation (BP) aims to estimate the marginals of the hidden variables, in our case $\bc(v)$ for $v \in [n]$.
BP draws its inspiration from a dynamic programming algorithm to compute the marginals when the underlying factor graph is a tree, and is broadly applicable to sparse settings where the local neighborhood of a vertex is tree-like.
In particular, while BP computes the marginals exactly on a tree, it is very succesful in practice over sparse factor models that are locally tree-like.

To visualize BP, it will be useful to consider the bipartite graph $\mathcal{H}$ with variables $[n]$ on one side and the factors (a.k.a. observations) $\bE$ on the other. There is an edge between a variable $v$ and an observation $e \in \bE$ if $v \in e$. 
The execution of BP is divided into rounds where in each round, the variable nodes send messages to factor nodes or vice versa.

Let $m^{v \to e}$ denote the message sent by a variable $v$ to a factor $e \in \bE$ and let $m^{e \to v}$ denote the message from a factor $e \in \bE$ to a variable $v$.
All messages exchanged are marginal distributions over the domain $[\numcols]$, i.e., $m^{v \to e} = (m^{v \to e}_1,\ldots, m^{v \to e}_{\numcols})$ and similarly $m^{e \to v} = (m^{e \to v}_1,\ldots, m^{e \to v}_{\numcols})$.
Intuitively speaking, $m^{v\to e}_c$ is an estimate of the marginal probability that $v$ is assigned the color $c$ when the factor $e$ is absent, and $m^{e\to u}_c$ is an estimate of the marginal probability that $u$ has color $c$ when all other factors involving $u$ are absent.

\prasadtemp{unnecessary?}
BP specifies an update rule for every variable/factor node to update its outgoing messages each round, depending on its incoming messages.
Let $\partial e$ denotes the set of variables incident a factor $e$ and let $\partial v$ denote the set of factors incident on a variable $v$.
BP specifies functions $\Upsilon_{v \to e}$, $\Upsilon_{e \to v}$ so that if $\{m^{v \to e}[t], m^{e \to v}[t]$ denote the messages in round $t$, then the updated messages are given by
\begin{align}
	m^{v \to e}[t+1] & = \Upsilon_{v \to e}\left( \{ m^{f \to v}[t] \mid f \in \partial v \backslash e \} \right)  \label{eq:bpupdate1}\\
	m^{e \to v}[t+1] & = \Upsilon_{e \to v}\left( \{ m^{u \to e} [t] \mid u \in \partial e \backslash v \} \right) \label{eq:bpupdate2}
\end{align}

We will describe the specific form of the functions $\Upsilon$ in \pref{app:BP-update-rule}, but there are two salient details that we would like to highlight at this time.
First, the functions $\Upsilon$ are smooth rational functions that map marginals over $[\numcols]$ to a marginal distribution over $[\numcols]$.
Second, the updated outgoing message $m^{v \to e}[t+1]$ depends on all messages incoming to variable $v$ {\bf except} the message $m^{e \to v}[t]$.  Similarly, the updated outgoing message $m^{e \to v}[t+1]$ is independent of the incoming message $m^{v \to e}[t]$.

The general schema of a BP algorithm is to start BP with some intialization of the messages
\[
	\{ m^{v \to e}[0], m^{e \to v}[0]\}_{v \in [n],e \in \bE}
\]
and iteratively update the messages as specified by the functions $\Upsilon$, until the messages stabilize into a  fixed point, i.e., a set of messages $\{ \hat{m}^{v \to e}, \hat{m}^{e \to v}\}$ so that,
\begin{align*}
\hat{m}^{v \to e} & = \Upsilon_{v \to e}\left( \{ \hat{m}^{f \to v} \mid f \in \partial v \backslash e \} \right) \\
\hat{m}^{e \to v} & = \Upsilon_{e \to v}\left( \{ \hat{m}^{u \to e}  \mid u \in \partial e \backslash v \} \right)
\end{align*}

While it can often be difficult at times to show convergence to a fixed point, BP is very succesful in practice over locally tree-like factor models. \prasadtemp{claiming without knowing}

\subsection{Stable Fixed Point Barrier}	\label{sec:stable-barrier}

A natural starting point for BP iteration for a model $\Model$ is given by the following:
\begin{align}
	\minit^{v \to e} & \defeq \text{ prior distribution } \prior_{\btau(v)} \label{eq:v-to-e-upd} \\
	\minit^{e \to v} & \defeq \text{ uniform distribution over support of } \prior_{\btau(v)}	\label{eq:e-to-v-upd}
\end{align}
Conjecturally, this canonical initialization $\minit$ plays a critical role in characterizing the computational complexity of inferring the hidden variables in model $\Model$.

There appear to be three possible cases with regards to this canonical initialization.

\paragraph{Case 1: $\minit$ is not a fixed point} Suppose $\minit$ is not a fixed point for the BP iteration over the model $\Model$, then BP iteration can be expected to make progress, thereby yielding a weak recovery of hidden variables.

In fact, we will present a self-contained algorithm that weakly-recovers the hidden coloring in this case. Formally, we will show the following in \pref{app:easy-cases}:
\begin{lemma}\label{lem:easy-case}
	If $\minit$ is not a fixed point for the BP iteration on model $\Model$, then there is a polynomial time algorithm $\calA$ and an $\epsilon > 0$ such that
	\begin{enumerate}
		\item if $(\bE,\btau)\sim\Model$: $\calA$ outputs a coloring that beats the correlation random guessing achieves with the hidden coloring by $\eps$,
		\item $\calA$ solves the $\Model$ vs.\ $\Model^{\times}$ distinguishing problem with high probability.
	\end{enumerate}
\end{lemma}
In light of the above lemma, it is natural to restrict our attention to the case where $\minit$ is a fixed point for the BP iteration.  $\minit$ being a fixed point of BP is equivalent to a ``detailed balance'' condition holding (in the sense of \pref{eq:detailed-balance}).

\paragraph{Case 2: $\minit$ is an unstable fixed point}

$\minit$ is an {\it unstable fixed point} if arbitrary small perturbations of $\minit$ will lead to the BP iteration moving away from the fixed point $\minit$.
BP is conjectured to succeed in weak-recovery of hidden coloring and distinguishing between $\Model$ vs.\ $\Model^{\times}$ in this case, and this has been extensively demonstrated experimentally \cite{DKMZ,ZMZ14}.

\paragraph{Case 3: $\minit$ is a stable fixed point}

$\minit$ is a {\it stable fixed point} if there exists a neighborhood $U$ around $\minit$ such that for any initialization $\hat{m} \in U$, BP iteration converges to the canonical fixed point $\minit$.
In this case, the canonical fixed point $\minit$ clearly highlights a potential failure of BP algorithm.
A priori, it is conceivable that by using BP with an alternative starting point or an entirely different algorithm,  one could still efficiently infer the hidden coloring in this case.

Surprisingly, it is conjectured that the existence of this canonical fixed point that is stable marks the onset of computational intractability!
Inspired by ideas from statistical physics, Krzakala and Zdeborova \cite{KrzakalaZ09} were the first to hypothesize that the existence of a trivial fixed point that is stable marks computational intractability.
Building on these intuitions, Decelle et. al. \cite{DKMZ} outlined a fascinating set of conjectures on community detection problem which fuelled a flurry of activity, resulting in algorithms matching the conjectured computational thresholds \cite{mossel2018proof,M14,BLM15,AS15}. 

\subsection{Analyzing Stability}
\label{sec:start-stab}
The stability of the canonical fixed point $\minit$ under BP iteration can be analyzed using derivatives of the BP update rule.
Suppose $\Gamma$ denote the map associated with running two rounds of BP iteration to produce the messages, i.e., 
 \[ \{ m^{v \to e}[t+2] \}_{v \in [n], e \ni v} = \Gamma \left( \{ m^{v \to e}[t] \}_{v \in [n], e \ni v} \right) \]
In other words, $\Gamma$ is given by the composition of the functions in \pref{eq:bpupdate1} and \pref{eq:bpupdate2}.
If $\minit$ is a fixed point of BP, then we will have,
\[ \Gamma(\{ \minit^{v \to e}\}) = \{\minit^{v \to e}\}\]

To analyze the stability of the fixed point $\minit$, one uses the linear approximation of $\Gamma$ in a neighborhood of $\minit$, by setting
\[ \Gamma(\minit+ \epsilon) = \minit + B \epsilon\]
where $B$ is the matrix of partial derivatives, i.e., 
\[ B[m^{u \to e},m^{u' \to e'}] = \frac{\partial \Gamma(m)^{u' \to e'}}{\partial m^{u \to e}} \mid_{\minit} \]
With this linear approximation $\Gamma^{\ell}(\minit + \epsilon) \approx \minit + B^{\ell} \epsilon $.  Therefore, the stability of the fixed point is characterized by the spectral radius of the operator $B$.  

Specifically, $\minit$ is a stable fixed point if and only if $\rho(B) \leq 1$ where $\rho(B) \defeq \max_{i} |\lambda_i(B)|$ is the largest magnitude of an eigenvalue of $B$.

Notice that $B$ is an asymmetric random matrix depending on the set of observations $\bE$.  The \emph{cavity method} is a heuristic to guess the spectral radius of a typical derivative matrix $B$ in terms of the spectral radius of some constant sized linear operator $L$. In the rest of the section we first use the cavity method to obtain a precise condition on $\Model$ for $\rho(B) \leq 1$, then state our main theorem that the distinguishing problem and the weak recovery problem are efficiently solvable when $\rho(B) > 1$, and finally define the operator $L$ whose spectral bound $\lambda_L$ satisifies that $\rho(B) = \lambda_L^{1/2}$.

\subsection{The local distributions of $\Model$}	\label{sec:local-dist-model}

Before diving into the calculation, we define a few local distributions of $\Model$ that would be used later.

\paragraph{The color assigment distribution $\ldist_i$} For each factor $\phi_i \in \Phi$, define a local distribution $\ldist_i$ over $[q]^{\arity(i)}$ as,
\begin{align}
	\ldist_i(c_1,\ldots,c_{\arity(i)}) \propto \left( \prod_{j \in [\arity(i)]} \prior_{\Cl(i)_j}(c_j) \right) \cdot \phi_i(\bc)
\end{align}
For each $\phi_i \in \Phi$ and $a,b \in [\arity(i)]$ define a matrix $\Facm_{i, a\mid b} \in \R^{[q] \times [q]}$ by fixing,
\begin{align}\label{eq:bTrans-cond-prob}
	\Facm_{i,a \mid b}(\alpha,\beta) \defeq \Pr_{(c_1,\ldots,c_{\arity(i)}) \sim \mu_i}[ c_a = \alpha | c_b = \beta  ]
\end{align}
This says that conditioned on that $\phi_i$ is in the observations $\bE$, the matrix $\Facm_{i,a \mid b}$ encodes the color distribution of $a$ conditioned on the color of $b$. 
Finally for $\phi_i \in \Phi$ and $a, b \in [\arity(i)]$ we define a matrix that is useful later,
\begin{align}\label{eq:bTrans}
\bTrans_{i,a \mid b} = (\mathbf{I} - \prior_{\Cl(i)_a} \mathbf{1}^T) \Psi_{i,a \mid b}.
\end{align}

\paragraph{The neighbor factor distribution of a variable}
We now take a closer look at a type $\tau$ variable's neighbor factor distribution. Here a variable's neighbor factors refer to all factors that are connected to the variable in the factor graph.

To study this neighborhood distribution, we first define random variables $\varfac{\tau}{i,j}$ for a type $\tau$ variable $v$.

\begin{definition}\label{def:varfac-rv}
For $\tau\in T, \phi_i\in \Phi$, $\varfac{\tau}{i,j}$ is the random variable denoting the number of type $\phi_i$ factors in the neighborhood of the type $\tau$ variable $v$ such that the index of $v$ in all these factor is $j$.
\end{definition}

From the definition, we see that each $\varfac{\tau}{i,j}$ is the sum of many binomial variables each of which indicates whether a specific type $\phi_i$ factor exisits in the factor graph. We formally define these binomial variables.

\begin{definition}\label{def:varfac-rv}
For $\tau\in T, (v_1,\dots,v_{\arity(i)})\in[n]^{\arity(i)}$, $\varfacbin{\tau}{v_1,\dots,v_{\arity(i)}}$ is the indicator variable of whether the type $\phi_i$ factor $e$ whose $j$-th variable is $v_j$ for all $j\in[\arity(i)]$ is in the observations $\bE$. 
\end{definition}

We can compute the probability of $\varfacbin{\tau}{v_1,\dots,v_{\arity(i)}} = 1$ in $\Modeln$. 
\begin{align*}
\Pr_{\Modeln}\left[\varfacbin{\tau}{v_1,\dots,v_{\arity(i)}} = 1\right] 
&= \frac{1}{\TypeDist(\tau)} \cdot \prod_{j=1}^{\arity(i)}\TypeDist(\Cl(i)_j) \cdot \sum_{(c_1,\dots,c_{\arity(i)})\in [\numcols]^{\arity(i)}} \prod_{j=1}^{\arity(i)}\ColDist_{\Cl(i)_j}(c_j) \cdot \frac{\phi_i(c_1,\dots,c_{\arity(i)})}{n^{\arity(i)-1}} \\
&= \frac{1}{\TypeDist(\tau)} \cdot \frac{\ACFacDeg{i}}{n^{\arity(i)-1}} \cdot \prod_{j=1}^{\arity(i)}\TypeDist(\Cl(i)_j) 
\end{align*}
Thus $\varfacbin{\tau}{v_1,\dots,v_{\arity(i)}}$ has distribution $\text{Binomial}\left(\frac{\ACFacDeg{i}}{\TypeDist(\tau)n^{\arity(i)-1}} \cdot \prod_{j=1}^{\arity(i)}\TypeDist(\Cl(i)_j) \right)$.

Now we can express $\varfac{\tau}{i,j}$ as the sum of $n^{\arity(i)-1}$ binomial random variables.
\[\varfac{\tau}{i,j} = \sum_{(v_j)_j\in[n]^{\arity(i)} \mid v_j = v} \varfacbin{\tau}{v_1,\dots,v_{\arity(i)}}.\]

We also note that most of the $\varfacbin{\tau}{v_1,\dots,v_{\arity(i)}}$s are independent. Two random variables $\varfacbin{\tau}{v_1,\dots,v_{\arity(i)}}$ and $\varfacbin{\tau}{v'_1,\dots,v'_{\arity(i)}}$ are not independent only if there exist $j,j'\in[\arity(i)]$ such that $v_j = v'_{j'}$ but $\Cl(i)_j \neq \Cl(i)_{j'}$. That is the two factors share some variable but require the variable to have different types. However, only $O\left(n^{-(\arity(i)-1)}\right)$ fraction of the pairs are correlated. Thus, when $n$ is large we can treat the $n^{\arity(i)-1}$ random variables as being independent. Then each $\varfac{\tau}{i,j}$ has a Poisson distribution.

\begin{claim}\label{clm:varfac-distribution}
$\varfac{\tau}{i,j} \sim \text{Poisson}\left(\frac{\ACFacDeg{i}}{\TypeDist(\tau)} \cdot \prod_{j=1}^{\arity(i)}\TypeDist(\Cl(i)_j)\right)$.
\end{claim}

For similar reason as above, when $n$ is large we can treat the random variables $\{\varfac{\tau}{i,j}\}_{i\in [F],j\in[\arity(i)]}$ of a variable $v$ as independent. Therefore for large $n$, the neighbor factor distribution of a type $\tau$ variable $v$ is very close to the product distribution of the random variables $\{\varfac{\tau}{i,j}\}_{i\in F,j\in[\arity(i)]}$.

\subsection{The stability condition}
\label{sec:neighorhood-of-variable}

Now we continue to explore the condition on $\Model$ that makes $\minit$ stable. We focus on sparse models whose average factor degrees $\ACFacDeg{i} = O(1)$ for all $\phi_i\in \Phi$. In such models, a variable node is contained in constant number of factors, and its $o(\log{n})$-neighborhood is locally tree-like with high probability. Set the tree depth $\ell$ be a function such that $\ell(n)\in o(\log{n})$, and consider the distance $(2\ell+1)$ neighborhood of a variable $v_0$. Assume each level-$(2\ell+1)$ factor node $e_\ell$'s outgoing message to some level-$2\ell$ variable $v_{\ell}$ is perturbed to $m^{e_\ell\to v_\ell} = \minit^{e_\ell\to v_\ell} + \epsilon_{e_{\ell}}$. Recall that $\minit^{e_\ell\to v_\ell}$ is the trivial fixed point message, and $\epsilon_{e_{\ell}}$ is the random perturbation that is independent across different edges $e_{\ell}\to v_{\ell}$. We want to compute the expected influence of the perturbations on the messages to the root $v_0$. 

We first consider the distance $2\ell+1$ neighborhood of a variable $v_0$. The treelike neighborhood can be constructed by the following process. 
\begin{enumerate}
\item Sample the type of the root variable $\btau(v_0)$ from $\TypeDist$.
\item Sample the level-$1$ factors: for each type of factor $\phi_i\in \Phi$ sample the number of type $\phi_i$ neighbor factors of $v_0$ by sampling $\{\varfac{\btau(v_0)}{i,j}\}_{i\in F,j\in[\arity(i)]}$ independently. Add $v_0$'s neighbor factors to level $1$. Add the other variables in these factors to the next level of the tree, assuming that there is no shared variables other then $v_0$. Note that these variables already have types.
\item Repeat step 2 for the new variables until we get a depth-$(2\ell+1)$ tree $\Tree_{\ell}$.
\end{enumerate}

We next use the tree $\Tree_{\ell}$ to give a precise condition on $\Model$ for the fixed point $\minit$ to be stable. 

In a $\Tree_{\ell}$, a leaf node $e_{\ell}$ is connected to the root node $v_0$ via a path $e_{\ell},v_{\ell},e_{\ell-1},\dots v_1,e_0,v_0$. A perturbation on the leaf message $m^{e_{\ell}\to v_{\ell}}$ influence the next level message $m^{e_{\ell-1}\to v_{\ell-1}}$ via the partial deriviative matrix $\frac{\partial \Gamma(m)^{e_{\ell-1}\to v_{\ell-1}}}{\partial m^{e_{\ell}\to v_{\ell}}}$.
%
%
We can express the partial deriviative matrix evaluated at the fixed point using the matrix defined in \pref{eq:bTrans}.

\begin{claim}[\pref{clm:transitional-matrix}]
\[\frac{\partial \Gamma(m)^{e_{\ell-1}\to v_{\ell-1}}}{\partial m^{e_{\ell}\to v_{\ell}}}\mid_{\minit} \,= \bTrans_{\FacType(e_{\ell-1}),\index_{e_{\ell-1}}(v_{\ell-1})\mid \index_{e_{\ell-1}}(v_{\ell})} .\]
\end{claim}
This claim is proved in \pref{app:calculus}. When writing the matrix $\bTrans_{\FacType(e),\index_{e}(v)\mid \index_{e}(v')}$, it's clear that the index function is associated with the factor $e$, so we drop the subscript $e$ in $\index_{e}$ for simplicity.

Using the chain rule, we can compose the partial deriviative matrices along a path, and conclude that each path influences the root message by 

\[\left(\prod_{j = 1}^{\ell} \bTrans_{\FacType(e_{j-1}),\index(v_{j-1})\mid \index(v_{j})}\right)\epsilon_{e_{\ell}} .\]

Thus the influence of all paths in $\Tree_{\ell}$ is
\[\sum_{(e_{\ell},v_{\ell},\dots,e_0,v_0)\in \Tree_{\ell}} \left(\prod_{j = 1}^{\ell} \bTrans_{\FacType(e_{j-1}),\index(v_{j-1})\mid \index(v_{j})}\right)\epsilon_{e_{\ell}} \]

To decide if the fixed point is stable, we compute the variance $\alpha_\ell$ of this influence.

\begin{align*}
 &\E_{\Tree_{\ell},\epsilon}\left[\norm*{\sum_{(e_{\ell},v_{\ell},\dots,e_0,v_0)\in \Tree_{\ell}} \left(\prod_{j = 1}^{\ell} \bTrans_{\FacType(e_{j-1}),\index(v_{j-1})\mid \index(v_{j})}\right)\epsilon_{e_{\ell}}}^2 \right] \\
=& \E_{\Tree_{\ell},\epsilon}\left[ \sum_{(e_{\ell},v_{\ell},\dots,e_0,v_0)\in \Tree_{\ell}} \epsilon_{e_{\ell}}^{\top}\left(\prod_{j = 1}^{\ell} \bTrans_{\FacType(e_{j-1}),\index(v_{j-1})\mid \index(v_{j})}\right)^{*} \left(\prod_{j = 1}^{\ell} \bTrans_{\FacType(e_{j-1}),\index(v_{j-1})\mid \index(v_{j})}\right)\epsilon_{e_{\ell}} \right] \\
=& \E_{\Tree_{\ell}}\left[ \sum_{(e_{\ell},v_{\ell},\dots,e_0,v_0)\in \Tree_{\ell}} \Tr \left ( \left(\prod_{j = 1}^{\ell} \bTrans_{\FacType(e_{j-1}),\index(v_{j-1})\mid \index(v_{j})}\right) \left(\prod_{j = 1}^{\ell} \bTrans_{\FacType(e_{j-1}),\index(v_{j-1})\mid \index(v_{j})}\right)^{*} \right) \right]\cdot\frac{\E_{\epsilon_{e_{\ell}}}[\norm{\epsilon_{e_{\ell}}}^2]}{\numcols}  
\end{align*}

So the squared norm of the perturbations $\epsilon_{e_{\ell}}$ is amplified by 

\[\alpha_\ell \defeq \E_{\Tree_{\ell}}\left[ \sum_{(e_{\ell},v_{\ell},\dots,e_0,v_0)\in \Tree_{\ell}} \Tr \left ( \left(\prod_{j = 1}^{\ell} \bTrans_{\FacType(e_{j-1}),\index(v_{j-1})\mid \index(v_{j})}\right) \left(\prod_{j = 1}^{\ell} \bTrans_{\FacType(e_{j-1}),\index(v_{j-1})\mid \index(v_{j})}\right)^{*} \right) \right] .\]

We note that for a model $\Model$, by definition of $\alpha_\ell$ and the operator $B$ in \pref{sec:start-stab}, $\rho(B) = \lim_{n\to\infty} \alpha_{\ell}^{1/2\ell}$ (recall that $\ell$ is a function of $n$).
Thus when $\lim_{n\to\infty} \alpha_{\ell}^{1/2\ell} \leq 1$, $\minit$ is a stable fixed point.

\subsection{Efficient recovery and detection when the fixed point is unstable}
\label{sec:dist-hard-case}
For a model $\Model$ whose fixed point $\minit$ is unstable, it is conjectured that the BP algorithm can successfully weak-recover the hidden coloring. We provide a BP-inspired spectral algorithm that solves the weak recovery problem in this regime.  We state the result somewhat informally below; the full formal statement can be found in \pref{thm:recovery-full}.  However, before we state the result we go on a small digression on how to set the benchmark for weak recovery.  A first attempt might be:
\begin{displayquote}
	For a fixed  type $\tau\in T$ and color $c\in[q]$,
	produce a vector $\overline{u}\in\R^{n}$ such that $w$ correlates with the following ``centered indicator vector'' of $(\tau,c)$: $\bchat^{\tau,c}$ where $\bchat^{\tau,\alpha}$ is an $n$-dimensional vector with $i$-th coordinate $\Ind[\tau(i) = \tau]\cdot(\Ind[\bc(i)=c]-\prior_{\tau}(c))$.
\end{displayquote}
However, this benchmark is unattainable since for a problem such as {\sc Planted-$q$-Coloring} there is no way to statistically discriminate between a given coloring and a different coloring obtained by permuting the names of the colors.  Thus, to account for this complication we consider the following modification of the above benchmark, which we first state in words.
\begin{displayquote}
	Produce a vector $\overline{u}\in\R^{n}$ such that after some permutation is applied to the names of the colors, for some type $\tau\in T$ and color $c\in [q]$, $\overline{u}$ correlates with the centered indicator vector of $(\tau,c)$.
\end{displayquote}
More formally:
\begin{theorem}\label{thm:informal-weak-recovery}
	If a model $\Model$ has $\lim_{\ell\to\infty} \alpha_{\ell}^{1/2\ell} > 1$, there is a spectral algorithm $\alg$ that solves the weak-recovery problem.  A bit more concretely, for $\bG\sim\Modeln$, the algorithm $\alg(\bG)$ produces $O_{\Model}(1)$ vectors $\{\ol{u}_1,\ldots,\ol{u}_{r}\}$ such that one of these vectors $\ol{u}_j$ has constant correlation with the planted coloring in the following sense:
	\begin{displayquote}
		There is a type $\tau\in T$, and a color $\alpha\in [q]$ such that if we construct $\bchat^{\tau,\alpha} \in \R^n$ as
		\[
			\bchat^{\tau,\alpha} [i] = \Ind[\tau(i) = \tau] (\Ind[\bc(i) = \alpha] - \prior_{\tau}(\alpha) )
		\]
		then
		\[
			\iprod{ \overline{u}_j, \bchat^{\tau,\alpha}} \geq \Omega_{\Model}(1) \cdot \sqrt{n}.
		\]
	\end{displayquote}

\end{theorem} 

For a model $\Model$ with $\minit$ as an unstable fixed point, it is conjectured that the BP algorithm can successfully distinguish it from the null model $\Null$. We provide a BP-inspired spectral algorithm that solves the detection problem in this regime.  We state the result below but leave the proof sketch to the technical overview section.

\begin{theorem}\label{thm:informal-distinguishing}
If a planted model $\Model$ has $\lim_{\ell\to\infty} \alpha_{\ell}^{1/2\ell} > 1$, there is a spectral algorithm $\alg : \text{factor graphs} \to \{\scp, \scn\}$ that solves the detection problem in the following sense 
\[\Pr_{\Modeln}[\alg(\text{factor graph}) = \scp] = 1 - o_n(1) \quad \text{and} \quad \Pr_{\Null_{n}}[\alg(\text{factor graph}) = \scn] = 1 - o_n(1).\]
\end{theorem} 

\subsection{The stability condition via a finite linear operator $L$} \label{sec:from-alpha-to-lambda}

In this part we define a finite linear operator $L$ whose spectral radius gives a criterion for when $\lim_{\ell\to\infty} \alpha_{\ell}^{1/2\ell}$ is greater than $1$ or less than $1$.

Naively, computing $\lim_{\ell\to\infty} \alpha_\ell^{1/2\ell}$ requires us to consider trees whose size grows with $\ell$. However we can simplify the expression for $\alpha_\ell^{1/2\ell}$ via an insight of \cite{BC19} by observing that the tree $\Tree_{\ell}$ is constructed recursively. For any even level variable node $v_{k}$ in the tree, the distribution of its children factor nodes depends only on $v_{k}$'s type. Furthermore, the factor node distribution $\{\varfac{\btau(v_k)}{i,j}\}_{i\in F,j\in[\arity(i)]}$ also fully describes the distance $2$ neighborhood of $v_k$. For a type $\tau$ variable $v$, define the random variable $\varvar{\tau}{i,j,j'}$ to be the number of variables $u$ that are connected to $v$ via some type $\phi_i$ factor, and additionally $u$ have index $j'$ in the factor and $v$ has index $j$. These random variables give a way to concisely express the total influence of the distance $2$ type $\taup$ variables to the type $\tau$ variable. This influence is 

\[\sum_{i, j, j'\mid \Cl(i)_{j'} = \taup} \varvar{\tau}{i,j,j'}\cdot\bTrans_{i,j \mid j'} \in \R^{[\numcols]\times[\numcols]}.\]

Then we can build a $2$ step quadratic influence operator $L:\R^{T\cdot[\numcols]\times T\cdot[\numcols]} \to \R^{T\cdot[\numcols]\times T\cdot[\numcols]}$ such that the $(\tau,\tau)$ block of $L(M)$ is

\[L(M)^{\tau,\tau} = \sum_{\taup} \sum_{i\in F, j,j'\in[\arity(i)]} \sum_{i, j, j'\mid \Cl(i)_{j'} = \taup} \varvar{\tau}{i,j,j'}\cdot\bTrans_{i,j \mid j'}\, M^{\taup,\taup}\, \bTrans_{i,j \mid j'}^{*},\]
and the off-diagonal blocks are $L(M)^{\tau,\taup} = \mathbf{0}$.

Suppose the input $M$ is such that each $M^{\taup,\taup}$ captures the quadratic influence of some path $v_1,e_1,\dots,v_{\ell},e_{\ell}$ such that the endpoint $v_1$ has type $\taup$,
\[\left(\prod_{j = 2}^{\ell} \bTrans_{\FacType(e_{j-1}),\index(v_{j-1})\mid \index(v_{j})}\right) \left(\prod_{j = 2}^{\ell} \bTrans_{\FacType(e_{j-1}),\index(v_{j-1})\mid \index(v_{j})}\right)^{*} .\] 
And all other blocks in $M$ are $0$. 
Then after applying the operator, every diagonal block $L(M)^{\tau,\tau}$ captures the expected quadratic influence of all paths $v_0,e_0,v_1,e_1,\dots,v_{\ell},e_{\ell}$ which are $2$-step extensions of the path $v_1,e_1,\dots,v_{\ell},e_{\ell}$ and whose endpoint $v_0$ has type $\tau$.

Using this operator $L$ we can rewrite $\alpha_\ell$ as 
\[\alpha_\ell 
= \E_{\Tree_{\ell}}\left[ \sum_{(e_{\ell},v_{\ell},\dots,e_0,v_0)\in \Tree_{\ell}} \Tr \left ( \left(\prod_{j = 1}^{\ell} \bTrans_{\FacType(e_{j-1}),\index(v_{j-1})\mid \index(v_{j})}\right) \left(\prod_{j = 1}^{\ell} \bTrans_{\FacType(e_{j-1}),\index(v_{j-1})\mid \index(v_{j})}\right)^{*} \right) \right] 
= \Tr\left(L^{\ell}(\Diag(\ColDist))\right) ,
\]
where $\Diag(\ColDist)\in \R^{T\cdot[\numcols]\times T\cdot[\numcols]}$ is a diagonal matrix whose $((\tau,c),(\tau,c))$ entry has value $\ColDist_{\tau}(c)$.

Use $\lambda_L$ to denote the maximum eigenvalue of $L$. Then by standard linear algebra fact, 
\[\lim_{n\to\infty}\Tr\left(L^{\ell}(\Diag(\ColDist))\right)\to \lambda_L^{\ell}.\] Therefore we obtain the following equivalence relation between $\alpha_{\ell}$ and $\lambda_L$.
\begin{lemma}\label{lem:from-alpha-to-lambda}
For a model $\Model$, $\lim_{\ell\to\infty}\alpha_\ell^{1/2\ell} = \lambda_L^{1/2}$. Thus the fixed point $\minit$ of $\Model$ is stable if and only if $\lambda_L \leq 1$. 
\end{lemma}

\section{Technical Overview}
\subsection{Algorithm for distinguishing}
We now describe our algorithm for distinguishing if an instance $\bG$ was sampled from the null distribution $\Null$ from the planted distribution $\Planted$.  Our algorithm constructs a matrix $M_{\bG}$ obtained from linearizing $\ell$ rounds of the belief propagation algorithm at the uninformative fixed point on input $\bG$ and tests if its largest eigenvalue exceeds a chosen threshold $\kappa$.  If it does then the algorithm declares that $\bG$ came from the planted distribution, and otherwise claims $\bG$ was sampled from the null distribution.  A bulk of the technical work is in proving that this particular matrix $M_{\bG}$ has all its eigenvalues bounded by the chosen threshold $\kappa$ when $\bG$ is sampled from the null model, and in illustrating that $M_{\bG}$ has an ``outlier'' eigenvalue exceeding $\kappa$ otherwise.  In this section, we delve more into the description of $M_{\bG}$ and then give a brief description of how we prove the statements about the eigenvalues of $M_{\bG}$ in the null and planted models.

More concretely, given a random instance $\bG$ sampled either from $\Null$ or $\Planted$, we set $M_{\bG}$ as the following matrix $\CentAdj_{\bG}^{(\ell)}$ called the \emph{length-$\ell$ centered nonbacktracking walk power} of $\bG$, for which we provide a slightly informal description below.
\begin{definition}[Centered nonbacktracking power (slightly informal)]	\label{def:centered-nb-power}
	$\CentAdj_{\bG}^{(\ell)}$ is a $nq\times nq$ matrix which we treat as a $n\times n$ grid of $q\times q$ blocks.  The block rows and columns are indexed by $[n]$.  In the $(i,j)$-th block, we place the following $q\times q$ matrix:
	\begin{align*}
		\sum_{\substack{i e_1 v_1 e_2 v_2 \dots e_{\ell} j \in \\ \text{all nonbacktracking walks} \\ \text{from $i$ to $j$ in complete} \\ \text{factor graph}}} \BPM{e_1}{i}{v_1}\cdot\BPM{e_2}{v_1}{v_2}\cdots\BPM{e_{\ell}}{v_{\ell-1}}{j}\cdot (\Ind[e_1\in\bG]-\Pr_{\Null}[e_1\in\bG])\cdots(\Ind[e_\ell\in\bG]-\Pr_{\Null}[e_{\ell}\in\bG]).
	\end{align*}
\end{definition}

Recall the matrix $L$ from \pref{sec:dist-hard-case} which the stability prediction of belief propagation was based on, and let $\lambda_L$ denote its largest eigenvalue.  The two main technical theorems we prove about $\CentAdj_{\bG}^{(\ell)}$ in service of proving that our algorithm is correct with high probability are:
\begin{theorem}[Local statistics in planted model]	\label{thm:eig-lower-bound-plant}
	When $\bG\sim\Planted$, with probability $1-o_n(1)$: $\lambda_{\max}(\CentAdj_{\bG}^{(\ell)}) \ge \frac{\lambda_L^{\ell}}{q}$.
\end{theorem}
The proof of this is carried out in \pref{sec:statistics} and uses two ingredients: the first is recognizing that $\CentAdj_{\bG}^{(\ell)}$ is self-adjoint under a certain inner product $\langle\cdot,\cdot\rangle_{\innerp}$, due to which for any vector $x$:
\[
	\langle x, \CentAdj_{\bG}^{(\ell)} x\rangle_{\innerp} \le \lambda_{\max}\left( \CentAdj_{\bG}^{(\ell)} \right)\cdot\langle x, x\rangle_{\innerp}.
\]
The second ingredient is in identifying a vector $x$ depending on the planted solution the instance $\bG$ was sampled with which makes the above quadratic $\langle x, \CentAdj_{\bG}^{(\ell)} x\rangle_{\innerp}$ larger than the desired lower bound of $\frac{\lambda_L^{\ell}}{q}$ with high probability.

The second main technical theorem, which is proved in \pref{thm:eig-upper-bound-null} is:
\begin{theorem}[Eigenvalue bound in null model]	\label{thm:eig-upper-bound-null}
	When $\bG\sim\Null$ and $(\log\log n)^2\le\ell\le \frac{\log n}{(\log\log n)^2}$, for every constant $\eps > 0$ with probability $1-o_n(1)$, all eigenvalues of $\CentAdj_{\bG}^{(\ell)}$ are bounded in magnitude by $((1+\eps)\sqrt{\lambda_L})^{\ell}$.
\end{theorem}
When $\lambda_L > 1$, we choose $\delta>0$ so that $1+\delta<\sqrt{\lambda_L}$, $\ell$ as $(\log\log n)^2$, and $\kappa$ as $((1+\delta)\sqrt{\lambda_L})^{\ell}$.  Then as an immediate consequence of \pref{thm:eig-lower-bound-plant} and \pref{thm:eig-upper-bound-null} we know that the algorithm correctly distinguishes between $\Null$ and $\Planted$ with high probability.

We now elaborate on \pref{thm:eig-upper-bound-null} and elucidate the exact random matrix concentration statement.

\subsection{Matrix concentration vignette}
Consider an Erd\H{o}s-R\'enyi graph $\bH$ sampled from $G\left(n,\frac{d}{n}\right)$.  Pick a random vertex $v$ in $\bH$ and observe a ``large'' radius neighborhood around $v$.  Typically, this neighborhood around $v$ will be a tree, and additionally, and in the large-$n$ limit the distribution of this tree is a \emph{Galton-Watson process} -- a random (possibly infinite) tree $\bbT$ generated by starting at a root vertex $r$, attaching $\Poi(d)$ children to $r$, and then attaching $\Poi(d)$ children to each child of $r$ and so on.\footnote{The reader is advised to not pay too much attention to the fact that the number of children are distributed according to a \emph{Poisson} random variable.  The important property is that a vertex has $d$ children on average.}  So this tells us that there is some sense in which $\bbT$ ``approximates'' the finite random graph.  This intuition is spectrally articulated by a theorem which is (implicitly) due to \cite{BLM15} (see also \cite{FM17} and \cite{BMR19}).  Before we state the theorem, we bring up a natural quantity to associate to the random tree: the \emph{growth rate} which is defined as
\[
	\lim_{\ell\to\infty}\E[\#\text{ of vertices at depth-}\ell]^{1/\ell},
\]
which is equal to $d$ for the aforementioned Galton-Watson process.
\begin{theorem}[{\cite{BLM15,FM17,BMR19}}]	\label{thm:BLM-main}
	Let $\CentAdj_{\bH} \coloneqq \Adj_{\bH} - \E\Adj_{\bH}$ be the centered adjacency matrix of $\bH$. Suppose $d > 1$, then:
	\[
		|\lambda|_{\max}(\CentAdj_{\bH}^{(\ell)}) \le ((1+o(1))\sqrt{d})^{\ell} = ((1+o(1))\sqrt{\text{growth rate of }\bbT})^{\ell}
	\]
	for $\ell\in\left[(\log\log n)^2,\frac{\log n}{\log\log n}\right]$.
\end{theorem}

Now, let's add a small twist: sample two \Erdos-\Renyi~graphs $\bH_1\sim G\left(n,\frac{d_1}{n}\right)$ and $\bH_2\sim G\left(n,\frac{d_2}{n}\right)$ and consider the weighted graph $\bH = 0.9\bH_1-\bH_2$.  The random tree that $\bH$ locally resembles is the following \emph{different} Galton-Watson process $\bbT'$: start at a root vertex $r$, connect $\Poi(d_1)$ children with edges of weight $0.9$ and $\Poi(d_2)$ children with edges of weight $-1$ to $r$, then repeat the same for each child vertex, and keep going.  The following quantity is the correct generalization of growth rate to weighted graphs, which we call the \emph{weighted growth rate} of the tree:
\[
	\wgr(\bbT') \coloneqq \lim_{\ell\to\infty} \E\left[\sum_{\substack{P\in\text{length-}\ell\text{ paths}\\ \text{starting at root}}} \prod_{e\in P}w_e^2\right].
\]
For illustrative purposes, one subcase of our matrix concentration result is:
\begin{theorem}	\label{thm:special-case}
	Let $\CentAdj_{\bH}\coloneqq \Adj_{\bH}-\E\Adj_{\bH}$ be the centered adjacency matrix of $\bH$.  Suppose $\wgr(\bbT')>1$, then:
	\[
		|\lambda|_{\max}(\CentAdj_{\bH}^{(\ell)}) \le ((1+o(1))\sqrt{\wgr(\bbT')})^{\ell}
	\]
	for $\ell\in\left[(\log\log n)^2,\frac{\log n}{(\log\log n)^2}\right]$
\end{theorem}
We now discuss our full matrix concentration theorem which captures both of the above mentioned theorems.  Before doing so, it is worth noting that the picture for random graphs being spectrally approximated by infinite graphs is far more well understood in the setting of models of random \emph{regular} graphs through works of \cite{Fri03,Bor19,BC19,MOP20a,OW20} but we defer the readers to \cite{OW20} for an extensive discussion of what is known in that setting.

\subsection{Matrix concentration statement}
Let $\langle \cdot,\cdot \rangle_{\nu}$ be an inner product on $\R^{nq}$ and let $M^*$ denote the adjoint of a matrix $M$ under this inner product.  We consider $nq\times nq$ random matrices sampled according to the following model (whose notation the reader should treat independently from the preceding notation related to distinguishing instances from the null and planted distributions).
\begin{definition}[Random matrix model (slightly informal)]	\label{def:random-matrix-model}
	The model has an underlying \emph{left vertex set} which is equal to $[n]$.  First, every vertex $v$ is assigned a \emph{type} $\btau(v)$ in $[T]$ sampled according to a distribution $\pi$.  There are $F$ types of right vertices, given by set $[F]$.  Each right vertex type $i$ comes with an \emph{arity} $k_i$, which is a positive integer, a \emph{profile} $\chi_i$ which is a tuple in $[T]^{k_i}$, a collection of ${k_i(k_i-1)}$ matrices $\{M_{i,(a,b)}\}_{(a,b)\in[k_i]^2:a\ne b}$, and a \emph{density} $\phi_i$.  A random instance $\bH$ is sampled in the following way: for every $(i,(v_1,\dots,v_{k_i}))$ for $i\in[F]$ and tuple $(v_1,\dots,v_{k_i})$ in $[n]^{k_i}$ of distinct elements such that $(\btau(v_1),\dots,\btau(v_{k_i}))=\chi_i$ we add $(i,(v_1,\dots,v_{k_i}))$ as a right vertex with probability $\frac{\phi_i}{n^{k_i-1}}$, connect edges to $v_1,\dots,v_{k_i}$ and mark the edge to $v_t$ with number $t$.  We use $\CompG_n$ to refer to the bipartite graph with left vertex set $[n]$, and the right vertex set containing every potential right vertex.  Now let $\gamma=(i,(v_1,\dots,v_{k_i}))$; for a two-step $v_a \to \gamma \to v_b$  in the complete graph for $a\ne b$ we use $M_{v_a\gamma v_b}$ to denote the matrix $M_{i,(a,b)}$. The random matrix we are interested in, which we denote $\CentAdj_{\bH}^{(\ell)}$, is the matrix where the $uv$ entry contains:
	\[
		\CentAdj_{\bH}^{(\ell)}[i,j] \coloneqq \sum_{\substack{i\gamma_1v_1\dots\gamma_{\ell}j\\ \text{nonbacktracking walks in }\CompG}} M_{i\gamma_1 v_1}\cdots M_{v_{\ell-1}\gamma_{\ell}j}\cdot(\Ind[\gamma_1\in\bH]-\Pr[\gamma_1\in\bH])\cdots(\Ind[\gamma_{\ell}\in\bH]-\Pr[\gamma_{\ell}\in\bH]).
	\]
\end{definition}

\begin{definition}[Galton-Watson tree approximating random matrix (informal)]	\label{def:local-tree-approx}
	For a given setting of parameters for the random model from \pref{def:random-matrix-model}, the bipartite Galton-Watson tree $\bbT$ which ``locally resembles'' an instance $\bH$ sampled from the model is as follows:
	\begin{enumerate}
		\item Start with a left root vertex $r$ and assign it type $\btau(r)\sim\pi$.
		\item \label{step:sample-child} For each $i\in[F]$ and each $j\in[k_i]$ such that $(\chi_i)_j = \btau(r)$ sample
		\[
			n_{i,j}\sim\Poi\left(\frac{\phi_i}{\pi_{\btau(v)}}\prod_{t=1}^{k_i}\pi_{(\chi_i)_t}\right),
		\]
		and attach $n_{i,j}$ right vertices of type $i$ to $r$ and mark the corresponding edge with $j$.  Then to each such right vertex, attach $k_i-1$ (left vertex) children and mark the edges with numbers from $[k_i]\setminus\{j\}$.  To each added child vertex $v$ with edge marked with $t$, assign it type $(\chi_i)_t$.
		\item Repeat \pref{step:sample-child} for each added left vertex child.
	\end{enumerate}
	We define the \emph{matrix weighted growth rate} of $\bbT$ to be the following:
	\[
		\mwgr(\bbT) \coloneqq \lim_{\ell\to\infty} \E\left[\Tr\left(\sum_{\substack{a\gamma_1v_1\dots\gamma_{\ell}b\\ \text{nonbacktracking walks in }\CompG}} M_{a\gamma_1v_1}\dots M_{v_{\ell-1}\gamma_{\ell}b} M_{b\gamma_{\ell}v_{\ell-1}}\dots M_{v_1\gamma_1 a} \right) \right]^{1/\ell}.
	\]
\end{definition}
We prove:
\begin{theorem}[Main matrix concentration theorem (slightly informal)]	\label{thm:main-mat-conc-informal}
	Let $\bH$ be a random instance of a setting of parameters for the model from \pref{def:random-matrix-model} and let $\bbT$ be the tree which locally approximates $\bH$ in the sense of \pref{def:local-tree-approx}.  Suppose:
	\begin{enumerate}
		\item For every right vertex $\gamma=(i,(v_1,\dots,v_{k_i}))$ in $\CompG_n$ and every $1\le a,b\le k_i$ for distinct $a,b$ the $nq\times nq$ matrix obtained by placing $M_{v_a\gamma v_b}$ in the $(v_a,v_b)$ block and zeros everywhere else is adjoint to the $nq\times nq$ matrix obtained by placing $M_{v_b\gamma v_a}$ in the $(v_b,v_a)$ block and zeros everywhere else under $\langle \cdot,\cdot \rangle_{\nu}$.
		\item There is a constant $C$ such that for any nonbacktracking walk $v_0\gamma_1 v_1\dots \gamma_s v_s$ in $\CompG_n$:
		\[
			\|M_{v_0\gamma_1v_1}\dots M_{v_{s-1}\gamma_sv_s}\|\le C.
		\]
		\item $\mwgr(\bbT)\ge1$.
	\end{enumerate}
	Then if $(\log\log n)^2\le \ell \le \frac{\log n}{(\log\log n)^2}$, with probability $1-o_n(1)$:
	\[
		|\lambda|_{\max}(\CentAdj_{\bH})^{(\ell)} \le ((1+o(1))\sqrt{\mwgr(\bbT)})^{\ell}.
	\]
\end{theorem}
The full formal set-up for \pref{thm:main-mat-conc-informal} along with its proof is in \pref{sec:eig-bounds}.

\section{A conjectured detection/recovery threshold} 
\label{sec:cavity-method}

In \pref{sec:from-alpha-to-lambda} we see the connection between $\lambda_L$ and $\alpha_{\ell}$ of a model $\Model$. In this section we prove this connection more rigorously, and conclude with a conjectured weak-recovery threshold in terms of $\lambda_L$. We start by quickly going through the definitions of the partial derivative matrices $\bTrans_{\FacType(e_j),\index_{e_j}(v_j)\mid\index_{e_j}(v_{j+1})}$, the color distribution matrix $\Dcolor_{\tau}$, the influence variance $\alpha_{\ell}$, and their connections. 

\begin{claim}\label{clm:transitional-matrix}
	The following matrices satisfy:
	\begin{enumerate}
		\item $\bTrans_{\FacType(e_j),\index_{e_j}(v_j)\mid\index_{e_j}(v_{j+1})} =\left(\mathbf{I} - \ColDist_{\btau(v_j)}\mathbf{1}^{\top} \right)\Facm_{\FacType(e_j),\index_{e_j}(v_j)\mid \index_{e_j}(v_{j+1})}$.
		\item \label{part:self-adj} Let $\Dcolor_{\tau} := \text{Diag}(\ColDist_{\tau})$.  Then $\Dcolor_{\btau(v_j)}^{\dagger}\BPM{\FacType(e_j)}{\index(v_j)}{\index(v_{j+1})} = \BPM{\FacType(e_j)}{\index(v_{j+1})}{\index(v_{j})}^{\top}\Dcolor_{\btau(v_{j+1})}^{\dagger}$ where the $\dagger$ in the superscript denotes the pseudoinverse of the matrix.
	\end{enumerate}
\end{claim}
This claim is proved in \pref{app:calculus}. From here on we drop the subscript ${e}$ from the index function $\index_{e}$ whenever it is clear from the context which factor is being considered.

Using these notations we can express the influence variance, or the \emph{amplification factor}, $\alpha_{\ell}$ as follows:
\[
	\alpha_\ell \coloneqq \E_{\Tree}\bracks*{\Tr\left(\sum_{(e_\ell,v_\ell,\dots,e_0,v_0)\in\Tree} \left(\prod_{i=1}^{\ell} \bTrans_{\FacType(e_{i-1}),\index(v_{i-1})\mid \index(v_{i})} \right) \left(\prod_{i=1}^{\ell} \bTrans_{\FacType(e_{i-1}),\index(v_{i-1})\mid \index(v_{i})} \right)^*\right)}.
\]
 The key property of the above amplification factor we are interested in is its limiting behavior as $k$ goes to infinity. We say that the uninformative fixed point of the belief propagation update rule is \emph{stable} if $\lim_{n\to\infty}\alpha_\ell^{1/2\ell} \leq 1$, and \emph{unstable} otherwise. Furthermore, if the fixed point is stable, the problem of weak-recovering the hidden coloring of $\Model$ is conjectured to be hard, and if the fixed point is unstable, this problem is conjectured to be easy.

In the remainder of this section, we focus on obtaining a simpler criterion for stability by simplifying the expression for the amplification factor.  In particular, we give a constant dimensional linear transformation whose top eigenvalue is greater than $1$ if the fixed point is stable and is less than $1$ if the fixed point is unstable.  For a leaf-to-root path $e_\ell v_\ell e_{\ell-1} v_{\ell-1}\dots e_0 v_0$ we say its type is
\[
	\FacType_\ell \to \goout_\ell \to \tau_\ell \to \goin_{\ell-1} \to \FacType_{\ell-1} \to \goout_{\ell-1} \to \dots \to \goout_{1} \to \tau_1 \to \goin_{0} \to \FacType_0 \to \goout_0 \to \tau_0
\]
where $\FacType_t$ is the factor type of $e_t$, $\tau_t$ is the variable type of $v_t$, $\goout_{t}$ is the index of $v_t$ in $e_{t}$, and $\goin_t$ is the index of $v_{t+1}$ in $e_{t}$.  The amplification factor can then be written as:
\begin{align*}
	\alpha_\ell = \Tr\left(\sum_{\FacType_\ell\to\dots\to\tau_0} M(\FacType_\ell\to\dots\to\tau_0)\right)
\end{align*}
where the sum is enumerated over all leaf-to-root path types and $M(\FacType_\ell\to\dots\to\tau_0)$ is defined as follows:
\begin{align*}
	M(\FacType_\ell\to\dots\to\tau_0) &= \TypeDist(\tau_0) \cdot \left(\prod_{t=0}^{\ell-1}\frac{\FacDeg_{\FacType_t}}{\TypeDist(\tau_{t})}\cdot \prod_{s=1}^{\arity(\FacType_t)} \TypeDist(\Cl(\FacType_t)_s)\cdot\Ind[\Cl(\FacType_t)_{\goout_t}=\tau_t]\cdot\Ind[\Cl(\FacType_t)_{\goin_{t}}=\tau_{t+1}]\right)\cdot\\
	&\left(\prod_{t=0}^{\ell-1} \bTrans_{\FacType_t,\goout_t\mid\goin_t} \right)\left(\prod_{t=0}^{\ell-1} \bTrans_{\FacType_t,\goout_t\mid\goin_t} \right)^*\cdot\frac{\FacDeg_{\FacType_\ell}}{\TypeDist(\tau_\ell)}\prod_{s=1}^{\arity(i_\ell)}\TypeDist(\Cl(\FacType_\ell)_s)\cdot\Ind[\Cl(\FacType_\ell)_{\goout_\ell}=\tau_\ell].
\end{align*}
Now, let's define $V_{\ell,\tau}$ as:
\begin{align*}
	V_{\ell,\tau} \coloneqq \sum_{\substack{\FacType_\ell\to\dots\to\tau_0 \\ \tau_0 = \tau}} M(\FacType_\ell\to\dots\to\tau_0)
\end{align*}
and $V_\ell$ as the following $|T|\cdot\numcols\times|T|\cdot\numcols$ block diagonal matrix comprised of $\numcols\times \numcols$-dimensional blocks with block rows and columns indexed by $T$:
\[
	V_\ell[\tau,\tau] = V_{\ell,\tau}.
\]
Finally, we define a linear transformation $L$ on the space of $|T|\cdot\numcols\times|T|\cdot\numcols$ matrices.  To define $L(M)$ we treat $M$ as a block matrix comprised of $\numcols\times\numcols$-dimensional blocks with blocks rows and columns indexed by $T$.
\begin{align*}
	L(M)[\tau,\tau] \coloneqq ~& \TypeDist(\tau)\sum_{\tau'\to\goin\to\FacType\to\goout\to\tau} \frac{\FacDeg_{\FacType}}{\TypeDist(\tau)} \prod_{s=1}^{\arity(\FacType)} \TypeDist(\Cl(\FacType)_s)\cdot\Ind[\Cl(\FacType)_{\goout}=\tau]\cdot\Ind[\Cl(\FacType)_{\goin}=\tau'] \cdot \\
	&\bTrans_{\FacType,\goout\mid\goin} \cdot M[\tau',\tau'] \cdot \bTrans_{\FacType,\goout\mid\goin}^*.\\
	L(M)[\tau_1,\tau_2] \coloneqq ~& 0 &\text{for $\tau_1\ne\tau_2$.}
\end{align*}
Now observe that $V_{\ell+1} = L(V_\ell)$; consequently $V_\ell = L^\ell(V_0)$ and $\alpha_\ell = \Tr\left(L^\ell(V_0)\right)$.

We now connect the limiting behavior of the amplification factor $\alpha_\ell$ to the eigenvalues of $L$.  We start by making a few observations:
\begin{observation}	\label{obs:L-nonneg-type}
	If $M$ is a positive semidefinite matrix, then $L(M)$ is also a positive semidefinite matrix.
\end{observation}

\begin{observation}	\label{obs:V0-pd}
	$V_0$ is a diagonal matrix with strictly positive entries on its diagonal and hence is positive definite.
\end{observation}

\begin{observation}	\label{obs:frob-norm-trace}
	Since $L^\ell(V_0)$ is positive semidefinite:
	\[
		\|L^\ell(V_0)\|_F \le \Tr\left(L^\ell(V_0)\right) \le \sqrt{\numcols\cdot|T|}\|L^\ell(V_0)\|_F.
	\]
\end{observation}

Our final ingredient is a lemma that appears in \cite[Theorem 16, part (ii)]{BC19}.
\begin{lemma}	\label{lem:spec-neg-type}
	If $L$ is a linear operator on $r\times r$ matrices with spectral radius $\lambda_L$ such that for any positive semidefinite matrix $M$, $L(M)$ is also positive semidefinite, then for any positive definite $M'$,
	\[
		\lambda_L = \lim_{\ell\to\infty} \|L^\ell(M')\|_F^{1/\ell}.
	\]
\end{lemma}

By \pref{obs:L-nonneg-type}, \pref{obs:V0-pd} and \pref{lem:spec-neg-type}: $\lambda_L = \lim_{\ell\to\infty} \|L^\ell(V_0)\|_F^{1/\ell}$.  \pref{obs:frob-norm-trace} lets us conclude:
\[
	\lambda_L = \lim_{\ell\to\infty} \Tr\left(L^\ell(V_0)\right)^{1/\ell} = \lim_{\ell\to\infty} \alpha_\ell^{1/\ell}.
\]
Thus we prove \pref{lem:from-alpha-to-lambda}, and make the following conjecture:
\begin{conjecture}	\label{conj:L-matrix-eigvals}
	If $\lambda_L > 1$, then $\lim_{\ell\to\infty}\alpha_\ell$ goes to $\infty$ we conjecture that recovery is easy and if $\lambda_L \leq 1$, then $\lim_{\ell\to\infty}\alpha_\ell = 0$ and we conjecture that it is hard.
\end{conjecture}

\section{A spectral distinguishing algorithm}	\label{sec:spectral}
We now describe the spectral distinguisher we use, which is based on linearizing the belief propagation algorithm outlined in \pref{sec:cavity-method}.  Recall that given $\bG$ sampled from either $\Null$ or $\Planted$ our goal is to output ``null'' if $\bG\sim\Null$ and ``planted'' if $\bG\sim\Planted$ with probability $1-o(1)$.  Further, the messages given by \pref{eq:v-to-e-upd} and \pref{eq:e-to-v-upd} are a fixed point for the BP update rule for $\Planted$ (which is equivalent to the detailed balanced condition \pref{eq:detailed-balance} holding).  The sample $\bG$ is given by the tuple $([n],\bE_1,\dots,\bE_F,\btau)$.  Our algorithm constructs a matrix called the \emph{null-centered nonbacktracking power matrix} and thresholds on its largest eigenvalue against a particular value $t$ which is a function of the null and planted models and outputs ``planted'' if the largest eigenvalue exceeds $t$ and ``null'' otherwise.

Recall the definition of the matrix $\CentAdj_{\bG}^{(\ell)}$.
\begin{align*}
		\sum_{\substack{i e_1 v_1 e_2 v_2 \dots e_{\ell} j \in \\ \text{all nonbacktracking walks} \\ \text{from $i$ to $j$ in complete} \\ \text{factor graph}}} \BPM{e_1}{i}{v_1}\cdot\BPM{e_2}{v_1}{v_2}\cdots\BPM{e_{\ell}}{v_{\ell-1}}{j}\cdot (\Ind[e_1\in\bG]-\Pr_{\Null|\btau}[e_1\in\bG])\cdots(\Ind[e_\ell\in\bG]-\Pr_{\Null|\btau}[e_{\ell}\in\bG])
\end{align*}
where the complete factor graph is defined as follows.
\begin{definition}	\label{def:comp-fac-graph}
	We define the \emph{complete factor graph} $\CompG_{n}=([n],E_1,\dots,E_F)$ where $E_i$ denotes the collection of all potential type-$i$ factors that could appear in $\bG$.
\end{definition}

We now describe our algorithm.\footnote{The details for why each step can be carried out efficiently are briefly discussed at the end of this section.}
\begin{itemize}
	\item Compute a matrix representation of the linear operator $L$ from the statement of \pref{conj:L-matrix-eigvals}.
	\item Let $\lambda_L$ be the spectral radius of $L$.
	\item Choose $\kappa$ strictly in between $\sqrt{\lambda}_L$ and $\lambda_L$.
	\item Let $s = \ceil{\sqrt{\log n}}$.  Compute $\CentAdj_{\bG}^{(s)}$ and compute its largest eigenvalue $\rho$.
	\item If $\rho > \kappa^s$, output ``planted'', otherwise output ``null''.
\end{itemize}

To prove that the above algorithm works it suffices to prove that when $\bG$ is sampled from the null distribution, all its eigenvalues are all less $\kappa^s$ and when $\bG$ is sampled from the planted distribution there is an eigenvalue greater than $\kappa^s$.  Henceforth we assume $\lambda_L > 1$.  To prove both of these facts under the hypothesis that $\lambda_L>1$, one ingredient we need is that the matrix $\CentAdj_{\bG}^{(s)}$ is self-adjoint under an appropriate inner product.

Given a vector in $\R^{nq}$ we treat it as a block vector comprising of $n$ blocks of dimension $q$ each where each block corresponds to a vertex in $[n]$.  Now, we define a $nq\times nq$-dimensional positive diagonal matrix $\innerp_{\btau}$ where the $(v,v)$ block is equal to:
\[
	\innerp_{\btau,(v,v)}[c,c] \coloneqq
	\begin{cases}
		\ColDist_{\btau(v)}(c) & \text{if $\ColDist_{\btau(v)}(c)>0$} \\
		1 & \text{otherwise.}
	\end{cases}
\]
We will use the following inner product on $\R^{nq}$:
\[
	\langle x, y\rangle_{\innerp} \coloneqq x^{\top}\innerp_{\btau}^{-1}y.
\]
\begin{remark}
	Ideally, we would like to simply place $\ColDist_{\btau(v)}(c)$ in every diagonal entry $[(v,c),(v,c)]$ and use the pseudoinverse of $\innerp$ instead.  But doing so leads to some complications of the defined bilinear form not necessarily satisfying strict positive definiteness required of an inner product.  The choice of $1$ is arbitrary and does not influence any of the statements since all the vectors we work with have zeros in the $(v,c)$ coordinates where $\ColDist_{\btau(v)}(c)$ is $0$, and also that coordinate subspace is in the kernel of every matrix we work with.
\end{remark}

\begin{claim}	\label{clm:self-adjointness}
	The matrix $\CentAdj_{\bG}^{(s)}$ is self-adjoint under $\langle\cdot,\cdot\rangle_{\innerp}$.
\end{claim}
\begin{proof}
	For any vectors $x,y$:
	\begin{align*}
		\langle x, \CentAdj_{\bG}^{(s)} y\rangle_{\innerp} &= \sum_{u,v\in[n]} \sum_{\substack{u e_1 v_1 e_2 v_2 \\\dots e_{\ell} v \in \CompG_{n}}} x[u]^{\top} \innerp^{-1} \BPM{e_1}{u}{v_1}\cdot\BPM{e_2}{v_1}{v_2}\cdots\BPM{e_{\ell}}{v_{\ell-1}}{v} y[v] \cdot \\ & (\Ind[e_1\in\bG]-\Pr_{\Null|\btau}[e_1\in\bG])\cdots(\Ind[e_\ell\in\bG]-\Pr_{\Null|\btau}[e_{\ell}\in\bG])
	\end{align*}
	From \pref{part:self-adj} of \pref{clm:transitional-matrix},
	\begin{align*}
		x[u]^{\top}\innerp^{-1}[u,u]\BPM{e_1}{u}{v_1}\cdots\BPM{e_{\ell}}{v_{\ell-1}}{v}y[v] = x[u]^{\top}(\BPM{e_{\ell}}{v}{v_{\ell-1}}\cdots\BPM{e_1}{v_1}{u})^{\top}\innerp^{-1}[v,v]y[v].
	\end{align*}
	Plugging this back into the above gives:
	\begin{align*}
		&= \sum_{u,v\in[n]} \sum_{\substack{u e_1 v_1 e_2 v_2 \\\dots e_{\ell} v \in \CompG_{n}}} x[u]^{\top}(\BPM{e_{\ell}}{v}{v_{\ell-1}}\cdots\BPM{e_1}{v_1}{u})^{\top}\innerp^{-1}[v,v]y[v] \cdot \\
		&(\Ind[e_1\in\bG]-\Pr_{\Null|\btau}[e_1\in\bG])\cdots(\Ind[e_\ell\in\bG]-\Pr_{\Null|\btau}[e_{\ell}\in\bG]) \\
		&= \sum_{u,v\in[n]} \sum_{\substack{u e_1 v_1 e_2 v_2 \\\dots e_{\ell} v \in \CompG_{n}}} (\BPM{e_{\ell}}{v}{v_{\ell-1}}\cdots\BPM{e_1}{v_1}{u} x[u])^{\top}\innerp^{-1}[v,v]y[v] \cdot \\
		&(\Ind[e_1\in\bG]-\Pr_{\Null|\btau}[e_1\in\bG])\cdots(\Ind[e_\ell\in\bG]-\Pr_{\Null|\btau}[e_{\ell}\in\bG]) \\
		&= \langle \CentAdj_{\bG}^{(s)} x, y \rangle_{\innerp}
	\end{align*}
	which proves the claim.
\end{proof}

We first focus on obtaining spectral norm bounds on $\CentAdj_{\bG}^{(s)}$ in the null model.  We obtain these bounds from \pref{thm:main-spectral-bound} so we verify that the matrix $\CentAdj_{\bG}^{(s)}$ indeed meets the hypothesis of the theorem statement.  \pref{cond:self-adj} is satisfied due to \pref{clm:self-adjointness}.  As a consequence of the first part of \pref{clm:transitional-matrix}, all the matrices $M(\Path)$ are Markov transition matrices with an eigenspace projected away, and hence have all their entries bounded by $1$.  Since these matrices have dimension $q\times q$, their operator norm is bounded by some constant $C$ depending only on $q$, and hence \pref{cond:matrix-product-small} is also satisfied.  Next, $\rho(\BlankCol, \GModPar)$ is exactly equal to $\sqrt{\lambda_L}$, which by our assumption is greater than $1$.  Finally, we chose $s$ in the range handled by the theorem statement and thus \pref{thm:main-spectral-bound} implies:
\begin{theorem}
	Suppose $\bG\sim\Null$.  For every constant $\eps > 0$, with probability $1-o(1)$:
	\[
		|\lambda|_{\max}\left(\CentAdj_{\bG}^{(s)}\right) \le ((1+\eps)\sqrt{\lambda_L})^s.
	\]
\end{theorem}
We can choose $\eps$ small enough so that $|\lambda|_{\max}\left(\CentAdj_{\bG}^{(s)}\right)\le\kappa^s$ for $\bG\sim\Null$ whp.

Finally, to prove that there is an eigenvalue greater than $\kappa^s$ when $\bG\sim\Planted$, by \pref{clm:self-adjointness} it suffices to illustrate a vector $x\in\R^{nq}$ such that $\frac{\langle x, \CentAdj_{\bG} x\rangle_{\innerp}}{\langle x,x\rangle_{\innerp}} \ge \kappa^s$.  Then as a direct consequence of \pref{thm:local-statistics-concentrate}:
\begin{theorem}
	Suppose $\bG\sim\Planted$.  There is an absolute constant $C$ such that with probability $1-o(1)$:
	\[
		|\lambda|_{\max}\left(\CentAdj_{\bG}^{(s)}\right) \ge C\lambda_L^s.
	\]
\end{theorem}
Since $s$ is super-constant and $\kappa$ is strictly less than $\lambda_L$, it is indeed true that $|\lambda|_{\max}\left(\CentAdj_{\bG}^{(s)}\right) \ge \kappa^s$ for $\bG\sim\Planted$ whp.  Consequently, we can summarize our main theorem on distinguishing the null distribution from the planted distribution:
\begin{theorem}
	When $\lambda_L > 1$, the task of distinguishing $\Null$ from $\Planted$ with high probability can be done in polynomial time.
\end{theorem}

\subsection{Implementation details}
Our first goal is to explain how to efficiently choose $\kappa$ which strictly between $\sqrt{\lambda_L}$ and $\lambda_L$ when $\lambda_L > 1$.  First note that there is a small enough $\eps$ such that if $\wt{\lambda}_L$ is an additive $\eps$-approximation of $\lambda_L$, then $\frac{\wt{\lambda_L}+\sqrt{\wt{\lambda}_L}}{2}$ lies strictly in between $\sqrt{\lambda_L}$ and $\lambda_L$.  By \pref{lem:spec-neg-type} there is large enough constant $C$ such that $\|L^C(\mathbf{I})\|_F^{1/C}$ is $\eps$-close to $\lambda_L$.  Thus, if our estimator $\wt{\lambda}_L$ is $\|L^{\log n}(\mathbf{I})\|_F^{1/\log n}$, then for large enough $n$, choosing $\kappa$ as $\frac{\wt{\lambda_L}+\sqrt{\wt{\lambda}_L}}{2}$ would give us a number strictly between $\sqrt{\lambda_L}$ and $\lambda_L$.

Our second goal is to explain how to efficiently compute the matrix $\CentAdj_{\bG}^{(s)}$.  Towards doing so, define the \emph{nonbacktracking walk generator matrix} as the matrix with rows and columns indexed by $vev'$ for variable vertices $v$ and $v'$ and constraint vertex $e'$:
\[
	B_{\bG}[(v_1e_1v_2),(v_3e_2v_4)] =
	\begin{cases}
		\BPM{e_2}{v_3}{v_4}\cdot \left(\mathbf{1}[e_2\in\bG] - \Pr_{\bG\sim\Planted|\btau}[e_2\in\bG]\right) &\text{if $e_1\ne e_2$, $v_2=v_3$}\\
		0 &\text{otherwise.}
	\end{cases}
\]
Let $S_{\bG}$ be the matrix with rows indexed by variables in $[n]$ and columns indexed by all $vev'$ where the $(v,vev')$ entry contains $\BPM{e}{v}{v'}\cdot (\mathbf{1}[e\in\bG] - \Pr_{\bG\sim\Planted|\btau}[e\in\bG])$ and the remaining entries contain $0$, and let $T_{\bG}$ be the matrix with rows indexed by all $vev'$ and columns indexed by variables in $[n]$ where the $(vev',v')$ entry is $\mathbf{I}$ and the remaining entries contain $0$.  Then $S_{\bG}B_{\bG}^{s-1}T_{\bG} = \CentAdj_{\bG}^{(s)}$, and it is apparent that the LHS can be computed efficiently.

Finally, since $\CentAdj_{\bG}^{(s)}$ is self-adjoint under $\langle\cdot,\cdot\rangle_{\innerp}$ by \pref{clm:self-adjointness} its largest eigenvalue can be efficiently computed via standard methods such as the power iteration method to a precision necessary for the distinguishing algorithm.


\section{Statistics for the planted model}	\label{sec:statistics}

Consider the planted model $\Model_n=(n,T,\TypeDist,C,\ColDist,\phi)$. Use $\bG=([n],\bE_1,\dots,\bE_F,\btau,\bc)$ to denote a sample from $\Model_{\scp}^{[n]}$.  Recall that in this section our goal is to prove a lower bound on the spectral radius of $\CentAdj_{\bG}^{(s)}$ for $s \le \frac{\log n}{(\log\log n)^2}$ by illustrating a ``witness'' vector with large quadratic form.

We define the \emph{local statistics vector} associated to $\bG$ denoted $\statvec\in\R^{\numcols \cdot n}$ to be the concatenation of vectors $\statvec^{v} = u_{\bc_v}$ for all $v\in[n]$, where $u_{c}\in\R^{q}$ is the indicator of vector of color $c$.  We will shorten $\langle g, g\rangle_{\innerp}$ to $\norm*{\statvec}^2$ in this section.

In this section we prove:
\begin{theorem}\label{thm:local-statistics-concentrate}
	Let $s \le \frac{\log n}{(\log\log n)^2}$.  There exists a constant $\gamma$ such that with probability $1-o_n(1)$:
	\[
		\frac{\left\langle \statvec, \CentAdj_{\bG}^{(s)}\statvec \right\rangle_{\innerp}}{\norm*{\statvec}^2} = \gamma \lambda_L^s.
	\]
\end{theorem}

To prove \pref{thm:local-statistics-concentrate} we will introduce and recall some notation to streamline the proofs.
First recall that we defined $\CompG_n$ in \pref{def:comp-fac-graph} as the instance on variable set $[n]$ with all potential factors that could appear in a graph sampled from $\Model_n$.
Let $\Path$ be some length-$2s$ walk $(v_0\to{e_{0}}\to v_{1} \dots v_{s-1}\to{e_{s-1}}\to v_s)$ in $\CompG_n$ starting and ending at a variable node.  Use $\Path_e$ to denote the factor nodes in $\Path$ and $\Path_v$ to denote the variable nodes in $\partial\Path_e$.\footnote{We would like to stress that $\Path_v$ also includes vertices that are \emph{not} walked on, but are incident to factor nodes which are walked on.}  We use $\Path_v[i]$ to denote $v_i$ and $\Path_e[i]$ to denote $e_i$.  Recall that we use $\theta(e)$ to denote the type of a factor node $e$.  We use $\Chi(\Path)$ to denote the number of excess edges in $\Path$.  Concretely:
\begin{definition}
	$\Chi(\Path) = \abs{\Path_v} - \sum_{j\in[s]}\left(\arity(\FacType(p_e[j])\right)$.
\end{definition}

\begin{definition}
	Let $\NB(s)$ denote the set of all length-$s$ nonbacktracking walks in $\CompG_n$.
\end{definition}

\begin{definition}
	$\bTrans_{\Path} \coloneqq \prod_{j=0}^{s-1} \bTrans_{\FacType(\Path_e[j]),\index(\Path_v[j])\mid \index(p_v[j+1])}$.
\end{definition}

\begin{definition}
	$\Facm_{\Path} \coloneqq \prod_{j=0}^{s-1} \Facm_{\FacType(\Path_e[j]),\index(\Path_v[j])\mid \index(p_v[j+1])}$.
\end{definition}

For $\bG\sim\Model_n$, we define the following notation:
\begin{definition}
	$\expwlk(\bG,\Path) := \bTrans_{\Path} \prod_{j=0}^{s-1}\left(\mathbf{1}[\Path_e[j]\in\bG] - \Pr_{\bG\sim\Planted|\btau}[\Path_e[j]\in\bG]\right)$.
\end{definition}

\begin{definition}
	We will use $\mathrm{wt}(\Path)$ to denote $\prod_{j=0}^{s-1}\left(\mathbf{1}[\Path_e[j]\in\bG] - \Pr_{\bG\sim\Planted|\btau}[\Path_e[j]\in\bG]\right)$.
\end{definition}

\begin{remark}
	In the above language, the centered nonbacktracking power matrix of $\bG$ is:
	\[
		\CentAdj_{\bG}^{(s)}[u,v] = \sum_{p\in\NB(s)|p_v[0]=u,p_v[s]=v} \expwlk(\bG,\Path).
	\]
\end{remark}

Recall that a block $(u,v)$ in the centered nonbacktracking power matrix $\CentAdj_{\bG}^{(s)}$ captures the matrix weight of all length-$2s$ nonbacktracking walks from $u$ to $v$.  Although the paths are nonbacktracking, there could be variables that are visited multiple times and there could also be off-path variables in $p_v$ connected to multiple factors.  It is hard to pinpoint the statistics precisely for this kind of walks, but luckily their contribution is negligible.  Thus, we first remove these ``bad'' walks from $\CentAdj_G^{(s)}$ and analyze the resulting matrix $\SACentAdj_{\bG}^{(s)}$, and then bound the contribution of these bad walks.

We give a formal definition of the ``nice'' walks that are kept in $\SACentAdj_{\bG}^{(s)}$.
\begin{definition}
	A path $\Path$ is \emph{self-avoiding} if 
	$\Chi(\Path) = \abs{\Path_v} - \sum_{j\in[s]}\left(\arity(\Class_{\Path}(j)) - 1\right) = 1$.
\end{definition}

In any self-avoiding path, every variable node in the interior of the path has degree $2$ (i.e. is contained in $2$ factors in $\Path$) and each of the other variable nodes has degree $1$  (i.e. is contained in $1$ factor in $\Path$).

\begin{definition}
	Let $\SA(s)$ denote the set of all length-$2s$ self-avoiding walks in $\CompG_n$.
\end{definition}

\begin{definition}	\label{def:sa-matrix}
	The centered self-avoiding-walk matrix of $\bG$ is $\SACentAdj_{\bG}^{(s)}\in \R^{\numcols n\times \numcols n}$ where the $(u,v)$-th block of the matrix is
	\[
		\SACentAdj_{\bG}^{(s)}[u,v] = \sum_{\Path\in\SA(s)\mid \Path_v[0]=u, \Path_v[s]=v} \expwlk(\bG,\Path).
	\]
\end{definition}

\subsection{Statistics for the centered self-avoiding-walk matrix}


\begin{claim}	\label{clm:sa-local-statistics-path}
For any self-avoiding path $\Path$ in $\CompG_n$, 
\[
	\E_{\bG}\left[\left\langle \statvec^{\Path_v[0]}, \expwlk(\bG,\Path)\statvec^{\Path_v[s]} \right\rangle_{\innerp}\right] = \Pr_{\bG\sim\Model|\btau}[\Path\in\bG]\cdot\Tr\left(\bTrans_{\Path}\bTrans_{\Path}^*\right). 
\]
\end{claim}

\begin{proof}
	The proof is via a chain of equalities.  To lighten notation we use $e_j$ to denote $\Path_e[j]$ and $v_j$ to denote $\Path_v[j]$.  We use $\mathrm{Int}(\Path_v)$ to denote the interior vertices of $\Path$.
	\begin{align*}
		\E_{\bG|\btau}\left[\left\langle \statvec^{v_0}, \expwlk(\bG, \Path) \statvec^{v_s} \right\rangle_{\innerp}\right] &= \sum_{c:\Path_v\to[\numcols]} \prod_{w\in\Path_v} \ColDist_{\btau(w)}(c(w)) \cdot
		\E_{\bG|\btau,\bc} \bigg[\innerp^{-1}[(v_0,c(v_0)),(v_0,c(v_0))] \cdot \\
		&\prod_{j=0}^{s-1} (\Ind[e_j\in\bG]-\Pr_{\bG\sim\Model|\tau}[e_j\in\bG]) \bTrans_{\Path}[c(v_0),c(v_s)] \bigg] \\
		&= \sum_{c:\Path_v\to[q]} \prod_{w\in\Path_v} \ColDist_{\btau(w)}(c(w)) \cdot \prod_{j=0}^{s-1} (\Pr_{\bG\sim\Model|\btau,\bc}[e_j\in\bG] - \Pr_{\bG\sim\Model|\btau}[e_j\in\bG]) \cdot \\
		& \frac{1}{\ColDist_{\btau(v_s)}(c(v_s))} \cdot
		\bTrans_{\Path}^*[c(v_s),c(v_0)] \\
		&= \sum_{c:\Path_v\to[q]} \prod_{w\in\mathrm{Int}(\Path_v)} \frac{1}{\ColDist_{\btau(w)}(c(w))} \cdot
		\prod_{j=0}^{s-1} \Pr_{\bG\sim\Model|\btau}[e_j\in\bG] \cdot \left(\mu_{e_j}(c(\partial e_j)) - \prod_{w\in\partial e_j} \ColDist_{\btau(w)}(c(w)) \right) \cdot \\
		& \frac{1}{\ColDist_{\btau(v_s)}(c(v_s))} \cdot
		\bTrans_{\Path}^*[c(v_s),c(v_0)] \\
		&= \Pr_{\bG\sim\Model|\btau}[p\in\bG] \sum_{c:\{v_0,\dots,v_s\}\to[q]} \prod_{j=0}^{s-1} (\Facm_{e_j,v_j|v_{j+1}}[c(v_j),c(v_{j+1})]-\ColDist_{\btau(v_j)}) \cdot \bTrans_{\Path}^*[c(v_s),c(v_0)] \\
		&= \Pr_{\bG\sim\Model|\btau}[p\in\bG] \sum_{c(v_0),c(v_s)} \bTrans_{\Path}[c(v_0),c(v_s)] \cdot \bTrans_{\Path}^*[c(v_s),c(v_0)] \\
		&= \Pr_{\bG\sim\Model|\btau}[p\in\bG] \cdot \Tr\left(\bTrans_{\Path}\bTrans_{\Path}^*\right)
	\end{align*}
\end{proof}

Now we give precise estimates for the statistics of the centered self-avoiding walk matrix.
\begin{lemma}\label{lem:sa-local-statistics-expectation}
	\[
		\E_{\bG}\left[\left\langle \statvec, \SACentAdj_{\bG}^{(s)}\statvec \right\rangle_{\innerp}\right] = (1-o_n(1))\lambda_{L}^s n.
	\]
\end{lemma}
\begin{proof}
	Expanding out $\SACentAdj_{\bG}^{(s)}$ as a sum by its definition in \pref{def:sa-matrix} gives:
	\begin{align*}
		\E_{\bG}\left[\left\langle \statvec, \SACentAdj_{\bG}^{(s)}\statvec \right\rangle_{\innerp}\right]
		=& \sum_{\Path\in \SA(s)} \E_{\bG}\left[\left\langle \statvec^{\Path_v[0]}, \expwlk(\bG,\Path)  \statvec^{\Path_v[s]}\right\rangle_{\innerp}\right] \\
		=& \sum_{\Path\in \SA(s)} \Pr_{\bG\sim\Model|\btau}[\Path\in\bG] \cdot\Tr\left(\bTrans_{\Path}\bTrans_{\Path}^*\right) \quad(\text{via \pref{clm:sa-local-statistics-path}}) \\
		=& n\cdot(1-o_n(1))\lambda_L^s \quad(\text{by the definition of }\lambda_L).
	\end{align*}
\end{proof}

We next bound the variance of the local statistics.

\begin{lemma}\label{lem:sa-local-statistics-variance}
There is an absolute constant $C$ such that:
\[
	\E_{\bG}\left[\left\langle \statvec, \SACentAdj_{\bG}^{(s)}\statvec \right\rangle_{\innerp}^2\right] -  \E_{\bG}\left[\left\langle \statvec, \SACentAdj_{\bG}^{(s)}\statvec \right\rangle_{\innerp}\right]^2 \leq  n(Cs)^{Cs}.
\]
\end{lemma}

\begin{proof}
If two walks $\Path$ and $\Pathp$ do not share any vertices, then their contribution to $\E_{\bG}\left[\left\langle \statvec, \SACentAdj_{\bG}^{(s)}\statvec \right\rangle_{\innerp}^2\right]$ satisifies 
\begin{align*}
\text{contribution of }\Path,\Pathp
=&\E_{\bG}\left[ \left\langle\statvec^{\Path_v[{0}]}, \expwlk(\bG,\Path)  \statvec^{\Path_v[{s}]}\right\rangle_{\innerp}   \left\langle\statvec^{\Pathp_v[{0}]}, \expwlk(\bG,\Pathp)  \statvec^{\Pathp_v[{s}]}\right\rangle_{\innerp} \right] \\=& \E_{\bG}\left[ \left\langle\statvec^{\Path_v[{0}]}, \expwlk(\bG,\Path)  \statvec^{\Path_v[{s}]}\right\rangle_{\innerp}\right]  \E_{\bG} \left[\left\langle\statvec^{\Pathp_v[{0}]}, \expwlk(\bG,\Pathp)  \statvec^{\Pathp_v[{s}]}\right\rangle_{\innerp} \right].
\end{align*}
So this contribution cancels out with the identical term in $\E_{\bG}\left[\left\langle \statvec, \SACentAdj_{\bG}^{(s)}\statvec \right\rangle_{\innerp}\right]^2$. Thus it suffices for us to consider self-avoiding walks $\Path$ and $\Pathp$ that share some vertices.

We write $\Path \parallel \Pathp$ if they share some variable or factor nodes and the shared nodes have consistent types and use $\Pathu$ to denote the union of the two walks.  Now, note that:
\begin{enumerate}
	\item \label{item:bound-subgraph-prob} Conditioned on $\btau$ and $\bc$ every factor node with arity $k$ is chosen independently with probability at most $\frac{\alpha}{n^{k-1}}$ for some constant $\alpha$.  Thus, for any subgraph of $\CompG_n$ on $e$ edges and $r$ factor nodes, the probability of it occurring in $\bG$ is at most $\alpha^{r} n^{e-r}$.
	\item \label{item:matrix-weight-entry-bound} The matrix weight $\bTrans_{\Path}$ of any self-avoiding path has entries bounded in magnitude by $1$ since it is a product of projected stochastic matrices.
\end{enumerate}
Using the above facts, a straightforward calculation tells us:
\[
	\left|\E_{\bG}\left[\left\langle \statvec^{\Path_v[0]}, \expwlk(\bG,\Path)\statvec^{\Path_v[s]}\right\rangle_{\innerp} \left\langle \statvec^{\Pathp_v[0]}, \expwlk(\bG,\Pathp)\statvec^{\Pathp_v[s]}\right\rangle_{\innerp}\right]\right| \le
	\left(\alpha'\right)^s n^{-(\sum_{e\in \Pathu_e}\arity(e)-1)}. \numberthis \label{eq:bound-prod-quadform}
\]
We say $\Pathu_1\sim\Pathu_2$ if the subgraphs induced by them are isomorphic.  $\sim$ partitions the space of all $\Pathu$ into equivalence classes.  We use $[\Pathu]$ to denote the equivalence class of $\Pathu$.  The number of equivalence classes can be bounded by $(C's)^{C's}$ for some constant $C'>1$ (since the graph representing the equivalence class of $\Path\cup\Pathp$ is on $O(s)$ vertices can be specified by a list of $O(s)$  edges).  Due to the shared vertices, $\Pathu$ is connected and $\Chi(\Pathu)\le 1$.  Thus, we now bound the variance as follows:
\begin{align*}
	&\E_{\bG}\left[\left\langle \statvec, \SACentAdj_{\bG}^{(s)}\statvec \right\rangle_{\innerp}^2\right] -  \E_{\bG}\left[\left\langle \statvec, \SACentAdj_{\bG}^{(s)}\statvec \right\rangle_{\innerp}\right]^2 \\
	=&\sum_{\Path \parallel \Pathp} \E_{\bG}\left[ \left\langle\statvec^{\Path_v[{0}]}, \expwlk(\bG,\Path)  \statvec^{\Path_v[{s}]}\right\rangle_{\innerp}   \left\langle\statvec^{\Pathp_v[{0}]}, \expwlk(\bG,\Pathp)  \statvec^{\Pathp_v[{s}]}\right\rangle_{\innerp} \right] \\
	\le& \sum_{\Pathu}\sum_{\Path\parallel\Pathp:\Path\cup\Pathp=\Pathu} (\alpha')^s n^{-(\sum_{e\in \Pathu_e}\arity(e)-1)} \\
	\le& \sum_{[\Pathu]} \sum_{\Pathu_{i}\in[\Pathu]} (9\alpha's)^s  n^{-(\sum_{e\in \Pathu_e}\arity(e)-1)} \\
	=& \sum_{[\Pathu]} (9\alpha's)^s n^{|\Pathu_v|-(\sum_{e\in \Pathu_e}\arity(e)-1)} \\
	=& \sum_{[\Pathu]} (9\alpha's)^s n^{\Chi(\Pathu)} \\
	\le& (9C'\alpha's^2)^{C's} n. 
\end{align*}
Thus, the claim follows.
\end{proof}

\subsection{Bounding contribution of non-self-avoiding walks}
In comparison to the self-avoiding walks, the non-self-avoiding walks have negligible contributions to the expectation and the variance of the statistics. We prove this statement using the following claim.

\begin{claim}\label{clm:nsa-bounded-contribution}
We can bound the statistics of non-self-avoiding walks as follows:
\[
	\left|\E_{\bG}\left[\sum_{\Path~\text{non-self-avoiding}}\left\langle \statvec^{\Path_v[{0}]}, \expwlk(\bG,\Path)\statvec^{\Path_v[{s}]} \right\rangle_{\innerp}\right]\right| \le (Cs)^{Cs}.
\]
\[
	\E_{\bG}\left[\sum_{\substack{\Path\parallel\Pathp \\ \text{$\Path$ or $\Pathp$ non-self-avoiding}}}\left\langle \statvec^{\Path_v[{0}]}, \expwlk(\bG,\Path)\statvec^{\Path_v[{s}]} \right\rangle_{\innerp}\cdot \left\langle \statvec^{\Pathp_v[{0}]}, \expwlk(\bG,\Pathp)\statvec^{\Pathp_v[{s}]} \right\rangle_{\innerp}\right] \le (Cs)^{Cs}
\]
for some absolute constant $C$.
\end{claim}
\begin{proof}
The first part is derived by applying \pref{item:bound-subgraph-prob} in the proof of \pref{lem:sa-local-statistics-variance}.
\begin{align*}
	\left|\E_{\bG}\left[\sum_{\Path~\text{non-self-avoiding}}\left \langle \statvec^{\Path_v[0]}, \expwlk(\bG,\Path)\statvec^{\Path_v[{s}]} \right\rangle_{\innerp}\right]\right|
\le& \sum_{\Path~\text{non-self-avoiding}} \left|\E_{\bG}\left[\left \langle \statvec^{\Path_v[0]}, \expwlk(\bG,\Path)\statvec^{\Path_v[{s}]} \right\rangle_{\innerp}\right]\right| \\
\le& \sum_{\Path~\text{non-self-avoiding}} \Pr_{\bG}[\Path\in\bG] \cdot O(1) \\
\le& \sum_{\Path~\text{non-self-avoiding}} (\alpha')^{s} n^{\Chi(\Path)-\abs{\Path_v}} \\
\le& \sum_{\Path~\text{non-self-avoiding}} (\alpha')^s n^{-\abs{\Path_v}}\\
\le&\, (Cs)^{Cs}.
\end{align*}
where the equality from the second to third line is a consequence of \pref{item:bound-subgraph-prob} in the proof of \pref{lem:local-statistics-variance}, and the last inequality is due to $\Chi(p)\le 0$ for a non-self-avoiding walk.

The second expression is equal to:
\begin{align*}
	\sum_{\substack{\Path\parallel\Pathp \\ \text{$\Path$ or $\Pathp$ non-self-avoiding}}} \E_{\bG}\left[\left\langle \statvec^{\Path_v[{0}]}, \expwlk(\bG,\Path)\statvec^{\Path_v[{s}]} \right\rangle_{\innerp}\cdot \left\langle \statvec^{\Pathp_v[{0}]}, \expwlk(\bG,\Pathp)\statvec^{\Pathp_v[{s}]} \right\rangle_{\innerp}\right].
\end{align*}
By \pref{eq:bound-prod-quadform} the above can be bounded by:
\begin{align*}
	\sum_{\substack{\Path\parallel\Pathp \\ \text{$\Path$ or $\Pathp$ non-self-avoiding}}} \left(\alpha'\right)^s n^{\Chi(\Pathu)-|\Pathu_v|}
\end{align*}
where $\Pathu$, recall, is the union of $\Path$ and $\Pathp$.
Since $\Path$ and $\Pathp$ share vertices, $\Chi(\Pathu)\le\min\{\Chi(\Path),\Chi(\Pathp)\}$, and since at least one of the two walks is non-self-avoiding, $\Chi(\Pathu)\le 0$.  This lets us bound the above by:
\begin{align*}
	\sum_{\substack{\Path\parallel\Pathp \\ \text{$\Path$ or $\Pathp$ non-self-avoiding}}} \left(\alpha'\right)^s n^{-|\Pathu_v|} &\le \sum_{\Pathu} \sum_{\Path\parallel\Pathp:\Path\cup\Pathp=\Pathu} \left(\alpha'\right)^s n^{-|\Pathu_v|} \\
	&= \sum_{[\Pathu]} \sum_{\Pathu_i\in[\Pathu]} (9\alpha's)^s n^{-|\Pathu_v|} \\
	&= \sum_{[\Pathu]} (9\alpha's)^s\\
	&\le (9C'\alpha's^2)^{C's}
\end{align*}
which gives us the desired bound in the second part of the statement.
\end{proof}

\subsection{Wrapping up estimates of statistics}

The following claim about the statistics of the centered non-backtracking-walk matrix $\CentAdj_{\bG}^{(s)}$ is an immediate consequence of combining \pref{lem:sa-local-statistics-expectation} and the first part of \pref{clm:nsa-bounded-contribution}.
\begin{lemma}\label{lem:local-statistics-expectation}
For $s \le \frac{\log n}{(\log\log n)^2}$:
\[
	\E_{\bG}\left[\left\langle \statvec, \CentAdj_{\bG}^{(s)}\statvec \right\rangle_{\innerp}\right] = (1\pm o_n(1))\lambda_{L}^s n.
\]
\end{lemma}

\begin{lemma}\label{lem:local-statistics-variance}
For some absolute constant $C$:
\[
	\E_{\bG}\left[\left\langle \statvec, \CentAdj_{\bG}^{(s)}\statvec \right\rangle_{\innerp}^2\right] -  \E_{\bG}\left[\left\langle \statvec, \CentAdj_{\bG}^{(s)}\statvec \right\rangle_{\innerp} \right]^2 \leq n(Cs)^{Cs}.
\]
\end{lemma}

\begin{proof}
Using the observation that only pairs of walks that share some vertices contribute to the variance, we have:
\begin{align*}
 	&\E_{\bG}\left[\left\langle \statvec, \CentAdj_{\bG}^{(s)}\statvec \right\rangle_{\innerp}^2\right] -  \E_{\bG}\left[\left\langle \statvec, \CentAdj_{\bG}^{(s)}\statvec \right\rangle_{\innerp}\right]^2 \\
	=& \, \E_{\bG}\left[\left\langle \statvec, \SACentAdj_{\bG}^{(s)}\statvec \right\rangle_{\innerp}^2\right] -  \E_{\bG}\left[\left\langle \statvec, \SACentAdj_{\bG}^{(s)}\statvec \right\rangle_{\innerp}\right]^2  \\
 	&+ \sum_{\substack{\Path \parallel \Pathp \\ \text{$\Path$ or $\Pathp$ non-self-avoiding}}} \E_{\bG}\left[\left\langle \statvec^{\lab(\Path_v^{0})}, \expwlk(\bG,\Path,\lab)\statvec^{\lab(\Path_v^{s})} \right\rangle_{\innerp}\cdot \left\langle \statvec^{\lab(\Pathp_v^{0})}, \expwlk(\bG,\Pathp,\labp)\statvec^{\labp(\Pathp_v^{s})} \right\rangle_{\innerp}\right]
\end{align*}
We can conclude the desired bound immediately from \pref{lem:sa-local-statistics-variance} and the second part of \pref{clm:nsa-bounded-contribution}.
\end{proof}

Finally, to establish \pref{thm:local-statistics-concentrate} we first use Chebyshev's inequality to conclude that when $s \le \frac{\log n}{(\log\log n)^2}$,
\[
	\left\langle \statvec, \CentAdj_{\bG}^{(s)}\statvec \right\rangle_{\innerp} = (1\pm o_n(1))\lambda_L^s n
\]
with probability $1-o_n(1)$.  Now, since the $\|\statvec\|^2$ is $(1\pm o_n(1))\frac{n}{\gamma}$ with probability $1-o_n(1)$ for some constant $\gamma$, we can conclude that with probability $1-o_n(1)$,
\[
	\frac{\left\langle \statvec, \CentAdj_{\bG}^{(s)}\statvec \right\rangle_{\innerp}}{\|\statvec\|^2} = (1\pm o_n(1)) \gamma \lambda_L^s
\]
thereby finishing the proof.

\section{Eigenvalue bounds}	\label{sec:eig-bounds}

In this section we show an eigenvalue upper bound for the centered nonbacktracking-walk matrix in the null model. We first describe the matrix distribution in detail. 

Recall the definition of a null model $\Model^{\times}$ (\pref{def:null-model}). Let $\bH\coloneqq \bigcup_{i=1}^{F}\bE_i$ be an observation sampled from $\Model^{\times}$. As discussed in \pref{sec:notations}, the observation $\bH$ has an associated bipartite graph which we denote $\Bip(\bH)$. The left vertex set is given by the variables $[n]$ and the right vertex set is given by the factors $\gamma\in \bH$. We will use $L(\bH)$ and $R(\bH)$ to denote the left and right vertex sets of $\Bip(\bH)$.

Next associate with each triple $v\gamma u$, where $v,u\in\gamma$ and $v\neq u$, a $q\times q$ matrix $M_{v\gamma u}$. Like before, the value of the matrix only depends on the factor type $\FacType(\gamma)$ and the two variables' indices $\index(v),\index(u)$ in $\gamma$. We use $\MatCol$ to denote the collection of these $q\times q$ matrices $\{M_{a\phi_ib}\}_{i\in[F],a\neq b\in[\arity(i)]}$. Now we are ready to define the matrix distribution.

\begin{definition}
	The matrix distribution is defined as follows. First sample an observation $H$ from some null model $\Model^{\times}$.  We define the \emph{length-$\ell$ $\Model^{\times}$-centered nonbacktracking power} $\CentAdj_{\bH}^{(\ell)}$ is the $n\times n$ block matrix where the $(i,j)$-block as the following $q\times q$ matrix:
	\[
		\CentAdj_{\bH}^{(\ell)}[i,j] \coloneqq \sum_{\substack{(v_0\gamma_1v_1\dots v_{\ell-1}\gamma_{\ell}v_{\ell}) \\ \in\NB(\CompG_{n},\ell,i,j)}}  \prod_{t=1}^{\ell}M_{v_{2t-2}\gamma_t v_{2t-1}}\left(\Ind[\gamma_t \in \bH]-\Pr_{\Model^{\times}}[\gamma_t \in\bH]\right),
	\]
	where $\NB(\CompG_{n},\ell,i,j)$ denote the set of all length-$2\ell$ nonbacktracking walks in the complete bipartite factor graph $\Bip(\CompG_{n})$ starting at variable $i$ and ending at variable $j$.
\end{definition}

We are interested in obtaining a high probability upper bound on $\|\CentAdj_{\bH}^{(\ell)}\|$ in terms of a particular quantity depending on $\MatCol$ and $\Model^{\times}$, which we denote by $\rho(\MatCol,\Model^{\times})$. Before giving its definition, we simplify the notation a bit:

\begin{definition}	\label{def:walk-matrix}
	Given a length-$2\ell$ nonbacktracking walk $W = v_0\gamma_1v_1\gamma_2v_2\dots v_{\ell-1}\gamma_{\ell}v_{\ell}$ in $\Bip(\bH)$, we define the weight of the walk with $M_W$ to be
	\[
		M_W \coloneqq M_{v_0\gamma_1v_1}M_{v_1\gamma_2v_2}\dots M_{v_{\ell-1}\gamma_{\ell}v_{\ell}}.
	\]
	We use $\gamma_i(W)$ to denote the $i$-th factor visited by $W$.
\end{definition}

We now define $\rho(\MatCol,\Model^{\times})$.

\begin{definition} 
	For a positive integer $m$, construct a length-$m$ path $P$ with random matrix weights in the following way.  Start with a vertex $v_0$ and assign to it a random label $\btau(v_0)$ sampled from $\TypeDist$.  We iteratively construct a path with edges weighed by matrices until its length is equal to $m$.  Suppose we already have a path $v_0 v_1 \dots v_{t}$ where each $v_i$ has a label $\btau(v_i)$, and for an edge $\{v_i,v_{i+1}\}$ its two directed edges $(v_i, v_{i+1})$ and $(v_{i+1},v_i)$ have matrix weights $\bW_{i,i+1}$ and $\bW_{i+1,i}$ respectively.  To grow this path, we sample $(s,a)$ where $s\in[F]$ and $a\in[\arity(s)]$ with probability proportional to $\ACFacDeg{s}\cdot\Ind[\Cl(s)_a=\btau(v_t)]\cdot\prod_{j=1,j\ne a}^t \TypeDist(\Cl(s)_j)$, followed by a uniformly random $b$ in $[\arity(s)]\setminus\{a\}$.  Add vertex $v_{t+1}$ and set $\btau(v_{t+1})=\Cl(s)_b$.  Then add edge $\{v_t,v_{t+1}\}$ and let the matrix weight of the directed edge $(v_t,v_{t+1})$ be $\bW_{t,t+1} = M_{a\phi_sb}\in \MatCol$, and matrix weight of directed edge $(v_{t+1},v_t)$ be $\bW_{t+1,t} = M_{b\phi_sa}\in \MatCol$.\footnote{For the sake of intuition the reader should think of the distribution of $v_{t+1}$ as first sampling $\bH\sim \GModPar$, then choosing a random vertex $v$ with label $\btau(v_t)$ and finally choosing as a random neighbor $w$ of $v$ within $\bH$.  $\btau(v_{t+1})$ is then set to the label of $w$ and the matrix weight on edge $(v_t,v_{t+1})$ is chosen as the matrix weight on $(v,w)$.}
	We then define $w_{m}$ as:
	\[
		w_{m} \coloneqq \E\Tr(\bW_{0,1}\bW_{1,2}\dots\bW_{m-1,m}\bW_{m,m-1}\dots \bW_{2,1}\bW_{1,0}).
	\]
	Now define $r(\MatCol,\Model^{\times})$ as
	\[
		r(\MatCol,\Model^{\times}) \coloneqq  \limsup_{m\to\infty} w_m^{\frac{1}{2m}}.
	\]
	Next, define $d(\Model^{\times})$ (which, intuitively, is the average degree of a vertex in a sample from $\Model^{\times}$) as
	\[
		d(\Model^{\times}) \coloneqq \sum_{t=1}^T \TypeDist(t) \sum_{i=1}^F \sum_{j=1}^{k_i} \ACFacDeg{i}\Ind[\Cl(i)_j=t].
	\]
	Finally, we define $\rho(\MatCol,\Model^{\times})$ as
	\[
		\rho(\MatCol,\Model^{\times}) \coloneqq r(\MatCol,\Model^{\times})\sqrt{d(\Model^{\times})}.
	\]
\end{definition}

We remark that if the weight collection $\BlankCol$ is defined such that $M_{b\phi_i a} = \bTrans_{i,a \mid b}$ (defined in \pref{eq:bTrans}) for all $i\in[F], a\neq b\in[\arity(i)]$, then $\rho(\MatCol,\Model^{\times}) = \sqrt{\lambda_L}$.

The main result of this section is that $\|\CentAdj_{\bH}^{(\ell)}\| \le \left((1+o_n(1))\rho(\BlankCol,\GModPar)\right)^{\ell}$ for a wide range of $\ell$ when $\bH\sim \GModPar$.
\begin{theorem}	\label{thm:main-spectral-bound}
	Suppose:
	\begin{enumerate}
		\item \label{cond:self-adj} There is an inner product $\langle\cdot,\cdot\rangle_{\nu}$ on $\R^{nq}$ such that for every right vertex $\gamma=(v_1,\dots,v_{\arity(i)})$ in $\CompG_{\BlankCol, n}$, for all $1\le a,b\le \arity(i)$, the $nq\times nq$ matrix obtained by placing $M_{v_a\gamma v_b}$ in the $(v_a,v_b)$ block and zeros everywhere else is the adjoint of the $nq\times nq$ matrix obtained by placing $M_{v_b\gamma v_a}$ in the $(v_b,v_a)$ block and zeroes everywhere else under $\langle\cdot,\cdot\rangle_{\nu}$.
		\item \label{cond:matrix-product-small} There is a constant $C\ge 1$ such that the weight $M_W$ every nonbacktracking walk $W$ in $\CompG_{n}$ satisfies:
		\[
			\|M_W\| \le C
		\]
		where $\|\cdot\|$ is the operator norm induced by $\langle\cdot,\cdot\rangle$.
		\item \label{cond:inf-spec-large} $\rho(\BlankCol,\GModPar) \ge 1$.
	\end{enumerate}
	Then for every $\eps > 0$ and $(\log\log n)^2\le \ell \le \frac{\log n}{(\log\log n)^2}$, with probability $1-o_n(1)$:
	\[
		|\lambda|_{\max}\left(\CentAdj_{\bH}^{(\ell)}\right) \le \left((1+\eps)\rho(\BlankCol,\GModPar)\right)^{\ell}.
	\]
\end{theorem}


\subsection{Proof of \pref{thm:main-spectral-bound}}
The proof of \pref{thm:main-spectral-bound} is via the \emph{trace method}.  One preliminary observation is that \pref{cond:self-adj} implies that $\CentAdj_{\bH}^{(\ell)}$ is self-adjoint and hence all its eigenvalues are real.  Consequently, for any positive even integer $k$:
\[
	\|\CentAdj_{\bH}^{(\ell)}\|^k \le \Tr\left(\left(\CentAdj_{\bH}^{(\ell)}\right)^k\right).	\numberthis \label{eq:trace-bounds-op-norm}
\]
Our goal is now to obtain a handle on $\bS\coloneqq\Tr\left(\left(\CentAdj_{\bH}^{(\ell)}\right)^k\right)$ and obtain a high probability bound on it.  We borrow some terminology from \cite{MOP20a}:
\begin{definition}[Linkages]
	A $(k\times 2\ell)$-nonbacktracking $\BlankCol$-linkage is a length-$2k\ell$ closed walk in $\Bip(\CompG_{n})$ that starts and ends in the left vertex set and can be expressed as a concatenation of $k$ nonbacktracking walks of length-$2\ell$ each.  Each length-$2\ell$ nonbacktracking segment is called a \emph{link}.  We use $\Lkgs(\BlankCol, n, k, \ell)$ to denote the collection of all $(k\times 2\ell)$-nonbacktracking $\BlankCol$-linkages.
\end{definition}

\begin{definition}
	Given a $(k\times 2\ell)$-nonbacktracking $\BlankCol$-linkage $W$, we use $L(W)$ to denote the set of left vertices visited by $W$, $R(W)$ to denote the set of right vertices visited by $W$, $V(W)$ to denote $L(W)\cup R(W)$, $E(W)$ to denote the set of edges visited by $W$, and $G(W)$ to denote the graph $(V(W),E(W))$ induced by $W$.
\end{definition}

With the above terminology and notation in hand, we can write $\bS$ as:
\[
	\bS = \sum_{W\in\Lkgs(\BlankCol,n,k,\ell)} \Tr(M_W) \prod_{t=1}^{k\ell}\left(\Ind[\gamma_t(W)\in \bH]-\E_{\bH|\btau}\Ind[\gamma_t(W)\in\bH]\right).
\]
A natural strategy to obtaining a high probability bound on $\bS$ is to bound $\E[\bS]$ by some $Z$ and use Markov's inequality to conclude that $\bS$ is bounded by, say, $nZ$ with high probability.  However, $\E[\bS]$ is not as small as we would hope due to blowing up in magnitude owing to the occurrences of certain rare and problematic subgraphs in $\Bip(\bH)$.  So this suggests a natural tweak of conditioning away these rare subgraphs and trying to carry out the same strategy.  This tweak is an idea that occurs in many previous papers in the line of work on getting eigenvalue bounds on sparse random matrices \cite{Fri03,BLM15,Bor19,BC19,MOP20a}.  These problematic subgraphs all share one common trait -- having multiple cycles in a small neighborhood.
\begin{definition}	\label{def:bicycle}
	We say a graph $\Gamma$ is \emph{$r$-bicycle free} if for every vertex $v$, the radius-$r$ ball around $v$ contains at most one cycle.  We say $\Gamma$ is an $r$-bicycle if it has at most $r$ edges and has at least two cycles, and we say $\Gamma$ is an $r$-bicycle frame if it is an $r$-bicycle such that no subgraph of it is an $r$-bicycle.
\end{definition}

\begin{lemma}	\label{lem:is-r-bike-free}
	With probability $1-o_n(1)$, $\Bip(\bH)$ is $r$-bicycle free for $r = \frac{\log n}{\log \log n}$.
\end{lemma}
We refer the reader to \pref{cor:bike-free} for a proof of this fact.

Henceforth, we use $\nobikes$ to denote the event that $\bH$ is $r$-bicycle free for $r = \frac{\log n}{\log\log n}$.  Now define $\bU\coloneqq \bS \cdot \Ind[\nobikes]$.  By \pref{lem:is-r-bike-free} with probability $1-o_n(1)$, $\bS = \bU$ so if we can prove that $\bU\le Z$ with probability $1-o_n(1)$, we can also show that $\bS \le Z$ with probability $1-o_n(1)$.  Thus, we turn our attention to bounding $\E[\bU]$.
\[
	\bU = \sum_{W\in\Lkgs(\BlankCol,n,k,\ell)} \Tr(M_W) \prod_{t=1}^{k\ell}\left(\Ind[\gamma_t(W)\in \bH]-\E_{\bH|\btau}\Ind[\gamma_t(W)\in\bH]\right)\Ind[\nobikes].	\numberthis \label{eq:U-formula}
\]
We now study the quantity $\prod_{t=1}^{k\ell}\left(\Ind[\gamma_t(W)\in \bH]-\E_{\bH|\btau}\Ind[\gamma_t(W)\in\bH]\right)$ more closely.
\begin{definition}
	For a given right vertex $\gamma$ of $\CompG_{n}$ the \emph{multiplicity} $m_W(\gamma)$ of $\gamma$ in $W$ is the number of times $\gamma$ is visited by $W$.  $\Sing(W)$ denotes the set of all right vertices that are visited exactly once and are called \emph{singleton right vertices}.  $\Dup(W)$ denotes the set of all right vertices that are visited more than once, and are called \emph{duplicative right vertices}.
\end{definition}
Henceforth, we shorten $\Ind[\gamma\in\bH]$ to $\Ind_{\gamma}$ and $\E_{\bH|\btau}\Ind[\gamma\in\bH]$ to $\Mean_{\gamma}$.  Thus:
\begin{align*}
	\prod_{t=1}^{k\ell}\left(\Ind_{\gamma_t(W)}-\Mean_{\gamma_t(W)}\right) &= \prod_{\gamma\in \Sing(W)}(\Ind_{\gamma}-\Mean_{\gamma}) \prod_{\gamma\in \Dup(W)}(\Ind_{\gamma}-\mu_{\gamma})^{m_W(\gamma)} \\
	&= \prod_{\gamma\in \Sing(W)}(\Ind_{\gamma}-\Mean_{\gamma}) \prod_{\gamma\in \Dup(W)}\left(\Ind_{\gamma} \cdot \sum_{i=1}^{m_W(\gamma)}(-\Mean_{\gamma})^{m_{W}(\gamma)-i}{m_W(\gamma)\choose i}+(-\Mean_{\gamma})^{m_W(\gamma)}\right)\\
	\intertext{To lighten notation, we use $\alpha_{\gamma}$ to denote $\sum_{i=1}^{m_W(\gamma)}(-\Mean_{\gamma})^{m_W(\gamma)-i}{m_W(\gamma)\choose i}$.}
	&= \prod_{\gamma\in \Sing(W)} (\Ind_{\gamma}-\Mean_{\gamma}) \prod_{\gamma\in \Dup(W)} (\Ind_{\gamma}\alpha_{\gamma} + (-\Mean_{\gamma})^{m_W(\gamma)})\\
	&= \prod_{\gamma\in \Sing(W)} (\Ind_{\gamma}-\Mean_{\gamma}) \sum_{L\subseteq \Dup(W)} \prod_{\gamma\in L}\Ind_{\gamma}\alpha_{\gamma} \prod_{\gamma\in \Dup(W)\setminus L} (-\Mean_{\gamma})^{m_W(\gamma)}\\
	&= \sum_{L\subseteq \Dup(W)} \prod_{\gamma\in L} \alpha_{\gamma} \prod_{\gamma\in \Dup(W)\setminus L} (-\Mean_{\gamma})^{m_W(\gamma)}\prod_{\gamma\in \Sing(W)}(\Ind_{\gamma}-\Mean_{\gamma}) \prod_{\gamma\in L}\Ind_{\gamma}. \numberthis \label{eq:scalar-weight}
\end{align*}
Plugging in \pref{eq:scalar-weight} into \pref{eq:U-formula} gives:
\begin{align*}
	\bU = \sum_{W\in\Lkgs(\BlankCol,n,k,\ell)} \Tr(M_W) \sum_{L\subseteq \Dup(W)} \prod_{\gamma\in L} \alpha_{\gamma} \prod_{\gamma\in \Dup(W)\setminus L} (-\Mean_{\gamma})^{m_W(\gamma)}\prod_{\gamma\in \Sing(W)}(\Ind_{\gamma}-\Mean_{\gamma})\prod_{\gamma\in L}\Ind_{\gamma} \Ind[\nobikes]
\end{align*}
We are interested in understanding $\E[\bU]$.  Note that this is equal to $\E_{\btau}\E_{\bH|\btau}[\bU]$.  We will first focus our attention on understanding $\E_{\bH|\btau}[\bU]$.  We have:
\begin{align*}
	\E_{\bH|\btau}[\bU] &= \sum_{W\in\Lkgs(\BlankCol,n,k,\ell)} \Tr(M_W) \sum_{L\subseteq \Dup(W)} \prod_{\gamma\in L} \alpha_{\gamma} \prod_{\gamma\in \Dup(W)\setminus L} (-\Mean_{\gamma})^{m_W(\gamma)}\E_{\bH|\btau}\left[\prod_{\gamma\in \Sing(W)}(\Ind_{\gamma}-\Mean_{\gamma})\prod_{\gamma\in L}\Ind_{\gamma} \Ind[\nobikes]\right]\\
	&\le \sum_{W\in\Lkgs(\BlankCol,n,k,\ell)} \Tr(M_W) \sum_{L\subseteq \Dup(W)} \prod_{\gamma\in L}|\alpha_{\gamma}| \prod_{\gamma\in \Dup(W)\setminus L} \Mean_{\gamma}^{m_W(\gamma)} \left|\E_{\bH|\btau}\left[\prod_{\gamma\in \Sing(W)}(\Ind_{\gamma}-\Mean_{\gamma})\prod_{\gamma\in L}\Ind_{\gamma} \Ind[\nobikes]\right]\right|.
\end{align*}
Notice that $\alpha_{\gamma} = (1-\Mean_{\gamma})^{m_W(\gamma)} - (-\Mean_{\gamma})^{m_W(\gamma)}$ and hence $|\alpha_{\gamma}|\le (1-\Mean_{\gamma})^{m_W(\gamma)} + \Mean_{\gamma}^{m_W(\gamma)}\le (1-\Mean_{\gamma})+\Mean_{\gamma} = 1$ where the second inequality is a consequence of $\Mean_{\gamma}\in[0,1]$.  The result is:
\[
	\E_{\bH|\btau}[\bU] \le \sum_{W\in\Lkgs(\BlankCol,n,k,\ell)} \Tr(M_W) \sum_{L\subseteq \Dup(W)} \prod_{\gamma\in \Dup(W)\setminus L} \Mean_{\gamma}^{m_W(\gamma)} \left|\E_{\bH|\btau}\left[\prod_{\gamma\in \Sing(W)}(\Ind_{\gamma}-\Mean_{\gamma})\prod_{\gamma\in L}\Ind_{\gamma} \Ind[\nobikes]\right]\right|.	\numberthis \label{eq:red-to-singleton-bound}
\]
Next, we would like to obtain a bound on $\left|\E_{\bH|\btau}\left[\prod_{\gamma\in \Sing(W)}(\Ind_{\gamma}-\Mean_{\gamma})\prod_{\gamma\in L}\Ind_{\gamma} \Ind[\nobikes]\right]\right|$.

Towards doing so, we first set up a couple of definitions and an observation.
\begin{definition}[Closure of subgraph]
	Given a subgraph $\Gamma$ of $\Bip(\CompG_{n})$ with right vertex set $R(\Gamma)$, we define its \emph{closure} $\Clos(\Gamma)$ as the induced subgraph on the vertex set $\left(\bigcup_{\gamma\in R(\Gamma)}N(\gamma)\right) \bigcup R(\Gamma)$.  We say $\Gamma$ is \emph{closed} if $\Clos(\Gamma) = \Gamma$.
\end{definition}

\begin{definition}[Excess]
	Given a graph $\Gamma$ on $e$ edges, $v$ vertices and $c$ connected components, we define the \emph{excess} of $\Gamma$, denoted $\Exc(\Gamma)$, to be $e-v+c$.
\end{definition}

The following is immediate from the observation that the excess of a graph cannot decrease on adding a new vertex or a new edge.
\begin{lemma}	\label{lem:exc-cant-grow}
	If $\Gamma=(V,E)$ and $\Gamma'=(V',E')$ are two graphs such that $\Gamma$ is a subgraph of $\Gamma'$, i.e. $V\subseteq V'$ and $E\subseteq E'$, then $\Exc(\Gamma')\ge\Exc(\Gamma)$.
\end{lemma}
If $\Clos(L)$ is not $r$-bicycle free, then $\prod_{\gamma\in L}\Ind_{\gamma}\Ind[\nobikes]$ is equal to $0$.  Otherwise \pref{lem:singleton-small-exp} then shows that:
\[
	\left|\E_{\bH|\btau}\left[\prod_{\gamma\in \Sing(W)}(\Ind_{\gamma}-\Mean_{\gamma})\prod_{\gamma\in L}\Ind_{\gamma} \Ind[\nobikes]\right]\right| \le \prod_{\gamma\in S\cup L}\Mean_{\gamma} \cdot 2^{|\Sing(W)|} \left(\frac{1}{n^{.5}}\right)^{\frac{|\Sing(W)|}{r}-\Exc(\Clos(\Sing(W)\cup L))}.
\]
Plugging the above into \pref{eq:red-to-singleton-bound} tells us:
\[
	\E_{\bH|\btau}[\bU] \le \sum_{W\in\Lkgs(\BlankCol,n,k,\ell)} \Tr(M_W) \sum_{\substack{L\subseteq \Dup(W) \\ \Clos(L)~r\text{-bicycle free}}} \prod_{\gamma\in R(W)}\Mean_{\gamma} \prod_{\gamma\in \Dup(W)\setminus L} \Mean_{\gamma}^{m_W(\gamma)-1} 2^{|\Sing(W)|} \left(\frac{1}{n^{.5}}\right)^{\frac{|\Sing(W)|}{r}-\Exc(\Clos(\Sing(W)\cup L))}.
\]
 Henceforth, we will shorten $\Exc(\Clos(\Sing(W)\cup \Dup(W)))$ to $\Exc_W$ for simplicity of notation.  Since $\Clos(\Sing(W)\cup L)$ is a subgraph of $\Clos(\Sing(W)\cup \Dup(W))$, by \pref{lem:exc-cant-grow} we have $\Exc_W\ge \Exc(\Clos(\Sing(W)\cup L))$, which means:
\begin{align*}
	\E_{\bH|\btau}[\bU] &\le \sum_{W\in\Lkgs(\BlankCol,n,k,\ell)} \Tr(M_W) \sum_{\substack{L\subseteq \Dup(W) \\ \Clos(L)~r\text{-bicycle free}}} \prod_{\gamma\in R(W)}\Mean_{\gamma} \prod_{\gamma\in \Dup(W)\setminus L} \Mean_{\gamma}^{m_W(\gamma)-1} 2^{|\Sing(W)|} \left(\frac{1}{n^{.5}}\right)^{\frac{|\Sing(W)|}{r}-\Exc_W} \\
	&= \sum_{W\in\Lkgs(\BlankCol,n,k,\ell)} \Tr(M_W) \cdot 2^{|\Sing(W)|} \left(\frac{1}{n^{.5}}\right)^{\frac{|\Sing(W)|}{r}-\Exc_W} \prod_{\gamma\in R(W)}\Mean_{\gamma} \sum_{\substack{L\subseteq \Dup(W) \\ \Clos(L)~r\text{-bicycle free}}} \prod_{\gamma\in \Dup(W)\setminus L} \Mean_{\gamma}^{m_W(\gamma)-1}	\numberthis \label{eq:about-to-bound-profligate}
\end{align*}
Next, we focus on bounding $\sum_{\substack{L\subseteq \Dup(W) \\ L~r\text{-bicycle free}}} \prod_{\gamma\in \Dup(W)\setminus L} \Mean_{\gamma}^{m_W(\gamma)-1}$.  For starters, observe that:
\[
	\sum_{\substack{L\subseteq \Dup(W) \\ \Clos(L)~r\text{-bicycle free}}} \prod_{\gamma\in \Dup(W)\setminus L} \Mean_{\gamma}^{m_W(\gamma)-1} \le \sum_{\substack{L\subseteq \Dup(W) \\ \Clos(L)~r\text{-bicycle free}}} \prod_{\gamma\in \Dup(W)\setminus L} \left(\frac{{\ACFacDeg{}}_{\max}}{n}\right)^{m_W(\gamma)-1}
\]
We proceed to bound this in a manner identical to \cite{BMR19}.  We define a weight function $w$ on subsets of $\Dup(W)$ as follows: $w(K) = \sum_{\gamma\in K} m_W(\gamma)-1$.  Choose $D^*(W)$ as a maximum weight subset (according to $w$) of $W$ such that $\Clos(D^*(W))$ is $r$-bicycle free, and let $\Profl(W)\coloneqq w(\Dup(W))-w(D^*(W))$.  Note that for any $L\subseteq \Dup(W)$ such that $\Clos(L)$ is $r$-bicycle free, $\Profl(W)\le w(\Dup(W)\setminus L)$.
\begin{align*}
	\sum_{\substack{L\subseteq \Dup(W) \\ \Clos(L)~r\text{-bicycle free}}} \prod_{\gamma\in \Dup(W)\setminus L} \left(\frac{{\ACFacDeg{}}_{\max}}{n}\right)^{m_W(\gamma)-1} &\le \sum_{\substack{L\subseteq \Dup(W) \\ \Clos(L)~r\text{-bicycle free}}} \left(\frac{{\ACFacDeg{}}_{\max}}{n}\right)^{w(\Dup(W)\setminus L)}
	\intertext{Since $L\subseteq \Dup(W)$ every $m_W(\gamma)-1\ge 1$.  Using this along with $\Profl(W)\le w(\Dup(W)\setminus L)$ we can bound the above by:}
	&\le \sum_{L\subseteq \Dup(W)} \left(\frac{{\ACFacDeg{}}_{\max}}{n}\right)^{\max\{|\Dup(W)\setminus L|,\Profl(W)\}}\\
	&= \sum_{i\le\Profl(W)} \left(\frac{{\ACFacDeg{}}_{\max}}{n}\right)^{\Profl(W)}\cdot{|\Dup(W)|\choose i} + \sum_{i>\Profl(W)} \left(\frac{{\ACFacDeg{}}_{\max}}{n}\right)^i {|\Dup(W)|\choose i}\\
	&\le (\Profl(W)+1)\left(\frac{{\ACFacDeg{}}_{\max}}{n}\right)^{\Profl(W)} + \sum_{i>\Profl(W)} \left(\frac{{\ACFacDeg{}}_{\max}|\Dup(W)|}{n}\right)^i \\
	&\le (\Profl(W)+2)\left(\frac{{\ACFacDeg{}}_{\max}}{n}\right)^{\Profl(W)} \\
	&\le 2\left(\frac{2{\ACFacDeg{}}_{\max}}{n}\right)^{\Profl(W)}
\end{align*}
Plugging this back into \pref{eq:about-to-bound-profligate} gives us:
\begin{align*}
	\pref{eq:about-to-bound-profligate} &\le 2\sum_{W\in\Lkgs(\BlankCol,n,k,\ell)} \Tr(M_W) \cdot 2^{|\Sing(W)|} \cdot\left(\frac{1}{n^{.5}}\right)^{\frac{|\Sing(W)|}{r}-\Exc_W} \cdot \left(\frac{2{\ACFacDeg{}}_{\max}}{n}\right)^{\Profl(W)} \cdot \prod_{\gamma\in R(W)}\Mean_{\gamma}
\end{align*}
As a first step towards simplifying the above quantity we make the following definition:
\begin{definition}[Shape of a linkage]	\label{def:linkage-shape}
	Given a $(k\times 2\ell)$-nonbacktracking linkage $W=v_0v_1\dots v_{2k\ell}$ that visits $v$ distinct vertices, we say the \emph{shape} of $W$ denoted $\Sh(W)$ is the $(k\times 2\ell)$-nonbacktracking linkage on graph on vertex set $[v]$ obtained by first constructing map $\xi:V(W)\to[2k\ell]$ where $\xi(v) = i$ where $v$ is the $i$-th distinct vertex visited by $W$ and defining the $t$-th step of the walk $\Sh(W)$ to be $\xi(v_{t-1})\xi(v_t)$.  We say the left vertex set of $\Sh(W)$ is $\xi(L(W))$ and the right vertex set is $\xi(R(W))$.  For $a\in V(\Sh(W))$ we will use the notation $W[a]$ to denote $\xi^{-1}(a)$.

	We use $\Shps(k,\ell)$ to denote the set of all distinct shapes of linkages in $\Lkgs(\BlankCol, n, k, \ell)$, and $\Shps(k,\ell,v,e)$ to denote the set of all shapes in $\Shps(k,\ell)$ on $v$ vertices and $e$ edges.
\end{definition}
We can rewrite the bound on \pref{eq:about-to-bound-profligate} as:
\begin{align*}
	\pref{eq:about-to-bound-profligate} &\le 2\sum_{\Sh\in\Shps(k,\ell)} 2^{|S(\Sh)|} \cdot\left(\frac{1}{n^{.5}}\right)^{\frac{|S(\Sh)|}{r}-\Exc_{\Sh}} \cdot \left(\frac{2{\ACFacDeg{}}_{\max}}{n}\right)^{\Profl(\Sh)}
	\sum_{\substack{W:W\in\Lkgs(\BlankCol,n,k,\ell) \\ \Sh(W)=\Sh}} \Tr(M_W) \cdot \prod_{\gamma\in R(W)}\Mean_{\gamma}.	\numberthis \label{eq:sep-into-shapes}
\end{align*}
For a given linkage $W$, we begin by deriving an upper bound on $\Tr(M_W)$.  A preliminary observation is:
\begin{observation}	\label{obs:bound-trace}
	$\Tr(M_W)\le q\|M_W\|$.
\end{observation}
Our next step is to decompose $W$ into simpler ``subwalks''.  This segment of the argument follows \cite{BC19,OW20}
\begin{definition}
	We call a vertex $v$ in $L(W)$ a \emph{landmark} of $W$ if it satisfies at least one of the following conditions: (i) is an endpoint of a link, (ii) $\deg_{G(W)}(v)\ge 3$, (iii) $\deg_{G(W)}(w)\ge 3$ for some $w\in R(W)$ which is incident to $v$ within $G(W)$.  We refer to the set of all landmark vertices in $W$ as $\Lm(W)$ We call any path between two landmark vertices $v_1$ and $v_2$ with no intermediate landmark vertices a \emph{trail}.  We call a trail a \emph{forked trail} if it has an intermediate vertex $w$ in $R(W)$ such that $\deg_{G(W)}(v)\ge 3$, and an \emph{unforked trail} otherwise.  We use $\Trs(W)$ to denote the collection of all trails in $W$, $\UTrs(W)$ to denote the collection of all unforked trails in $W$ and $\FTrs(W)$ to denote the collection of all forked trails in $W$.
\end{definition}
\begin{observation}
	Any forked trail must be a single two-step of the form $u\gamma v$ where $u$ and $v$ are left vertices and $\gamma$ is a right vertex.
\end{observation}
Any $W\in\Lkgs(\BlankCol, n, k, \ell)$ is a sequence of nonbacktracking walks on trails.  $W$ can be written as the sequence of vertices visited $v_0\gamma_1v_1\dots \gamma_{k\ell} v_{k\ell}$.  Let $T$ be the set of all times $t$ such that $v_t$ is a landmark.  Using $T$, we construct a set of \emph{pause} times $P$ in the following way:
\begin{displayquote}
	For each $t\in T$, if the trail starting or ending at time $t$ is visited for the first or second time, we add $t$ to $P$.
\end{displayquote}
Recall from \pref{def:walk-matrix} that $M_W$ is the product of $k\ell$ matrices $M_{v_0\gamma_1v_1}\dots M_{v_{k\ell-1}\gamma_{k\ell}v_{k\ell}}$.  Let $p_1<\dots<p_s$ be the sequence of all pause times.  By submultiplicativity of the operator norm,
\[
	\|M_W\| \le \|M_{v_0\gamma_1v_1}\cdots M_{v_{p_1-1}\gamma_{p_1}v_{p_1}}\|\cdot\|M_{v_{p_1}\gamma_{p_1+1}v_{p_1+1}}\cdots M_{v_{p_2-1}\gamma_{p_2}v_{p_2}}\|\cdots\|M_{v_{p_{s}}\gamma_{p_{s}+1}v_{p_{s}+1}}\cdots M_{v_{k\ell-1}\gamma_{k\ell}v_{k\ell}}\|.
\]
Each segment between consecutive pauses $p_i$ and $p_{i+1}$ falls into one of the following categories:
\begin{itemize}
	\item $\Seg_{\le 2}(W)$: the segment is composed of exactly one trail and is the first or second visit to the trail.
	\item $\Seg_{>2}(W)$ the segment is a union of trails and each of these trails has already been visited at least twice before.
\end{itemize}
Rewriting the above upper bound, we now have:
\[
	\|M_W\| \le \prod_{\omega\in\Seg_{\le 2}(W)}\|M_{\omega}\|\cdot\prod_{\omega\in\Seg_{>2}(W)}\|M_{\omega}\|.
\]
Given a trail $T$ with endpoints $u$ and $v$, there are two nonbacktracking walks $\omega_1$ and $\omega_2$ that cover $T$, one from $u$ to $v$ and another from $v$ to $u$.  By \pref{cond:self-adj}, $M_{\omega_1}^* = M_{\omega_2}$ where the $*$ in the superscript refers to the adjoint induced by the inner product $\langle \cdot, \cdot\rangle_v$ and so $\|M_{\omega_1}\| = \|M_{\omega_2}\|$.  Henceforth we use $\|M_T\|$ to denote $\|M_{\omega_1}\|=\|M_{\omega_2}\|$.
Using the notation $\UTrs_{\ge 2}(W)$ for unforked trails that are visited more than once, we can write the above as:\sidtemp{modify this to make it more readable}
\[
	\|M_W\| \le \prod_{T\in\UTrs_{\ge 2}(W)} \|M_T\|^2 \cdot \prod_{\omega\in\Seg_{>2}(W)\cup\Seg_1(W)\setminus\UTrs_{\ge 2}(W)}\|M_{\omega}\|.
\]
\pref{cond:matrix-product-small} further lets us get the following bound:
\[
	\|M_W\| \le C^{|\Seg_{1}(W)|+|\Seg_{\ge 2}(W)|} \cdot \prod_{T\in\UTrs_{\ge 2}(W)} \|M_T\|^2 \le C^{|P|+1}\cdot\prod_{T\in\UTrs_{\ge 2}(W)} \|M_T\|^2	\numberthis \label{eq:separated-squared-walks}
\]
We now turn our attention to bounding $|P|$.  Given a landmark vertex $v$ and an edge $e$ incident to it, there are at most $\arity_{\max}$ trails that start at $v$ and tread on $e$ on their first step, and hence the number of distinct trails starting at $v$ is at most $\arity_{\max}\cdot\deg_{G(W)}(v)$.  By the construction of $P$, the number of pauses at vertex $v$ is at most twice the number of distinct trails starting at $v$, and hence is at most $2\arity_{\max}\cdot\deg_{G(W)}(v)$.  Thus:
\[
	|P| \le 2\arity_{\max} \sum_{v\in\Lm(W)} \deg_{G(W)}(v) \le 2\arity_{\max} \left(2|\Lm_{\le 2}(W)| + \sum_{v\in\Lm_{\ge 3}(W)} \deg_{G(W)}(v)\right)
\]
where $\Lm_{\le 2}(W)$ and $\Lm_{\ge 3}(W)$ denote the sets of landmark vertices of degree-$\le 2$ and degree-$\ge 3$ respectively.  Each vertex in $\Lm_{\le 2}(v)$ is either an endpoint of a link or a neighbor of a degree-$\ge 3$ right vertex.  There are exactly $k+1$ endpoints of links, and at most $\sum_{v\in G(W):\deg_{G(W)}(v)\ge 3}\deg_{G(W)}(v)$ landmark vertices induced as neighbors of degree-$\ge 3$ right vertices, and hence:
\[
	|P| \le 2\arity_{\max} \left(2(k+1) + 2\sum_{v\in G(W):\deg_{G(W)}(v)\ge 3} \deg_{G(W)}(v) + \sum_{v\in G(W):\deg_{G(W)}(v)\ge 3} \deg_{G(W)}(v) \right)	\numberthis \label{eq:bound-on-pauses}
\]
It remains to bound $\sum_{v\in G(W):\deg_{G(W)}(v)\ge 3} \deg_{G(W)}(v)$.  Let $X$ be a set of edges of size $\Exc_W$ such that $T(W,X)\coloneqq (V(W),E(W)\setminus X)$ is a tree.  Since $G(W)$ has at most $k$ leaves, $T(W,X)$ has at most $k+2\Exc_W$ leaves.  We now state the following well known fact about trees and refer the reader to \cite[Fact 6.35]{BMR19} for a proof.
\begin{fact}	\label{fact:tree-leaves-deg}
	Let $T$ be a tree with $l$ leaves.  Then $3l \ge \sum_{v\in V(T):\deg_{T}(v)\ge 3}\deg_T(v)$.
\end{fact}
As a consequence of \pref{fact:tree-leaves-deg}:
\[
	\sum_{v\in V(T(W,X)):\deg_{T(V,X)}(v)\ge 3} \deg_{T(V,X)}(v) \le 3(k+2\Exc_W).
\]
Now observe that for any graph $\Gamma$ and graph $\Gamma'$ obtained by adding a single edge to $\Gamma$:
\[
	\sum_{v\in V(\Gamma'):\deg_{\Gamma'}(v)\ge 3} \deg_{\Gamma'}(v) \le \left(\sum_{v\in V(\Gamma):\deg_{\Gamma}(v)\ge 3} \deg_{\Gamma}(v)\right) + 6,
\]
and thus
\[
	\sum_{v\in G(W):\deg_{G(W)}(v)\ge 3} \deg_{G(W)}(v) \le 3(k+2\Exc_W) + 6\Exc_W = 3k + 12\Exc_W.	\numberthis \label{eq:total-edges-deg-3}
\]
Plugging this back into \pref{eq:bound-on-pauses} and using $k\ge 2$ gives:
\[
	|P| \le 2\arity_{\max}\left(12k+36\Exc_W\right) = 24\arity_{\max}\left(k+3\Exc_W\right).
\]
\begin{remark}	\label{rem:bound-landmark-deg}
Observe that the above also proved the following, which will be of utility later in the proof:
\[
	\sum_{v\in\Lm(W)} \deg_{G(W)}(v) \le 12k+36\Exc_W.
\]
\end{remark}
Next, plugging the bound on $|P|$ back into \pref{eq:separated-squared-walks} along with \pref{obs:bound-trace} gives:
\[
	\Tr(M_W) \le qC^{24\arity_{\max}(k+3\Exc_W)+1}\cdot\prod_{T\in\UTrs_{\ge 2}(W)} \|M_T\|^2 \le qC^{25\arity_{\max}(k+3\Exc_W)} \cdot \prod_{T\in\UTrs_{\ge 2}(W)} \|M_T\|^2.
\]
And via \pref{eq:sep-into-shapes} and introducing the expected value over the randomness of $\btau$:
\begin{align*}
	\E_{\btau}\E_{\bH|\btau}[\bU] \le 2q\sum_{\Sh\in\Shps(k,\ell)} \left(\frac{2^r}{n^{.5}}\right)^{\frac{|S(\Sh)|}{r}} \cdot \left(n^{.5}C^{75\arity_{\max}}\right)^{\Exc_{\Sh}} C^{25\arity_{\max}k} \cdot \left(\frac{2{\ACFacDeg{}}_{\max}}{n}\right)^{\Profl(\Sh)}\\
	\sum_{\substack{W:W\in\Lkgs(\BlankCol,n,k,\ell) \\ \Sh(W)=\Sh}} \prod_{T\in\UTrs_{\ge 2}(W)} \|M_T\|^2 \cdot \E_{\btau} \prod_{\gamma\in R(W)}\Mean_{\gamma}	\numberthis \label{eq:bound-sep-trace}
\end{align*}
Thus, for a fixed shape $\Sh$ we now restrict our attention to bounding:
\begin{align*}
	\sum_{\substack{W:W\in\Lkgs(\BlankCol,n,k,\ell) \\ \Sh(W)=\Sh}} \prod_{T\in\UTrs_{\ge 2}(W)} \|M_T\|^2 \cdot \E_{\btau} \prod_{\gamma\in R(W)}\Mean_{\gamma}.	\numberthis \label{eq:restrict-to-shape}
\end{align*}

We now define the notion of $\Cl$-consistent.
\begin{definition}	\label{def:consistent-subgraph}
	Let $\Gamma$ be a subgraph of the complete bipartite factor graph $\Bip(\CompG_{n})$.  We say $\Gamma$ is $\Cl$-consistent if there exists a $\tau$ such that every $\gamma = (v_1,\dots,v_{\arity(i)})\in R(\Gamma)$ satisfies $(\tau(v_1),\dots,\tau(v_{\arity(i)})) = \Cl(\FacType(\gamma))$. 
	If $\Gamma$ is $\Cl$-consistent we use $\tau_{\Gamma}$ to refer to the unique $\tau$ such that $\Gamma$ is $(\tau,\Cl)$-consistent.
\end{definition}

\begin{observation}	\label{obs:inconsistent-0}
	For $W\in\Lkgs(\BlankCol, n, k, \ell)$, $\prod_{\gamma\in R(W)}\Mean_{\gamma}$ is equal to $0$ if $W$ is not $\Cl$-consistent.
\end{observation}

Thus:
\begin{align*}
	\pref{eq:restrict-to-shape} &= \sum_{\substack{W:W\in\Lkgs(\BlankCol,n,k,\ell) \\ \Sh(W)=\Sh \\ W~\Cl\text{-consistent}}} \prod_{T\in\UTrs_{\ge 2}(W)} \|M_T\|^2 \cdot \prod_{\gamma\in R(W)} \frac{\ACFacDeg{\FacType(\gamma)}}{n^{\arity(\FacType(\gamma))-1}} \cdot \prod_{v\in L(\Clos(W))} \TypeDist(\tau_W(v)) .
\end{align*}

\begin{definition}	\label{def:clos-equiv}
	We call two walks $W_1$ and $W_2$ \emph{equivalent} denoted $W_1\sim W_2$ if
	\begin{itemize}
		\item $\Sh(W_1)=\Sh(W_2)=:\Sh$,
		\item for any $a\in R(\Sh)$, $\RType(W_1[a])=\RType(W_2[a])$,
		\item for any edge $\{v,\gamma\}$ in $\Sh$ for $v\in L(\Sh)$ and $\gamma\in R(\Sh)$, $\Index(W_1[\gamma], W_1[v]) = \Index(W_2[\gamma], W_2[v])$.
	\end{itemize}
	We say $W_1$ and $W_2$ are \emph{closure equivalent} denoted $W_1\sim_{\Clos} W_2$ if $W_1\sim W_2$ and the graphs induced by $\Clos(W_1)$ and $\Clos(W_2)$ are isomorphic.
\end{definition}
The relationship $\sim$ partitions the space of all $\Cl$-consistent $W$ in $\Lkgs(\BlankCol, n, k, \ell)$ with shape $\Sh$ into a collection of equivalence classes $\calC$.  We use $[W]$ to denote the equivalence class it is contained in. $\sim_{\Clos}$ further partitions each equivalence class $[W]\in\calC$ into a collection of sub-equivalence classes $\calC_{[W]}$, and we denote the sub-equivalence class of $W$ with $[[W]]$.  We use $J(W)$ to denote $|E(\Clos(W))|-|E(W)|$.  With these definitions and notation in hand, we can write:
\begin{align*}
	\pref{eq:restrict-to-shape} &= \sum_{[W]\in\calC} \sum_{[[W]]\in\calC[W]} \sum_{W'\in[[W]]} \prod_{T\in\UTrs_{\ge 2}(W)} \|M_T\|^2 \cdot \prod_{\gamma\in R(W)} \frac{\ACFacDeg{\FacType(\gamma)}}{n^{\arity(\FacType(\gamma))-1}} \cdot \prod_{v\in L(\Clos(W))} \TypeDist(\tau_W(v))\\
	&\le \sum_{[W]\in\calC} \sum_{[[W]]\in\calC[W]} \prod_{T\in\UTrs_{\ge 2}(W)} \|M_T\|^2 \cdot \prod_{\gamma\in R(W)} \frac{\ACFacDeg{\FacType(\gamma)}}{n^{\arity(\FacType(\gamma))-1}} \cdot \prod_{v\in L(\Clos(W))} \TypeDist(\tau_W(v)) \cdot n^{|L(\Clos(W))|}\\
	&= \sum_{[W]\in\calC} \sum_{t=0}^{J(W)} \sum_{\substack{[[W]]\in\calC[W] \\ |L(\Clos(W))|=|L(W)|+t}} \prod_{T\in\UTrs_{\ge 2}(W)} \|M_T\|^2 \cdot \prod_{\gamma\in R(W)} \frac{\ACFacDeg{\FacType(\gamma)}}{n^{\arity(\FacType(\gamma))-1}} \cdot \prod_{v\in L(\Clos(W))} \TypeDist(\tau_W(v)) \cdot \frac{n^{|L(W)|+|J(W)|}}{n^{|J(W)|-t}}
	\intertext{To enumerate the innermost sum, first observe that any $W$ can be changed to $\Clos(W)$ where $|L(\Clos(W))|-|L(W)|=t$ via the following procedure.  First add $J(W)$ new vertices and to each $\gamma\in R(W)$ attach an edge from $\gamma$ to $\arity(\FacType(\gamma))-\deg_W(\gamma)$ of the new vertices.  There exists a sequence of $J(W)-t$ ``merge'' operations on the left vertices, and a labeling of the newly added left vertices in $[n]$ that would result in $\Clos(W)$.  The number of possible sequences of merge operations is at most $(\arity_{\max}k\ell)^{2(|J(W)|-t)}$.  Thus, the above sum can be bounded by:}
	&\le \sum_{[W]\in\calC} \sum_{t=0}^{J(W)} \prod_{T\in\UTrs_{\ge 2}(W)} \|M_T\|^2 \cdot \prod_{\gamma\in R(W)} \frac{\ACFacDeg{\FacType(\gamma)}}{n^{\arity(\FacType(\gamma))-1}} \cdot \prod_{v\in L(\Clos(W))} \TypeDist(\tau_W(v)) \cdot n^{|L(W)|+|J(W)|}\cdot\left(\frac{(\arity_{\max}k\ell)^2}{n}\right)^{|J(W)|-t}\\
	&= \sum_{[W]\in\calC} \sum_{t=0}^{J(W)} \prod_{T\in\UTrs_{\ge 2}(W)} \|M_T\|^2 \cdot \prod_{\gamma\in R(W)} \left(\frac{\ACFacDeg{\FacType(\gamma)}}{n^{\arity(\FacType(\gamma))-1}} \cdot \prod_{v\in L(\Clos(W))\setminus L(W)} \TypeDist(\tau_W(v)) \right) \cdot\\
	&\prod_{v\in L(W)} \TypeDist(\tau_W(v)) \cdot \prod_{v\in L(\Clos(W))\setminus L(w)}\frac{1}{\TypeDist(\tau_W(v))^{\deg_W(v)-1}} \cdot n^{|L(W)|+|J(W)|}\cdot\left(\frac{(\arity_{\max}k\ell)^2}{n}\right)^{|J(W)|-t}\\
	&\le \sum_{[W]\in\calC} \sum_{t=0}^{J(W)} \prod_{T\in\UTrs_{\ge 2}(W)} \|M_T\|^2 \cdot \prod_{\gamma\in R(W)} \left(\frac{\ACFacDeg{\FacType(\gamma)}}{n^{\arity(\FacType(\gamma))-1}} \cdot \prod_{v\in L(\Clos(W))\setminus L(W)} \TypeDist(\tau_W(v)) \right) \cdot\\
	&\prod_{v\in L(W)} \TypeDist(\tau_W(v)) \cdot n^{|L(W)|+|J(W)|}\cdot\left(\frac{(\arity_{\max}k\ell)^2}{\TypeDist_{\min}n}\right)^{|J(W)|-t}
	\intertext{We can use the bound $\sum_{t=0}^{J(W)} \left(\frac{(\arity_{\max}k\ell)^2}{\TypeDist_{\min}n}\right)^{|J(W)|-t} \le 2$ to deduce:}
	&\le 2\sum_{[W]\in\calC} \prod_{T\in\UTrs_{\ge 2}(W)} \|M_T\|^2 \cdot \prod_{\gamma\in R(W)} \left(\frac{\ACFacDeg{\FacType(\gamma)}}{n^{\arity(\FacType(\gamma))-1}} \cdot \prod_{v\in L(\Clos(W))\setminus L(W)} \TypeDist(\tau_W(v)) \right)\cdot \prod_{v\in L(W)} \TypeDist(\tau_W(v)) \cdot n^{|L(W)|+|J(W)|}\\
	&= 2\sum_{[W]\in\calC} \prod_{T\in\UTrs_{\ge 2}(W)} \|M_T\|^2 \cdot \prod_{\gamma\in R(W)} \left(\ACFacDeg{\FacType(\gamma)} \cdot \prod_{v\in L(\Clos(W))\setminus L(W)} \TypeDist(\tau_W(v)) \right) \cdot \prod_{v\in L(W)} \TypeDist(\tau_W(v)) \cdot n^{1-\Exc_W}. \numberthis \label{eq:about-to-move-into-UTrs2}
\end{align*}

We decompose $\prod_{\gamma\in R(W)} \left(\ACFacDeg{\eta(\gamma)}\cdot\prod_{v\in L(\Clos(W))\setminus L(W)}\TypeDist(\tau_W(v))\right)$ into three parts: the contribution of singleton right vertices $\prod_{\gamma\in\Sing(W)} \left(\ACFacDeg{\eta(\gamma)}\cdot\prod_{v\in L(\Clos(W))\setminus L(W)}\TypeDist(\tau_W(v))\right)$ which can be bounded by ${\ACFacDeg{}}_{\max}^{|\Sing(\Sh)|}$, the contribution of duplicative degree-$\ge 3$ right vertices $\prod_{\substack{\gamma\in\Dup(W) \\ \deg_{W}(\gamma)\ge 3}}\left(\ACFacDeg{\eta(\gamma)}\cdot\prod_{v\in L(\Clos(W))\setminus L(W)}\TypeDist(\tau_W(v))\right)$ which via \pref{eq:total-edges-deg-3} can be bounded by ${\ACFacDeg{}}_{\max}^{3k+12\Exc_{\Sh}}$, and finally the contribution of duplicative degree-$2$ right vertices $\prod_{\substack{\gamma\in\Dup(W) \\ \deg_{W}(\gamma) = 2}} \left(\ACFacDeg{\eta(\gamma)}\cdot\prod_{v\in L(\Clos(W))\setminus L(W)}\TypeDist(\tau_W(v))\right)$.  For each $\gamma$ considered in the final case we can identify a unique $T\in\UTrs_{\ge 2}(W)$ such that $\gamma$ is in $T$, and likewise for every $T\in\UTrs_{\ge 2}(W)$, every right vertex $\gamma$ in $T$ is duplicative and has degree exactly $2$ and hence appears in the product.  Thus, the third product can be written as
\[
	\prod_{T\in\UTrs_{\ge2}(W)} \prod_{\gamma\in R(T)} \left(\ACFacDeg{\FacType(\gamma)}\cdot\prod_{v\in L(\Clos(T))\setminus L(T)} \TypeDist(\tau_W(v))\right).
\]
Next, observe that using the facts that each $\pi_i\in[0,1]$ and every interior left vertex of a trail occurs in no other trail, $\prod_{v\in L(W)}\TypeDist(\tau_W(v))$ can be upper bounded by $\prod_{T\in\UTrs_{\ge 2}(W)}\prod_{v\in L(T)}\TypeDist(\tau_W(v)) \cdot \prod_{v\in\Lm(W)}\frac{1}{\TypeDist(\tau_W(v))^{\deg_{G(W)}(v)}}$, which by \pref{rem:bound-landmark-deg} is at most $\prod_{T\in\UTrs_{\ge 2}(W)}\prod_{v\in L(T)}\TypeDist(\tau_W(v)) \cdot \left(\frac{1}{\TypeDist_{\min}}\right)^{12k+36\Exc_{\Sh}}$.
\begin{align*}
	\pref{eq:about-to-move-into-UTrs2} &\le 2{\ACFacDeg{}}_{\max}^{|\Sing(\Sh)|+3k+12\Exc_{\Sh}} \left(\frac{1}{\TypeDist_{\min}}\right)^{12k+36\Exc_{\Sh}} n^{-\Exc_{\Sh}+1} \cdot \\ &\sum_{[W]} \prod_{T\in\UTrs_{\ge 2}(W)} \left(\|M_{T}\|^2 \cdot \prod_{\gamma\in R(T)}\ACFacDeg{\FacType(\gamma)} \cdot\prod_{v\in L(\Clos(T))\setminus L(T)} \TypeDist(\tau_W(v)) \cdot \prod_{v\in L(T)}\TypeDist(\tau_W(v))\right)\\
	&\le 2{\ACFacDeg{}}_{\max}^{|\Sing(\Sh)|+3k+12\Exc_{\Sh}} \left(\frac{1}{\TypeDist_{\min}}\right)^{12k+36\Exc_{\Sh}} n^{-\Exc_{\Sh}+1} \cdot \\ & \prod_{T\in\UTrs_{\ge 2}(\Sh)} \sum_{[U]:\Sh(U)=T} \left(\|M_{U}\|^2 \cdot \prod_{\gamma\in R(U)}\ACFacDeg{\FacType(\gamma)} \cdot\prod_{v\in L(\Clos(U))\setminus L(U)} \TypeDist(\tau_W(v)) \cdot \prod_{v\in L(U)}\TypeDist(\tau_W(v))\right)
	\intertext{For any constant $\eps$, there exists a constant $C_{\eps}$ such that the above is at most:}
	&\le 2{\ACFacDeg{}}_{\max}^{|\Sing(\Sh)|+3k+12\Exc_{\Sh}} \left(\frac{1}{\TypeDist_{\min}}\right)^{12k+36\Exc_{\Sh}} n^{-\Exc_{\Sh}+1} \prod_{T\in\UTrs_{\ge2}(\Sh)} C_{\eps}\cdot\left((1+\eps)\rho(\BlankCol,\GModPar)\right)^{2|T|}
	\intertext{Since $|\UTrs_{\ge2}(\Sh)|$ is at most the sum of degrees of landmark vertices on which we have an upper bound by \pref{rem:bound-landmark-deg}, and since $\sum_{T\in\UTrs_{\ge2}(\Sh)} |T| \le \frac{k\ell}{2}$:}
	&\le 2{\ACFacDeg{}}_{\max}^{|\Sing(\Sh)|+3k+12\Exc_{\Sh}} \left(\frac{1}{\TypeDist_{\min}}\right)^{12k+36\Exc_{\Sh}} n^{-\Exc_{\Sh}+1} C_{\eps}^{12k+36\Exc_{\Sh}}\cdot\left((1+\eps)\rho(\BlankCol,\GModPar)\right)^{k\ell}.
\end{align*}
Since \pref{eq:about-to-move-into-UTrs2} is an upper bound on \pref{eq:restrict-to-shape}, rearranging the terms in the above gives:
\[
	\pref{eq:restrict-to-shape} \le 2n\cdot{\ACFacDeg{}}_{\max}^{|S(\Sh)|}\cdot\left(\frac{{\ACFacDeg{}}_{\max}^3C_{\eps}^{12}}{\TypeDist_{\min}^{12}}\right)^k\cdot\left(\frac{{\ACFacDeg{}}_{\max}^{12}C_{\eps}^{36}}{\TypeDist_{\min}^{36}n}\right)^{\Exc_{\Sh}}\cdot\left((1+\eps)\rho(\BlankCol,\GModPar)\right)^{k\ell}.
\]
Plugging this upper bound on \pref{eq:restrict-to-shape} into \pref{eq:bound-sep-trace} gives us:
\begin{align*}
	\E_{\btau}\E_{\bH|\btau}[\bU] \le&
	4nq\left((1+\eps)\rho(\BlankCol,\GModPar)\right)^{k\ell}\cdot \\
	&\sum_{\Sh\in\Shps(k,\ell)}
	\left(\frac{2^r{\ACFacDeg{}}_{\max}^r}{n^{.5}}\right)^{\frac{|S(\Sh)|}{r}}
	\left(\frac{C^{75\arity_{\max}}{\ACFacDeg{}}_{\max}^{12}C_{\eps}^{36}}{\TypeDist_{\min}^{36}n^{.5}}\right)^{\Exc_{\Sh}}
	\left(\frac{C^{25\arity_{\max}}{\ACFacDeg{}}_{\max}^3C_{\eps}^{12}}{\TypeDist_{\min}^{12}}\right)^k \left(\frac{2{\ACFacDeg{}}_{\max}}{n}\right)^{\Profl(\Sh)}.
\end{align*}
To notationally lighten the above, we choose $\beta$ as a constant larger than $2{\ACFacDeg{}}_{\max}$, $\frac{C^{75\arity_{\max}}{\ACFacDeg{}}_{\max}^{12}C_{\eps}^{36}}{\TypeDist_{\min}^{36}}$ and $\frac{C^{25\arity_{\max}}{\ACFacDeg{}}_{\max}^3C_{\eps}^{12}}{\TypeDist_{\min}^{12}}$.  Then:
\[
	\E_{\btau}\E_{\bH|\btau}[\bU] \le
	4nq\left((1+\eps)\rho(\BlankCol,\GModPar)\right)^{k\ell} \beta^k \cdot
	\sum_{\Sh\in\Shps(k,\ell)}
	\left(\frac{\beta^r}{n^{.5}}\right)^{\frac{|S(\Sh)|}{r}}
	\left(\frac{\beta}{n^{.5}}\right)^{\Exc_{\Sh}}
	\left(\frac{\beta}{n}\right)^{\Profl(\Sh)}
	\numberthis \label{eq:rho-has-entered-the-picture}
\]
To obtain a bound on \pref{eq:rho-has-entered-the-picture} we first bound:
\[
	\sum_{\Sh\in\Shps(k,\ell)}
	\left(\frac{\beta^r}{n^{.5}}\right)^{\frac{|S(\Sh)|}{r}}
	\left(\frac{\beta}{n^{.5}}\right)^{\Exc_{\Sh}}
	\left(\frac{\beta}{n}\right)^{\Profl(\Sh)} 
	\numberthis \label{eq:sum-over-all-shapes}
\]

We proceed by partition all $\Sh\in\Shps(k,\ell)$ into sets where each set of $\Sh$ share the same $|S(\Sh)|$, $\Exc_{\Sh}$, and $\Profl(\Sh)$. We bound the sum for each set with the following claim proved in \pref{app:random-graph-lemmas}.

\begin{claim}\label{clm:shape-count}
Let $\mathcal{U}_{s,x,\Profl}$ denote the set of all $\Sh\in\Shps(k,\ell)$ with $|S(\Sh)| = s$, $\Exc_{\Sh} = x$, and $\Profl(\Sh) = \Profl$. Then 
\[\sum_{\Sh\in\mathcal{U}_{s,x,\Profl}}
\left(\frac{\beta^r}{n^{.5}}\right)^{\frac{|S(\Sh)|}{r}}
	\left(\frac{\beta}{n^{.5}}\right)^{\Exc_{\Sh}}
	\left(\frac{\beta}{n}\right)^{\Profl(\Sh)} \leq 
	\left(\frac{\beta^r\cdot(4k\ell)^r}{n^{.5}}\right)^{\frac{s}{r}}
	\left(\frac{\beta\cdot 2(k\ell)^3}{n^{.5}}\right)^{x} 
	\left(\frac{4\beta\cdot(2k\ell)^2}{n}\right)^{\Profl}(2k\ell)^{O(k\log{k\ell})} .\]
\end{claim}

Using this claim we can derive 
\begin{align*}
\pref{eq:sum-over-all-shapes} 
&= \sum_{s\in[k\ell],x\in[2k\ell],\Profl\in[k\ell]} \sum_{\Sh\in\mathcal{U}_{s,x,\Profl}} \left(\frac{\beta^r}{n^{.5}}\right)^{\frac{s}{r}}
	\left(\frac{\beta}{n^{.5}}\right)^{x}
	\left(\frac{\beta}{n}\right)^{\Profl} \\
&\leq \sum_{s\in[k\ell],x\in[2k\ell],\Profl\in[k\ell]} \left(\frac{\beta^r\cdot(4k\ell)^r}{n^{.5}}\right)^{\frac{s}{r}}
	\left(\frac{\beta\cdot 2(k\ell)^3}{n^{.5}}\right)^{x} 
	\left(\frac{4\beta\cdot(2k\ell)^2}{n}\right)^{\Profl}(2k\ell)^{O(k\log{k\ell})} \\
&\leq (2k\ell)^{O(k\log{k\ell})} k\ell \cdot 2k\ell \cdot k\ell \max_{s\in[k\ell],x\in[2k\ell],\Profl\in[k\ell]} \left(\frac{\beta^r\cdot(4k\ell)^r}{n^{.5}}\right)^{\frac{s}{r}}
	\left(\frac{\beta\cdot 2(k\ell)^3}{n^{.5}}\right)^{x} 
	\left(\frac{4\beta\cdot(2k\ell)^2}{n}\right)^{\Profl}
\end{align*} 

We set the bicycle-free radius to $r := o(\log{n^{.5}}/\log{4k\ell\beta})$ so that all three terms $\frac{\beta^r\cdot(4k\ell)^r}{n^{.5}}$, $\frac{\beta\cdot 2(k\ell)^3}{n^{.5}}$, $\frac{4\beta\cdot(2k\ell)^2}{n}$ are less than $1$. Then we observe that 
\[\pref{eq:sum-over-all-shapes} \leq 2k^3\ell^3(2k\ell)^{O(k\log{k\ell})} .\]

Plugging this bound into \pref{eq:rho-has-entered-the-picture} gives us:

\[\E_{\btau}\E_{\bH|\btau}[\bU] \le
	\left((1+\eps)\rho(\BlankCol,\GModPar)\right)^{k\ell}  \cdot 2nq\beta^k2k^3\ell^3(2k\ell)^{O(k\log{k\ell})}\]

Set $k\ell = O\left(\log{n}\cdot\log{\log{n}}\right)$, and $\ell = \omega(\log{k\ell})$. Then we have the bound
\begin{align*}
\left(n\E_{\btau}\E_{\bH|\btau}[\bU]\right)^{1/k} 
&\le \left((1+\eps)\rho(\BlankCol,\GModPar) \cdot \left(4q\beta^k\cdot n^2 \cdot (k\ell) \cdot (2k\ell)^{O(\log{k\ell})/\ell}\right)^{1/k\ell}\right)^{\ell} \\
&\le \left((1+\eps)\rho(\BlankCol,\GModPar) \cdot \left(8q\beta^k\cdot n^3 \right)^{1/k\ell}\right)^{\ell} \\
&\le \left((1+\eps)\rho(\BlankCol,\GModPar) \cdot \left(1 + \frac{2\log(8q\beta^k)}{k\ell} + \frac{2\log(n^3)}{k\ell} \right)\right)^{\ell} \\
&\le \left(\left(1+\eps + O\left(\frac{1}{\log{\log{n}}}\right)\right)\rho(\BlankCol,\GModPar)\right)^{\ell}.
\end{align*} 

We now complete the proof of \pref{thm:main-spectral-bound}.

 \newcommand{\bfm}{\mathrm{m}}





\section{Weak Recovery} \label{sec:recovery}

We begin the section by briefly describing an algorithm for weak recovery.

\noindent\rule{16cm}{0.4pt}

{\bf Weak Recovery Algorithm}

\begin{enumerate}
	\item Fix $\delta > 0$ such that $\lambda_L \geq (1+\delta)^{4}$.

	\item $C = O_{\Model,\delta}(1)$ is a sufficiently large constant depending on model $\Model$ and $\delta$.

	\item For $(\log \log n)^3 \leq t \leq (\log \log n)^5$,  $v_t \in \R^{nq}$ be the eigenvector with largest eigenvalue of $\CentAdj_{\bG} ^{(t)}$ and let $\Lambda_t^t$ denote the largest eigenvalue.  Compute $v_t$, $\Lambda_t$ and $\CentAdj_{\bG}^{(t)}$ for all $t$ in this range.

	\item Find $\bfm \in [(\log \log n)^3 ,  (\log \log n)^5]$ such that for all $s \in [(\log \log n)^3, \bfm]$, we have
		\[ \norm{\CentAdj_{\bG}^{(s)} v_{\bfm}} \leq \Lambda_{\bfm}^s (1+\delta)^{\bfm-s} \]
		(see \pref{claim:whatever} for proof of existence of $\bfm$)

	 \item For each $0 \leq \ell \leq C$, set 
	 \[ w_\ell \defeq \CentAdj_{\bG}^{(\bfm - \ell)} v_{\bfm} \]  
	 and let $\overline{w}_\ell \defeq \frac{1}{\norm{w_{\ell}}} \cdot w_{\ell}$.
	
	\item Output the set of vectors $\{ \overline{u}_{\ell,\beta}\}$ in $\R^{n}$ for $0 \leq \ell \leq C$, $\beta \in [q]$ defined as,
	 \[ \overline{u}_{\ell,\beta} [i] = \overline{w}_\ell[i, \beta]\]

\end{enumerate}
\noindent\rule{16cm}{0.4pt}

\begin{theorem}	\label{thm:recovery-full}
	There exists a constant $C = O_{\Model,\delta}(1)$ depending on model $\Model$ and $\delta$ such that the following holds with probability $1- o_n(1)$:  
	For some $\ell \in \{1,2,\ldots, C\}$ and $\beta \in [q]$, the
	unit vector $\overline{u}_{\ell,\beta}$ is correlated with the coloring in the following sense:
	$\exists \tau \in T, \alpha \in [q]$ such that if we construct $\bchat^{\tau,\alpha} \in \R^n$ as
	\[ \bchat^{\tau,\alpha} [i] = \Ind[\tau(i) = \tau] (\Ind[\bc(i) = \alpha] - \prior_{\tau}(\alpha) )\]
	then 
			\[ |\iprod{ \overline{u}_{\ell,\beta}, \bchat^{\tau,\alpha}}| \geq \Omega_{\Model}(1) \cdot \sqrt{n} \]	
\end{theorem}

In the rest of the section, we will outline the proof of correctness of the above described weak-recovery algorithm.  To this end, we begin by recalling the matrix $\CentAdj_{\bG} \in \R^{nq \times n q}$.  For all $i \neq j \in [n]$,
\begin{align*}
\CentAdj_{\bG}[i,j] \defeq \sum_{e\in \CompG_{n}, e \ni i j} \BPM{e}{i}{j} \cdot (\Ind[e \in \bG] - \Pr_{\Model}[e \in \bG | \btau ])
\end{align*}
and $\CentAdj_{\bG}[i,i] = 0$.  Similarly, for all $i \neq j \in [n]$ we set
\begin{align*}
\CentBdj_{\bG}[i,j] \defeq \sum_{e\in \CompG_{n}, e \ni i j} \BPM{e}{i}{j} \cdot (\Ind[e \in \bG] - \Pr_{\Model}[e \in \bG | \btau, \bc])
\end{align*}
and $\CentBdj_{\bG}[i,i] = 0$.  Finally, let 
\begin{align*}
\CentRdj_{\bG}[i,j] &\defeq \CentAdj_{\bG} [i,j] - \CentBdj_{\bG} [i,j] \\
&= \sum_{e\in \CompG_{n}, e \ni i j} \BPM{e}{i}{j} \left(\Pr_{\Model}[e \in \bG | \btau, \bc]  - \Pr_{\Model}[e \in \bG | \btau]  \right)
\end{align*}

Let $\vec{x}_{\bG}\in\R^{\abs{T}}$ encode the number of variables of each type in $\bG$ and $\vec{y}_{\bG}\in\R^{q\abs{T}}$ encode the number of variables of each type and color in $\bG$. Then block $\CentRdj_{\bG}[i,j]$ only depends on $\btau(i),\btau(j), \bc(i),\bc(j)$ and $\vec{y}_{\bG}$. More specifically

\begin{align*}
\CentRdj_{\bG}[i,j] &= \sum_{e\in \CompG_{n}, e \ni i j} \BPM{e}{i}{j} \left(\Pr_{\Model}[e \in \bG | \btau(i),\btau(j),\bc(i),\bc(j),\vec{y}_{\bG}]  - \Pr_{\Model}[e \in \bG | \btau(i), \btau(j),\vec{x}_{\bG}]  \right)
\end{align*}

\begin{remark}\label{rem:RG-condition}
We remark that with probability $1-o_n(1)$ each entry of $\vec{y}_{\bG}$ satisifies $\vec{y}_{\bG}[\tau,\alpha] \in (1\pm \epsilon)\cdot\E_{\Model}[\vec{y}_{\bG}[\tau,\alpha]]$ for some small constant $\epsilon$. From now on we only consider $\bG$ that satisifies this condition.
\end{remark}

Next we introduce notation for the non-backtracking product of two matrices.
\begin{definition}
	For two matrices $A, B \in \R^{nq \times nq}$, define $A \nb B$ to be the non-bactracking product of $A$ and $B$ by setting for all $i, j \in [n]$
	\[ (A \nb B)[i,j] = \begin{cases} \sum_{k} A[i,k]B[k,j] & \text{ if } i \neq j \\ 0 & \text{ otherwise } \end{cases} \]
	Inductively define $A^{(s)} \defeq A^{(s-1)} \nb A$
\end{definition}

Suppose $v_t \in \R^{nq}$ be the eigenvector with largest eigenvalue of $\CentAdj_{\bG}^{(t)}$, and let its eigenvalue be $\Lambda_t^t$.
By \pref{thm:local-statistics-concentrate}, we know that for each $s$, with probability $1- o_n(1)$,
\[ \lambda_{max}(\CentAdj_{\bG}^{(s)} ) \geq  \Omega_{\Model}(1) \cdot \lambda_L^s \]
Therefore,
\begin{align*}
	v_s \CentAdj_{\bG}^{(s)} v_s \geq  \Omega_{\Model}(1)  \cdot \lambda_L^s
\end{align*}

On the other hand, by the spectral norm bound for all $s \geq (\log \log n)^3$ in the null model, 
\begin{align*}
	v_s \CentBdj_{\bG}^{(s)} v_s \leq (1+o(1))^s \myrho^s
\end{align*}
For simplifying notation, we will ignore the $(1+o(1))^s$ term in the above bound, here in the rest of the section.  
Rewriting the difference we get,
\begin{align*}
	v_t^T \left( \CentAdj_{\bG}^{(t)} - \CentBdj_{\bG}^{(t)} \right) v_t = \sum_{s = 0}^{t-1} v_t^T \CentBdj_{\bG}^{(s)} \nb \CentRdj_{\bG} \nb \CentAdj_{\bG}^{(t-s-1)} v_t
\end{align*}

Now we can replace the non-backtracking product in the above expression with the usual matrix product using \pref{lem:nbproductToNormal}.  
\begin{lemma} \label{lem:nbproductToNormal}
	For all $A, B, R \in \R^{nq \times nq}$, 
	\begin{align*}
\norm{A^{(s)} \nb R \nb B^{(t)} - A^{(s)} R B^{(t)}} 
 & \leq q \norm{R}_\infty \cdot 
		\left(\norm{A^{(s)}}\norm{B^{(t-s)}} + \norm{A}_{1 \to 1} \norm{A^{(s-1)}} \norm{B^{(t)}} + \norm{B^T}_{1 \to 1} \norm{A^{(s)}} \norm{B^{(t-1)}}\right)
	\end{align*}
where $\norm{R}_\infty = \max_{\ell, \ell' \in [nq]} |R_{\ell,\ell'}|$
\end{lemma}
We will postpone the proof of this Lemma to later in the section, and proceed with the argument.

Notice that under the condition in \pref{rem:RG-condition}, 
\begin{equation}
\norm{\CentRdj_{\bG}}_\infty \leq O_{\Model}(1) \cdot \frac{1}{n}
\end{equation}
The maximum degree of a variable in the factor graph is $O(\log n)$ with probability $1 - o_n(1)$.  Therefore a naive bound on the spectral norm of $\CentAdj_{\bG}^{(s)}$  would be
\begin{equation}
\norm{\CentAdj_{\bG}^{(s)}} \leq O(\log n)^s \leq o(n^{1/4})
\end{equation}
for $s \leq o(\log n/\log \log n)$.  Similarly, we can bound $\norm{\CentBdj_{\bG}^{(t)}} \leq o(n^{1/4})$.  Using these bounds in \pref{lem:nbproductToNormal}, we can replace non-backtracking product by the usual product to conclude,
\begin{align*}
	v_t^T \left( \CentAdj_{\bG}^{(t)} - \CentBdj_{\bG}^{(t)} \right) v_t = \sum_{s = 0}^{t-1} v_t^T \CentBdj_{\bG}^{(s)}  \CentRdj_{\bG} \CentAdj_{\bG}^{(t-s-1)} v_t + o_n(1)
\end{align*}

We will now rewrite the matrix $\CentRdj_{\bG}$ explicitly in terms of the coloring $\bc$.  To this end, we make a few definitions.
For types $\tau, \tau' \in T$ and colors $\alpha, \alpha' \in [q]$ define $\Gamma_{\alpha,\alpha'}^{\tau,\tau'} \in \R^{[q] \times [q]}$ as,
\begin{align*}
\Gamma_{\alpha,\alpha'}^{\tau,\tau'} \defeq  \sum_{e\in \CompG_{n}, e \ni i j} \BPM{e}{i}{j} \cdot \Pr_{\Model}[e \in \bG | \btau(i) = \tau,\btau(j) = \tau', \bc(i) = \alpha,\bc(j) = \alpha', \vec{y}_{\bG}])
\end{align*}
In terms of the matrices $\{ \Gamma_{\alpha,\alpha'}^{\tau,\tau'}\}$ we can write for $i \neq j$,
\begin{align*}
\CentRdj_{\bG}[i,j] = \sum_{\tau,\tau' \in T} \sum_{\alpha,\alpha' \in [q]} \Gamma^{\tau, \tau'}_{\alpha,\alpha'} \cdot \Ind[\btau(i) = \tau] \Ind[\btau(j) = \tau'] \left(\Ind[\bc(i) = \alpha] \Ind[\bc(j) = \alpha']  - \Pr[\bc(i) = \alpha| \btau(i),\vec{y}_{\bG}] \Pr[\bc(j) = \alpha' | \btau(j),\vec{y}_{\bG}] \right)
\end{align*}

For every type $\tau$ and a color $\alpha$, let $\bchat^{\tau,\alpha}, \bc^{\tau,\alpha}, \bmu^{\tau,\alpha} \in \R^{n}$ be defined as follows:
\begin{align*}
\bchat^{\tau,\alpha}[i] & \defeq \Ind[\tau(i) = \tau] \cdot \left(\Ind[\bc(i) = \alpha] - \Pr[\bc(i) = \alpha| \btau(i) = \tau,\vec{y}_{\bG}] \right) \\
\chi^{\tau,\alpha}[i] & \defeq \Ind[\tau(i) = \tau] \cdot \left(\Ind[\bc(i) = \alpha]\right) \\
\bmu^{\tau,\alpha}[i] & \defeq \Ind[\btau(i) = \tau] \cdot \Pr[\bc(i) = \alpha| \btau(i) = \tau,\vec{y}_{\bG}] 
\end{align*}
Hence for $i \neq j$ we have,
\begin{align*}
\CentRdj_{\bG}[i,j] = \sum_{\tau,\tau' \in T} \sum_{\alpha,\alpha' \in [q]} \Gamma^{\tau, \tau'}_{\alpha,\alpha'} \cdot \left(\chi^{\tau,\alpha}[i] \bchat^{\tau',\alpha'}[j] + \bchat^{\tau,\alpha}[i] \bmu^{\tau',\alpha'}[j] \right)
\end{align*}
Define $\CentRdj[i,i]$ so that we have the equality,
\begin{align} \label{eq:Rexpression}
\CentRdj_{\bG} = \sum_{\tau,\tau',\alpha,\alpha'} \Gamma^{\tau,\tau'}_{\alpha,\alpha'} \otimes \left( \chi^{\tau,\alpha} (\bchat^{\tau',\alpha'})^T + \bchat^{\tau,\alpha} (\bmu^{\tau',\alpha'})^T\right)
\end{align}

Using \eqref{eq:Rexpression} for the matrix $\CentRdj$, we can write,
\begin{align*}
	v_t^T \left( \CentAdj_{\bG}^{(t)} - \CentBdj_{\bG}^{(t)} \right) v_t = 
	\sum_{\tau,\tau',\alpha,\alpha'}  v_t^T  \sum_{s = 1}^t  \CentBdj_{\bG}^{(s)}  \left( \Gamma^{\tau,\tau'}_{\alpha,\alpha'} \otimes \left( \chi^{\tau,\alpha} (\bchat^{\tau',\alpha'})^T + \bchat^{\tau,\alpha} (\bmu^{\tau',\alpha'})^T \right)  \right) \CentAdj_{\bG}^{(t-s-1)} v_t + o_n(1)
\end{align*}

The second term corresponding to $\bchat^{\tau,\alpha} (\bmu^{\tau',\alpha'})^T$ is negligible.  Specifically, we will prove the following Lemma.
\begin{lemma} \label{lem:annoy}
With probability $1- o_n(1)$ the following holds, for all $\tau, \alpha, \tau',\alpha'$
For all $1 \leq s,t \leq \sqrt{\log n}$,

\begin{align*}
	\norm{\CentBdj_{\bG}^{(s)} \left( \Gamma_{\alpha,\alpha'}^{\tau,\tau'} \otimes  \bchat^{\tau,\alpha} (\bmu^{\tau',\alpha'})^T \right) \CentAdj_{\bG}^{(t)}} \leq \myrho^{s+t} \cdot O_{\Model}(1) 
\end{align*}

\end{lemma}
We postpone the proof to later in the section and proceed with the main argument.

Using \pref{lem:annoy},  we can drop all terms arising from $\bchat^{\tau,\alpha}(\bmu^{\tau',\alpha'})^T$ by losing less than $ O_{\Model}(1) \sum_{s = 0}^{t-1} \myrho^s$.  Since $\lambda_L > (1+\delta)^4$, for all $t$ larger than a fixed constant $\Theta_M(1)$, this sum $ O_{\Model}(1) \sum_{s = 0}^{t-1} \myrho^s \leq \frac{1}{10}\lambda_L^t$.
Therefore, we arrive at our inequality,
\begin{align*}
	\sum_{\tau,\tau',\alpha,\alpha'}  v_t^T  \sum_{s = 0}^{t-1}  \CentBdj_{\bG}^{(s)} \left( \Gamma^{\tau,\tau'}_{\alpha,\alpha'} \otimes \ \chi^{\tau,\alpha}  (\bchat^{\tau',\alpha'})^T \right) \CentAdj_{\bG}^{(t-s-1)} v_t \geq  0.9 \cdot \Lambda_t^t 
\end{align*}
Let us write the matrix $\Gamma_{\alpha,\alpha'}^{\tau,\tau'} = \sum_{\beta,\beta' \in [q]} \Gamma_{\alpha,\alpha'}^{\tau,\tau'} [\beta,\beta'] \cdot \bfe_\beta (\bfe_{\beta'})^T$ where $\bfe_{\beta}, \bfe_\beta'$ are standard basis vectors in $\R^q$.
Note that $\Gamma_{\alpha,\alpha'}^{\tau,\tau'}$ has entries that are $O_{\Model}(1)/n$.  
There must exist some choice of $\tau,\tau',\alpha,\alpha',\beta,\beta'$ such that,

\begin{align*}
	v_t^T  \sum_{s = 0}^{t-1}  \CentBdj_{\bG}^{(s)} \left( \bfe_{\beta} \otimes \ \chi^{\tau,\alpha}  (\bfe_{\beta'} \otimes \bchat^{\tau',\alpha'})^T \right) \CentAdj_{\bG}^{(t-s-1)} v_t \geq  \Omega_{\Model}(1) \cdot \Lambda_t^t \cdot n
\end{align*}
where $\Omega_\Model(1)$ hides a constant depending on the model $\Model$.  Rewriting the above inequality,
\begin{align*}
\Omega_{\Model}(1) \cdot \Lambda_t^t \cdot n 
& \leq 
  \sum_{s = 0}^{t-1} \Iprod{ v_t^T ,  \CentBdj_{\bG}^{(s)}  \bfe_{\beta} \otimes \ \chi^{\tau,\alpha}} \Iprod{ \bfe_{\beta'} \otimes \bchat^{\tau',\alpha'},  \CentAdj_{\bG}^{(t-s-1)} v_t} \\
  & \leq 
  \sum_{s = 0}^{t-1} \norm{\CentBdj_{\bG}^{(s)}  \bfe_{\beta} \otimes \ \chi^{\tau,\alpha}} \left | \Iprod{ \bfe_{\beta'} \otimes \bchat^{\tau',\alpha'},  \CentAdj_{\bG}^{(t-s-1)} v_t} \right | 
\end{align*}
Using \pref{lem:normstatistic} on the fixed vector $\bfe_\beta \otimes \chi^{\tau,\alpha}$ and the planted distribution we get that with probability $1-o_n(1)$,
\begin{align*}
\Omega_{\Model}(1) \cdot \Lambda_t^t \cdot n 
  & \leq 
  \sum_{s = 0}^{t-1} \myrho^s \cdot \sqrt{n} \left| \Iprod{ \bfe_{\beta'} \otimes \bchat^{\tau',\alpha'},  \CentAdj_{\bG}^{(t-s-1)} v_t} \right|
\end{align*}

For notational convenience, let us reparametrize $s \to t-s$ and conclude,
\begin{align*}
\Omega_{\Model}(1) \cdot \Lambda_t^t \cdot \sqrt{n}
  & \leq 
  \sum_{s = 1}^{t} \myrho^{t-s}  \left| \Iprod{ \bfe_{\beta'} \otimes \bchat^{\tau',\alpha'},  \CentAdj_{\bG}^{(s-1)} v_t}\right| 
\end{align*}
With high probability, the maximum degree of a variable is $O(\log n) \ll o(\log^2 n)$, and therefore $\norm{\CentAdj^{(s)}} \leq o(\log n)^{2s}$.

Since $\Lambda_t \geq \lambda_L \geq (1+\delta) \myrho $ we can bound the terms for $s = 1,\ldots, (\log\log n)^3$ as follows,
\begin{align*}
 \sum_{s = 1}^{ (\log \log n)^3 } \myrho^{t-s} \left| \Iprod{ \bfe_{\beta'} \otimes \bchat^{\tau',\alpha'},  \CentAdj_{\bG}^{(s-1)} v_t} \right|  \leq \sum_{s = 1}^{ (\log \log n)^3 } \frac{\Lambda_t^{t-s}}{(1+\delta)^{t-s}} \cdot o((\log n)^{2s})  \cdot \sqrt{n} \leq o(1) \cdot \Lambda_t^t \cdot \sqrt{n}
\end{align*}
where the last inequality holds for $t > (\log \log n)^4$.
Deleting terms for small $s$, we have the correlation inequality,
\begin{align} \label{eq:correq}
\Omega_{\Model}(1) \cdot \Lambda_t^t \cdot \sqrt{n} 
  & \leq 
  \sum_{s = (\log \log n)^3}^{t} \myrho^{t-s} \left | \Iprod{ \bfe_{\beta'} \otimes \bchat^{\tau',\alpha'},  \CentAdj_{\bG}^{(s-1)} v_t} \right |
\end{align}
for $t > (\log \log n)^4$.

For each $s$, recall that $v_s$ is the top eigenvector of $\CentAdj_{\bG}^{(s)}$, and $\Lambda^s_s$ denotes the largest eigenvalue.
\begin{claim} \label{claim:whatever}
	There exists a $\bfm \in [(\log\log n)^{4.5}, (\log\log n)^5 ]$ such that, for all $s \in [(\log \log n)^3, \bfm]$,
	\[ \norm{A^{(s)} v_{\bfm}} \leq \Lambda_{\bfm}^s (1+\delta)^{t-s}\]
\end{claim}
Before we see the proof of above claim, let us see how it leads to an algorithm.
Applying \eqref{eq:correq} for this choice of $\bfm$, we conclude that
\begin{equation} \label{eq:correq3}
 \sqrt{n} \cdot \Omega_{\Model}(1) \cdot \Lambda_{\bfm}^{\bfm} \leq  \sum_{s = (\log \log n)^3}^{\bfm} \myrho^{\bfm-s} \left | \Iprod{ \bfe_{\beta'} \otimes \bchat^{\tau',\alpha'},  \CentAdj_{\bG}^{(s-1)} v_{\bfm}} \right | 
\end{equation}
In \eqref{eq:correq3}, we can bound the sum of all terms with $s < t^{*} - C$ as follows.
\begin{align} 
  \sum_{s = (\log \log n)^3}^{\bfm - C} \myrho^{\bfm-s} \left| \Iprod{ \bfe_{\beta'} \otimes \bchat^{\tau',\alpha'},  \CentAdj_{\bG}^{(s-1)} v_t} \right | & \leq 
    \sum_{s = (\log \log n)^3}^{\bfm - C} \myrho^{\bfm-s} \norm{\bfe_{\beta'} \otimes \bchat^{\tau',\alpha'}} \left \| \CentAdj_{\bG}^{(s-1)} v_t \right \|   \\
    & \leq  \sum_{s = (\log \log n)^3}^{\bfm - C} \myrho^{\bfm-s} \Lambda_{\bfm}^{s-1} \cdot \sqrt{n} \cdot  (1+\delta)^{\bfm-s+1}\\
    & \leq  (1+\delta) \Lambda_{\bfm}^{\bfm} \sqrt{n} \cdot  \left( \sum_{s = (\log \log n)^3}^{\bfm - C} \left(\frac{\myrho (1+\delta)}{\Lambda_{\bfm} } \right)^{\bfm - s} \right )\\
 &   =  \Lambda_{\bfm}^{\bfm} \cdot \sqrt{n} \cdot \left(\frac{1}{\delta (1+\delta)^{C-2}}\right)
 \label{eq:this-equation}
 \end{align}
Using \eqref{eq:correq3} and \eqref{eq:this-equation}, for sufficiently large $C = O_{\Model,\delta}(1)$, we conclude that
\begin{align} \label{eq:correq4}
 \Omega_{\Model}(1) \cdot \Lambda_{\bfm}^{\bfm} \cdot \sqrt{n} \leq  \sum_{s = \bfm - C+1}^{\bfm -1} \myrho^{\bfm-s} \left | \Iprod{ \bfe_{\beta'} \otimes \bchat^{\tau',\alpha'},  \CentAdj_{\bG}^{(s-1)} v_{\bfm}} \right | 
\end{align}
Note that there are only $C = O_{\Model,\delta}(1)$ terms in the sum, so one of them is large.  In particular, there exists some $\ell \in [\bfm-C, \bfm]$ such that if we set
\[ w_\ell \defeq \CentAdj_{\bG}^{(\ell)} v_{\bfm}\]
then,
\[ \left | \Iprod{ \bfe_{\beta'} \otimes \bchat^{\tau',\alpha'},  w_\ell} \right | \geq \frac{1}{C \myrho^{C}} \cdot \Lambda_{\bfm}^{\bfm} \cdot \Omega_{M,\delta}(1) \cdot \sqrt{n} \ ,\]
But note that by the choice of $\bfm$, 
\[ \norm{w_\ell} \leq \Lambda_{\bfm}^{\ell} \cdot (1+\delta)^{\bfm - \ell}  \leq \Lambda_{\bfm}^{\bfm} \cdot (1+\delta)^{C}\ .\]
So there exists some $\ell \in [\bfm - C, \bfm]$, such that the unit vector $\overline{w}_\ell = \frac{w_\ell}{\norm{w_\ell}}$ satisfies,
\[ \left | \Iprod{ \bfe_{\beta'} \otimes \bchat^{\tau',\alpha'},  \overline{w}_\ell} \right | \geq  \Omega_{M,\delta}(1) \cdot \sqrt{n} \ .\]

In other words, for some choice of $\ell \in [\bfm - C, \bfm]$ and $\beta' \in [q]$, if we construct $\overline{u} \in \R^n$ as,
\[  \overline{u}[i] = \overline{w}_\ell[i,\beta']\] 
then $| \iprod{\overline{u}, \bchat^{\tau',\alpha'}} | \geq \Omega_{\Model,\delta}(1) \cdot \sqrt{n}$.  This finishes the proof of \pref{thm:recovery-full}. \footnote{Note that the $\bchat^{\tau',\alpha'}$ here is a bit different from the $\bchat^{\tau',\alpha'}$ defined in the theorem, but by \pref{rem:RG-condition} they are within a multiplicative factor of $(1\pm \epsilon)$ from each other. Thus the inequality still hold for the $\bchat^{\tau',\alpha'}$ defined in the theorem statement.}

%
%

\begin{proof} (Proof of \pref{claim:whatever})
 The idea behind the proof is a descent/bootstrap argument to get a contradiction.  Let us start with $t = (\log \log n)^5$ as the guess for $\bfm$.  If current value of $t$ satisfies the condition of the claim, we are done.  Otherwise, there exists $s< t$ such that,

\[ \norm{\CentAdj_{\bG}^{(s)} v_t} \geq \Lambda_t^s (1+\delta)^{t-s}\]
This implies that,
\[\Lambda^s_s \geq \Lambda_t^s (1+\delta)^{t-s}\]
or equivalently,
\[ \log \Lambda_s \geq \log \Lambda_t + \log (1+\delta) \cdot (t-s)/s\]
Suppose we use $s$ as the new candidate for $\bfm$ and recurse.  Let us suppose we iteratively construct a sequence of $t_0 = (\log\log n)^5 > \ldots > t_r $ in this manner.  The value of $\log \Lambda_{t_i}$ increases along the sequence.  By Fact~\pref{fac:silly}, if we obtain a sequence of $t_0 = (\log\log n)^5 > \ldots > t_r = (\log \log n)^{4.5}$ then we will have,
\[ \log \Lambda_r \geq (\log t_0 - \log t_r)\log(1+\delta) \geq \log(1+\delta) \cdot \Omega(\log \log \log n) \]
This suggests that $\norm{\CentAdj_{\bG}^{(t_r)}}^{1/t_r} \geq \omega(1)$ for some $t_r = \Omega(\log \log n)^{4.5}$.  A contradiction, since with probability $1- o_n(1)$, we will have $\norm{\CentAdj_{\bG}^{(t_r)}}^{1/t_r}  = O(1)$.  This follows from the fact that with probability $1- o_n(1)$, degree of every vertex in $A_{\bG}^{(s)}$ is at most $O(D^s)$ for some constant $D$ for all $s > (\log \log n)^2$.  Therefore, the sequence terminates and we find a $t_r \in [(\log \log n)^{4.5}, (\log \log n)^5]$, implying the claim.

\begin{fact} \label{fac:silly}
	Given a sequence of positive integers, $a_1 \geq a_2 \geq \ldots a_r$,
	\begin{align}
	\sum_{i = 1}^r \frac{a_i - a_{i-1}}{a_{i-1}} \geq \sum_{i = 1}^r \sum_{x = a_{i-1}}^{a_i-1} \frac{1}{x} = \sum_{x = a_r}^{a_1} \frac{1}{x} \approx \ln (a_1) - \ln (a_r)
	\end{align}
\end{fact}
\end{proof}

\begin{proof} (Proof of \pref{lem:nbproductToNormal})
	For $s \in \N$, let $\calP_s$ be the set of length $s$ non-backtracking walks in complete graph $\CompG_{n}$.
	So $\ualpha = (\alpha_0, \alpha_1,\ldots,\alpha_s) \in \calP_s$ will be a non-backtracking path in $\CompG_{n}$ with vertices $\alpha_0,\ldots,\alpha_s$.

	For $\ell, \ell' \in [n]$, we can write
	\begin{align}
	A^{(s)} R B^{(t)}[\ell,\ell'] = \sum_{ \substack{\ualpha \in \calP_s , \ubeta \in \calP_t \\ \alpha_0 = \ell, \beta_t = \ell' }}. \prod_{i=1}^s A_{\alpha_{i-1} \alpha_i}  \cdot R_{\alpha_s, \beta_0} \cdot  \prod_{j = 1}^t B_{\beta_{j-1} \beta_j} 
	\end{align}

	\begin{align}
	A^{(s)} \nb R \nb B^{(t)}[\ell,\ell'] = \sum_{ \substack{\ualpha \in \calP_s , \ubeta \in \calP_t \\ \alpha_0 = \ell, \beta_t = \ell' \\ \alpha_s \neq \beta_0, \alpha_{s-1} \neq \beta_0, \alpha_{s} \neq \beta_1 }}. \prod_{i=1}^s A_{\alpha_{i-1} \alpha_i}  \cdot R_{\alpha_s, \beta_0}  \prod_{j = 1}^t B_{\beta_{j-1} \beta_j} 
	\end{align}

	It is clear that the difference $A^{(s)} \nb R \nb B^{(t)} - A^{(s)} R B^{(t)}$ consists of three different terms.

	{\bf Term 1:  $\alpha_s = \beta_0$}
		Consider the block-diagonal matrix $\calD_1 \in \R^{nq \times nq}$ given by, 
		\[ \calD_1[i,j] = \Ind[i = j] \cdot R_{ij}.\]
		then we can write this term as $A^{(s)} \calD_1 B^{(t)}$.  Hence we get the following bound,
		\[ \norm{A^{(s)} \calD_1 B^{(t)}} \leq q \norm{R}_\infty \cdot \norm{A^{(s)}} \cdot \norm{B^{(t-s)}} \]

	{\bf Term 2: $\alpha_{s-1} = \beta_0$}
		Consider the block-diagonal matrix $\calD_2 \in \R^{nq \times nq}$ given by,
		\[ \calD_2[i,i] =   \sum_{j} A_{ij} R_{ji} .\]
		then we can write this term as $A^{(s-1)} \calD_2 B^{(t)}$.  We have the upper bound,
		\[ \norm{\calD_2} \leq q \norm{R}_\infty  \norm{A}_{1 \to 1} \]
		which implies that 
\[\norm{A^{(s-1)} \calD_2 B^{(t)}} \leq q \norm{x}_{\infty} \norm{y}_\infty \norm{A}_{1 \to 1} \norm{A^{(s-1)}} \norm{B^{(t)}} \]

	{\bf Term 3}:  $\alpha_s = \beta_1$
		Consider the block-diagonal matrix $\calD_3 \in \R^{nq \times nq}$ given by,
		\[ \calD_3[i,i] =   (\sum_j R_{ij} B_{ji}) .\]
		then we can write this term as $A^{(s)} \calD_3 B^{(t-1)}$.  Analogous to the previous case, we get an upper bound of,
\[\norm{A^{(s)} \calD_3 B^{(t-1)}} \leq q \norm{R}_\infty \norm{B^T}_{1 \to 1} \norm{A^{(s)}} \norm{B^{(t-1)}} \]

	Adding the three terms, we have the claim of the lemma.

\end{proof}

We will need the following theorem about local statistics on the expectation and concentration of local statistics in order to complete the proof of \pref{lem:annoy}.

\begin{lemma} \label{lem:normstatistic}
	Fix a vector $x \in \R^{nq}$, such that $\norm{x}_\infty = O(1)$.  With probability $1-o_n(1)$, for all $1 \leq s \leq \sqrt{\log n}$ we have,
	\begin{align}
		\| \CentAdj_{\bG}^{(s)}x \| \leq C\myrho^s \cdot \sqrt{n}
	\end{align}
	for an absolute constant $C\ge 1$.
\end{lemma}
\begin{proof}[Proof sketch]
	This is equivalent to proving $\langle \CentAdj_{\bG}^{(s)}x,\CentAdj_{\bG}^{(s)}x\rangle \le C\myrho^{2s}n$.  Indeed, this quantity can be rewritten as:
	\[
		\langle \left(\CentAdj_{\bG}^{(s)}\right)^2, xx^{\top}\rangle.
	\]
	Via similar calculations to the ones done in \pref{sec:statistics}, we can show that this quantity is dominated by the contribution of walks that are self-avoiding for the first $s$ steps, and retrace the same steps taken in the next $s$ steps, which in turn can be used to show that this quantity concentrates around the expected total weight of walks in the associated random tree that walk out $s$ steps and walk back $s$ steps, and hence for large enough $s$ is at most $C\myrho^{2s}n$ for an absolute constant $C$.
\end{proof}

\prasadtemp{to do, include the proof}

\begin{proof}[Proof of \pref{lem:annoy}]
	Let $\Gamma_{\alpha,\alpha'}^{\tau,\tau'} = \sum_{j \in [q]} u_j v_j^T$ be the singular decomposition of $\Gamma_{\alpha,\alpha'}^{\tau,\tau'}$.
	\begin{align} \label{eq:233}
		\norm{\CentBdj_{\bG}^{(s)} \Gamma_{\alpha,\alpha'}^{\tau,\tau'} \otimes \bchat^{\tau,\alpha} (\bmu^{\tau',\alpha'})^T \CentAdj_{\bG}^{(t)}} 
&		\leq \sum_{j} \norm{\CentBdj_{\bG}^{(s)} u_j \otimes \bchat^{\tau,\alpha}} \norm{ v_j \otimes (\bmu^{\tau',\alpha'})^T \CentAdj_{\bG}^{(t)}} 
	\end{align}
Applying \pref{lem:normstatistic} to the planted model $\Model$ with the fixed vector $u_j \otimes \bchat^{\tau,\alpha}$, we conclude that with probability $1-o_n(1)$, 
\begin{align} \label{eq:234}
\norm{\CentBdj_{\bG}^{(s)} u_j \otimes \bchat^{\tau,\alpha}} \leq  C\myrho^{s} \cdot \norm{u_j \otimes \bchat^{\tau,\alpha}} \leq  \myrho^{s} \cdot \norm{u_j} \cdot n^{1/2}
\end{align}

Similarly, applying \pref{lem:normstatistic} to the null model $\Model^{\times}$ with fixed vector $v_j \otimes \mu^{\tau',\alpha'}$, we conclude that with probability $1-o_n(1)$, for all $1 \leq t \leq \sqrt{\log n}$,
\begin{align} \label{eq:235}
\norm{ v_j \otimes (\bmu^{\tau',\alpha'})^T \CentAdj_{\bG}^{(t)} } \leq  C\myrho^{s} \cdot \norm{ v_j \otimes (\bmu^{\tau',\alpha'})^T } \leq  \myrho^{t} \cdot \norm{v_j} \cdot n^{1/2}
\end{align}

Finally, note that $\sum_{j} \norm{u_j} \norm{v_j} = \norm{\Gamma_{\alpha,\alpha'}^{\tau,\tau'}}_{Fr} = O_{\Model}(\frac{1}{n})$ where $O_{\Model}$ hides a fixed constant depending on the model.
Using \eqref{eq:234} and \eqref{eq:235} in \eqref{eq:233}, we conclude the proof.
\end{proof}

\section*{Acknowledgments}
We would like to thank Tselil Schramm for insightful discussions around belief propagation, and we would like to thank Ryan O'Donnell and Xinyu Wu for conversations that brought clarity to various parts of \cite{BC19}.  We would also like to thank Lenka Zdeborov{\'a} and anonymous reviewers for bringing several relevant references to our attention.

\bibliographystyle{alpha}
\bibliography{main,sdp-hds}

\appendix
\section{Belief propagation for $\Model$}\label{app:BP-update-rule}
We briefly describe the belief propagation (BP) algorithm that aims to estimate the mariginal distribution of $\bc(v)$, $v\in[n]$ under the Boltzmann distribution $\boltz$ with Hamiltonian $\Ham$. Define the messages $\{m^{v\to e}_c\}_{c\in [\numcols]}$ that a variable $v$ passes to some constraint $e\in\bE_i$, and the messages $\{m^{e\to u}_c\}_{c\in [\numcols]}$ that a constraint $e \in \bE_i$ passes to a variable $u$. Intuitively speaking, $m^{v\to e}_c$ is an estimate of the marginal probability that $v$ is assigned the color $c$ when the constraint $e$ is absent, and $m^{e\to u}_c$ is an estimate of the marginal probability that $u$ has color $c$ when all other constraints involving $u$ are absent. Since the distribution of $\bc(v)$ under $\mu$ depends on the constraints that contain $v$, we only focus on the messages $m^{v\to e}$, $m^{e\to u}$ such that $v\in\partial e$ (i.e. $e$ contains $v$) and $e\in\partial u$ (i.e. $e$ contains $u$).

\begin{equation}\label{eq:v-to-e-msg}
m^{v\to e}_c[t+1] = \frac{1}{Z^{v\to e}}\ColDist_{\btau(v)}(c)\prod_{f\in \partial v\setminus e} m^{f\to v}_{c}[t] \enspace,
\end{equation}
where $Z^{v\to e} = \sum_{c\in [\numcols]} \ColDist_{\btau(v)}(c)\prod_{f\in \partial v\setminus e} m^{f\to v}_{c}[t]$ .

For factors $e\in\bE_i$, the messages are defined as

\begin{equation}\label{eq:e-to-u-msg}
m^{e\to u}_c[t+1] = \frac{1}{Z^{e\to u}} \sum_{\bc_{e}\mid \bc_{e}(u) = c}\phi_i(\bc_e)\prod_{v\in e\setminus u} m^{v\to e}_{\bc_e(v)}[t]\enspace, 
\end{equation}
where $Z^{e\to u} = \sum_{c\in [\numcols]} \sum_{\bc_{e}\mid \bc_{e}(u) = c} \phi_i(\bc_e)\prod_{v\in e\setminus u} m^{v\to e}_{\bc_e(v)}[t]$.

To obtain an estimate of the marginal probability of the assignment to a variable $v$, apply the message update rules until reaching some fixed point $\{\hat{m}^{v\to e}_c, \hat{m}^{e\to u}_c\}_{c\in [\numcols]}$. The estimate is called the belief and is given by

\[b^{v}_c = \frac{1}{Z^{v}}\ColDist_{\btau(v)}(c)\prod_{f\in \partial v} \hat{m}^{f\to v}_{c}\enspace, \]

where $Z^{v} = \sum_{c\in [\numcols]} \ColDist_{\btau(v)}(c)\prod_{f\in \partial v} \hat{m}^{f\to v}_{c}$.

\section{Proof of \pref{lem:easy-case}}  \label{app:easy-cases}

Recall that:
\begin{align}	\label{eq:average-factor-nums-match}
	\ACFacDeg{i} = \sum_{(c_1,\dots,c_{\arity(i)})\in [q]^{\arity(i)}} \left(\prod_{k=1}^{\arity(i)}\ColDist_{\Cl(i)_k}(c_k)\right) \cdot \phi_i(c_1,\dots,c_{\arity(i)})
\end{align}
We first explain how to solve the distinguishing problem and then explain the recovery algorithm.
By definition of the BP update functions $\Upsilon_{v\to e}$ (equation \pref{eq:v-to-e-msg}) and $\Upsilon_{e\to v}$ (equation \pref{eq:e-to-u-msg}), the set of trivial messages $\minit$ being a BP fixed point is equivalent to:
\begin{align*}
	\ACFacDeg{i} = \sum_{\substack{(c_1,\dots,c_{\arity(i)})\in [q]^{\arity(i)}:\\ c_j=c}} \left(\prod_{k\ne j}\ColDist_{\Cl(i)_k}(c_k)\right)\cdot \phi_i(c_1,\dots,c_{\arity(i)}) & & \forall c\in [q]:\ColDist_{\Cl(i)_j}(c) > 0	\numberthis \label{eq:detailed-balance}
\end{align*}
If the set of trivial messages is not a fixed point of the belief propagation update rule: then there exist $c\in [q]$, $i\in[F]$, and $j\in[\arity(i)]$ with $\ColDist_{\Cl(i)_j}(c) > 0$ such that
\[
	\ACFacDeg{i} \ne \sum_{\substack{(c_1,\dots,c_{\arity(i)})\in [q]^{\arity(i)}:\\ c_j=c}} \left(\prod_{k\ne j}\ColDist_{\Cl(i)_k}(c_k)\right)\cdot \phi_i(c_1,\dots,c_{\arity(i)}).
\]
Let $\varfac{v}{i,j}$ be the number of type-$i$ factors with variable $v$ in the $j$-th position.  Via standard results for $\Poi(d)$ approximating $\mathrm{Binom}(n,d/n)$ we have the following:
\begin{itemize}
	\item In $\Model^{\times}$, for any variable $v$ of type $\Cl(i)_j$ and any constant $T$,\sidtemp{Find citation for Poisson approximating binomial and justify, maybe via an appendix lemma}
	\[
		\E{\varfac{v}{i,j}}^T = \E \bX^T \pm o_n(1)
	\]
	where $\bX\sim\Poi(\lambda)$ and $\lambda = \ACFacDeg{i}$.
	\item On the other hand, in the planted model $\Model$:
	\[
		\E{\varfac{v}{i,j}}^T = \E \bY^T \pm o_n(1)
	\]
	where $\bY$ is distributed as the mixture of Poisson distributions $p_1\Poi(\lambda_1)+\dots+p_{s}\Poi(\lambda_{s})$ where $s$ is the number of colors which vertex $v$ has nonzero probability of attaining, not all $\lambda_i$ are equal, and all $p_i > 0$.
\end{itemize}
By \pref{eq:average-factor-nums-match} $p_1\lambda_1+\dots+p_s\lambda_s = \lambda$.  We first recall the following well known fact about Poisson random variables. \sidtemp{Maybe having a preliminaries section and adding it there might be more appropriate?} \sidtemp{cite?} 
\begin{fact}	\label{fact:second-moment-poisson}
	If $\bA\sim\Poi(\mu)$, then $\E\bA^2 = \mu^2+\mu$.
\end{fact}
As a consequence of \pref{fact:second-moment-poisson}: $\E\bX^2 = \lambda^2 + \lambda$, and $\E\bY^2 = p_1(\lambda_1^2+\lambda_1) + \dots + p_s(\lambda_s^2 + \lambda_s)$.
\[
	\E\bY^2 - \E\bX^2 = p_1f(\lambda_1) + \dots p_s f(\lambda_s) - f(\lambda) = p_1f(\lambda_1) + \dots p_s f(\lambda_s) - f(p_1\lambda_1+\dots+p_s\lambda_s)
\]
Since not all $\lambda_i$ are equal, all $p_i>0$ and $f$ is strictly convex, $\E\bY^2 - \E\bX^2$ is equal to a constant $\delta$ strictly greater than $0$.  Suppose $n_{i,j,2}(\bG)\coloneqq\E\sum_{v\in [n]}{\varfac{v}{i,j}}^2$, then $|\E_{\bG\sim\Null} n_{i,j,2}(\bG) - \E_{\bG\sim\Planted}n_{i,j,2}(\bG)|\ge \Omega(n)$.  Since $\E\bY^4$ and $\E\bX^4$ are constants, the variance of $n_{i,j,2}(\bG)$ is $O(n)$ for both $\bG\sim\Null$ and $\bG\sim\Planted$. This informs using the following polynomial time distinguisher:
\begin{displayquote}
	Compute $n_{i,j,2}(\bG)$ and if $|n_{i,j,2}(\bG)-\E_{\bG\sim\Null}[n_{i,j,2}(\bG)]|<|n_{i,j,2}(\bG)-\E_{\bG\sim\Planted}[n_{i,j,2}(\bG)]|$ output ``null''; otherwise output ``planted''.
\end{displayquote}


We now discuss performing recovery.  Recall the inner product $\langle \cdot, \cdot\rangle_{\innerp}$ from \pref{sec:spectral} which is defined as follows:
First, we define a $nq\times nq$-dimensional positive diagonal matrix $\innerp_{\btau}$ where the $(v,v)$ block is equal to:
\[
	\innerp_{\btau,(v,v)}[c,c] \coloneqq
	\begin{cases}
		\ColDist_{\btau(v)}(c) & \text{if $\ColDist_{\btau(v)}(c)>0$} \\
		1 & \text{otherwise.}
	\end{cases}
\]
The inner product on $\R^{nq}$ is then:
\[
	\langle x, y\rangle_{\innerp} \coloneqq x^{\top}\innerp_{\btau}^{-1}y.
\]
And let $\|\cdot\|$ denote the norm induced by the above inner product.  Let $\bc$ be the hidden coloring.  Our goal in recovery is to output a vector $v$ such that $\langle v, \bc - \E\bc|\btau\rangle_{\innerp}\ge \eps \cdot \|v\| \cdot \|\bc-\E\bc|\btau\|$.
Let $c$ and $c'$ be two colors such that:
\[
	d_c = \sum_{\substack{(c_1,\dots,c_{\arity(i)})\in [q]^{\arity(i)}:\\ c_j=c}} \left(\prod_{k\ne j}\ColDist_{\Cl(i)_k}(c_k)\right)\cdot \phi_i(c_1,\dots,c_{\arity(i)}) > \sum_{\substack{(c_1,\dots,c_{\arity(i)})\in [q]^{\arity(i)}:\\ c_j=c'}} \left(\prod_{k\ne j}\ColDist_{\Cl(i)_k}(c_k)\right)\cdot \phi_i(c_1,\dots,c_{\arity(i)}) = d_{c'}
\]
The distribution of the number of type-$i$ factors that a color $c$ vertex is part of is $\Poi(d_c)$ and similarly is $\Poi(d_{c'})$ for a color $c'$ vertex.  The following algorithm can then be shown to produce a vector $v$ meeting the aforementioned goal.
\begin{displayquote}
	For each vertex $u$ of type $\TypeDist(\Cl(i)_j)$, let $m_u$ be the number of type-$i$ factors it is part of in the $j$-th position.  If $m_u$ has a higher probability of being sampled from $\Poi(d_c)$ than $\Poi(d_{c'})$ then assign the $u$-th block of vector $v$ to be the indicator of color $c$.  Otherwise assign the $u$-th block of vector $v$ to be the indicator of color $c'$.
\end{displayquote}
Since $d_c\ne d_{c'}$ there is a constant $\eps > 0$ such that with high probability $\left(\frac{1}{2}+\eps\right)\ColDist_{\Cl(i)_j}(c)\TypeDist(\Cl(i)_j)n$ variables of color $c$ and type $\Cl(i)_j$ are assigned the correct color and also $\left(\frac{1}{2}+\eps\right)\ColDist_{\Cl(i)_j}(c)\TypeDist(\Cl(i)_j)n$ variables of color $c'$ and type $\Cl(i)_j$ are assigned the correct color.  Consequently:
\[
	\langle v, \bc - \E\bc|\btau\rangle_{\innerp}\ge \eps' \cdot \|v\| \cdot \|\bc-\E\bc|\btau\|
\]
for some $\eps'>0$.

\section{The partial derivative matrix}\label{app:calculus}

Recall the BP update function $\Gamma$ defined by equations \pref{eq:v-to-e-msg} and \pref{eq:e-to-u-msg}.

We observe that by definition
\[\bTrans_{\FacType(e_j),\index(v_j)\mid\index(v_{j+1})} = \frac{\partial \Gamma(m)^{v_j\to e_{j-1}}}{\partial m^{e_j\to v_j}}\vline_{\minit} \cdot \frac{\partial \Gamma(m)^{e_j\to v_j}}{\partial m^{v_{j+1}\to e_j}}\vline_{\minit} .\]
Thus we first compute the two derivative matrices. For any pairs of colors $c,d\in [\numcols]$,

\begin{align*}
\frac{\partial \Gamma(m)_c^{v_j\to e_{j-1}}}{\partial m_d^{e_j\to v_j}}
=& \frac{1}{Z^{v_j\to e_{j-1}}} \ColDist_{\btau(v_j)}(c)\prod_{a\in\partial v_j\setminus\{e_{j-1},e_{j}\}} m^{a\to v_j}_c \cdot\mathbf{1}_{d=c} - \frac{m^{v_j\to e_{j-1}}_c}{Z^{v_j\to e_{j-1}}}\cdot\ColDist_{\btau(v_j)}(d)\prod_{a\in\partial v_j\setminus\{e_{j-1},e_j\}} m^{a\to v_j}_d \\
=& \frac{1}{m^{e_j\to v_j}_d}\cdot\frac{1}{Z^{v_j\to e_{j-1}}} \ColDist_{\btau(v_j)}(c)\prod_{a\in\partial v_j\setminus e_{j-1}} m^{a\to v_j}_c \cdot\mathbf{1}_{d=c} - \frac{m^{v_j\to e_{j-1}}_c}{m^{e_j\to v_j}_d}\cdot\frac{\ColDist_{\btau(v_j)}(d)}{Z^{v_j\to e_{j-1}}}\prod_{a\in\partial v_j\setminus e_{j-1}} m^{a\to v_j}_d \\
=&  \frac{m^{v_j\to e_{j-1}}_c}{m^{e_j\to v_j}_d} \cdot\mathbf{1}_{d=c} - \frac{m^{v_j\to e_{j-1}}_c}{m^{e_j\to v_j}_d}\cdot m^{v_j\to e_{j-1}}_d \quad 
\end{align*}
The last equality is derived from the fixed point identity $m^{v\to e}_c= \frac{1}{Z^{v\to e}} \ColDist_{\btau(v)}(c) \prod_{a\in\partial v\setminus e} m^{a\to v}_c$.

Evaluating the derivative at the factorized fixed point gives the transformation matrix
\[
	\frac{\partial \Gamma(m)_c^{v_j\to e_{j-1}}}{\partial m_d^{e_j\to v_j}}\vline_{\minit} = \support(\ColDist_{\btau(v_j)})\cdot\left(\ColDist_{\btau(v_j)}(c)\cdot \mathbf{1}_{d=c} -  \ColDist_{\btau(v_j)}(c)\cdot\ColDist_{\btau(v_j)}(d) \right) \enspace,
\]
where $\support(\ColDist_{\btau(v_j)})$ denote the size of $\ColDist_{\btau(v_j)}$'s support. 
To write the matrix compactly we define $\Dcolor_{\tau} := \text{Diag}(\ColDist_{\tau})$, and derive from the above computation that $ \frac{\partial \Gamma(m)^{v_j\to e_{j-1}}}{\partial m^{e_j\to v_j}}\vline_{\minit} = \support(\ColDist_{\btau(v_j)})\cdot\left(\Dcolor_{\btau(v_j)} -  \ColDist_{\btau(v_j)}\ColDist_{\btau(v_j)}^T \right)$.

For any edge of the form $v_{j+1}\tot{e_j}v_{j}$ on this path where $\FacType(e_j) = \phi_i$  we have,

\begin{align*}
\frac{\partial \Gamma(m)_c^{e_j\to v_j}}{\partial m_d^{v_{j+1}\to e_j}} =& \frac{1}{Z^{e_j\to v_j}} \sum_{\bc_{e_j}\mid \bc_{e_j}(v_j,v_{j+1}) = (c,d)} \phi_i(\bc_{e_j}) \prod_{w\in e_j\setminus \{v_j,v_{j+1}\}} m^{w\to e_j}_{\bc_{e_j}(w)} \\
 & - \frac{m^{e_j\to v_j}_c}{Z^{e_j\to v_j}}\sum_{c'\in C} \sum_{\bc_{e_j}\mid \bc_{e_j}(v_j,v_{j+1}) = (c',d)} \phi_i(\bc_{e_j}) \prod_{w\in e_j\setminus \{v_j,v_{j+1}\}} m^{w\to e_j}_{\bc_{e_j}(w)} \\
=&  \frac{1}{Z^{e_j\to v_j}} \cdot \frac{1}{m^{v_{j+1}\to e_j}_d} \sum_{\bc_{e_j}\mid \bc_{e_j}(v_j,v_{j+1}) = (c,d)} \phi_i(\bc_{e_j}) \prod_{w\in e_j\setminus v_j} m^{w\to e_j}_{\bc_{e_j}(w)} \\
 & - \frac{m^{e_j\to v_j}_c}{Z^{e_j\to v_j}} \cdot \frac{1}{m^{v_{j+1}\to e_j}_d}\sum_{c'\in C}  \sum_{\bc_{e_j}\mid \bc_{e_j}(v_j,v_{j+1}) = (c',d)}  \phi_i(\bc_{e_j}) \prod_{w\in e_j\setminus v_j} m^{w\to e_j}_{\bc_{e_j}(w)} 
\end{align*}

Recall we defined a distribution $\ldist_{i}$ over $\bc_{e_j}$ and stochastic matrices $\Facm_{\FacType(e_j),\index(v_{j+1})\mid\index(v_j)}$ before.


Evaluating the derivative at the factorized fixed point gives the transformation matrix
\begin{align*}
\frac{\partial \Gamma(m)^{e_j\to v_j}}{\partial m^{v_{j+1}\to e_j}}\vline_{\minit}
=& \frac{1}{\ColDist_{\btau(v_{j+1})}(d)\cdot\support(\ColDist_{\btau(v_j)})}\left(\Pr_{\bc_{e_j}\sim D_{\theta}}[\bc_{e_j}(v_{j+1}) = d \mid \bc_{e_j}(v_j) = c]\right. \\ 
 &- \left.\frac{1}{\support(\ColDist_{\btau(v_j)})} \sum_{c'\in C}\Pr_{\bc_{e_j}\sim D_{\theta}}[\bc_{e_j}(v_{j+1}) = d\mid \bc_{e_j}(v_j) = c']\right) \enspace,
\end{align*}
Then $\frac{\partial m^{e_j\to v_j}}{\partial m^{v_{j+1}\to e_j}}\vline_{(m)^{\fix}} = \frac{1}{\support(\ColDist_{\btau(v_j)})}\left(\mathbf{I} - \frac{1}{\support(\ColDist_{\btau(v_j)})} \mathbf{1} \mathbf{1}^{\top}\right)\Facm_{\FacType(e_j),\index(v_{j+1})\mid \index(v_j)}^T \Dcolor_{\btau(v_{j+1})}^{\dagger}$ where the $\dagger$ in the superscript denotes the pseudoinverse of the matrix.

Then we can write the transformation matrix for the step $v_{j+1}\tot{e_j}e_{j-1}$ as
\begin{align*}
\bTrans_{\FacType(e_j),\index(v_j)\mid\index(v_{j+1})}
&= \left(\Dcolor_{\btau(v_j)} -  \ColDist_{\btau(v_j)}\ColDist_{\btau(v_j)}^{\top} \right) \left(\mathbf{I} - \frac{1}{\support(\ColDist_{\btau(v_j)})} \mathbf{1} \mathbf{1}^{\top}\right)\Facm_{\FacType(e_j),\index(v_{j+1})\mid \index(v_j)}^{\top} \Dcolor_{\btau(v_{j+1})}^{\dagger} \\
&= \Dcolor_{\btau(v_j)}\left(\mathbf{I} - \mathbf{1}\ColDist_{\btau(v_j)}^{\top} \right) \left(\mathbf{I} - \frac{1}{\support(\ColDist_{\btau(v_j)})} \mathbf{1} \mathbf{1}^{\top}\right)\Facm_{\FacType(e_j),\index(v_{j+1})\mid \index(v_j)}^{\top} \Dcolor_{\btau(v_{j+1})}^{\dagger}\\
&= \Dcolor_{\btau(v_j)}\left(\mathbf{I} - \mathbf{1}\ColDist_{\btau(v_j)}^{\top} \right)\Facm_{\FacType(e_j),\index(v_{j+1})\mid \index(v_j)}^{\top} \Dcolor_{\btau(v_{j+1})}^{\dagger} \\
&= \left(\mathbf{I} - \ColDist_{\btau(v_j)}\mathbf{1}^{\top} \right)\Dcolor_{\btau(v_j)}\Facm_{\FacType(e_j),\index(v_{j+1})\mid \index(v_j)}^{\top} \Dcolor_{\btau(v_{j+1})}^{\dagger} \\
&= \left(\mathbf{I} - \ColDist_{\btau(v_j)}\mathbf{1}^{\top} \right)\Facm_{\FacType(e_j) ,\index(v_j)\mid \index(v_{j+1})} 
\enspace.
\end{align*}
This establishes the first part of \pref{clm:transitional-matrix}.  To establish the second part, consider the following chain of equalities where the first equality is one we know from the above chain.
\begin{align*}
	\bTrans_{\FacType(e_j),\index(v_j)\mid\index(v_{j+1})} &= \left(\mathbf{I} - \ColDist_{\btau(v_j)}\mathbf{1}^{\top} \right)\Dcolor_{\btau(v_j)}\Facm_{\FacType(e_j),\index(v_{j+1})\mid \index(v_j)}^T \Dcolor_{\btau(v_{j+1})}^{\dagger} \\
	&= \Dcolor_{\btau(v_j)} \left(\mathbf{I} - \mathbf{1}\ColDist_{\btau(v_j)}^{\top} \right)\Facm_{\FacType(e_j),\index(v_{j+1})\mid \index(v_j)}^{\top} \Dcolor_{\btau(v_{j+1})}^{\dagger} \\
	&= \Dcolor_{\btau(v_j)} \left(\Facm_{\FacType(e_j),\index(v_{j+1})\mid \index(v_j)} \left(\mathbf{I} - \ColDist_{\btau(v_j)}\mathbf{1}^{\top} \right) \right)^{\top} \Dcolor_{\btau(v_{j+1})}^{\dagger} \\
	&= \Dcolor_{\btau(v_j)} \left(\Facm_{\FacType(e_j),\index(v_{j+1})\mid \index(v_j)} - \ColDist_{\btau(v_{j+1})}\mathbf{1}^{\top} \right)^{\top} \Dcolor_{\btau(v_{j+1})}^{\dagger} \\
	&= \Dcolor_{\btau(v_j)} \left(\left(\mathbf{I}-C_{\btau(v_{j+1})}\mathbf{1}^{\top}\right)\Facm_{\FacType(e_j),\index(v_{j+1})\mid \index(v_j)} \right)^{\top} \Dcolor_{\btau(v_{j+1})}^{\dagger} \\
	&= \Dcolor_{\btau(v_j)} \BPM{\FacType(e_j)}{\index(v_{j+1})}{\index(v_j)}^{\top} \Dcolor_{\btau(v_{j+1})}^{\dagger}
\end{align*}
If $\ColDist_{\btau(v_j)}(c) = 0$, then the $c$-th column of $\BPM{\FacType(e_j)}{\index(v_{j+1})}{\index(v_j)}$ and the $c$-th row of $\BPM{\FacType(e_j)}{\index(v_{j})}{\index(v_{j+1})}$ are $0$, and hence the second part of \pref{clm:transitional-matrix} follows as well.
\section{Random graph lemmas}\label{app:random-graph-lemmas}

\subsection{Proof of \pref{lem:is-r-bike-free}}
Let $\GModPar$ be a null model, and let $\bH\sim \GModPar$.  We use $\arity_{\max}$ to denote $\max_{i\in[F]} \arity(i)$ and ${\ACFacDeg{}}_{\max}$ to denote $\max_{i\in[F]}\ACFacDeg{i}$.

Define the notion of $(\tau,\Cl)$-consistent:
\begin{definition}
	Given $\tau$ and $\Cl$, we say a subgraph $\Gamma$ of a bipartite factor graph with right vertex set $R(\Gamma)$ is \emph{$(\tau,\Cl)$-consistent} if every $\gamma = (v_1,\dots,v_{\arity(i)})\in R(\Gamma)$ satisfies $(\tau(v_1),\dots,\tau(v_{\arity(i)})) = \Cl(\FacType(\gamma))$.
\end{definition}


It is easy to see the following.
\begin{observation}	\label{obs:prob-of-subgraph}
	Suppose $\Gamma$ is a subgraph of $\Bip(\CompG_{n})$.  Then the probability that $\Gamma$ is a subgraph of $\bH$ is equal to $\prod_{\gamma\in R(\Gamma)}\frac{\ACFacDeg{\FacType(\gamma)}}{n^{\arity(\FacType(\gamma))-1}}\Ind[\Gamma\text{ is $(\tau,\Cl)$-consistent}]$.
\end{observation}

\begin{definition}[Partially labeled graph]
	A \emph{partially labeled graph} $\Gamma = (L, R, P, p, E)$ is given by a left vertex set $L$, a right vertex set $R$, a distinguished set of left vertices $P$ along with an injective labeling of the distinguished vertices $p:P\to[n]$, and edge set $E$.
\end{definition}

\begin{definition}[Occurrence of partially labeled graph]
	An \emph{occurrence} of a partially labeled graph $\Gamma = (L, R, P, p, E)$ in $\Bip(\bH)$ is a pair of injective functions $f_L:L\to L(\bH)$ and $f_R:R\to R(\bH)$ such that for all $v\in P$ satisfies $f_L(v) = p(v)$, and if $\{u,v\} \in E$, then $\{f_L(u), f_R(v)\}\in E(\Bip(\bH))$.
\end{definition}

Given a partially labeled subgraph $\Gamma$ we are interested in bounding the expected number of occurrences of $\Gamma$ in $\Bip(\bH)$.
\begin{lemma}	\label{lem:exp-occurrences}
	Given partially labeled graph $\Gamma = (L, R, P, p)$ with no isolated right vertices, the expected number of occurrences of~$\Gamma$ in $\Bip(\bH)$ is at most
	\[
		n^{|L|+|R|-|P|-|E|} (F\arity_{\max}{\ACFacDeg{}}_{\max})^{|E|}.
	\]
\end{lemma}
\begin{proof}
	There are at most $n^{|L|-|P|}$ choices for $f_L$.  For each potential choice of $f_R$, we can associate $t_{f_R}:R\to[F]$ such that $t_{f_R}(r)$ is the type of $f_R(r)$.  There are at most $F^{|R|}$ possible values for $t_{f_R}$.  For each fixed choice of $f_L$ and $t$, we wish to bound the expected number of $f_R$ such that $(f_L,f_R)$ is an occurrence and $t_{f_R}=t$.  The number of such potential $f_R$ is bounded by
	\[
		\prod_{i\in[R]}\arity(t(i))^{\deg_{\Gamma}(i)} n^{\arity(t(i))-\deg_{\Gamma}(i)}
	\]
	and the probability that $(f_L, f_R)$ is a valid occurrence for a given such $f_R$ is at most
	\[
		\prod_{i\in[R]}\frac{\ACFacDeg{i}}{n^{\arity(t(i))-1}},
	\]
	which gives us a bound of
	\[
		\prod_{i\in[R]} \left(\frac{\arity(t(i))}{n}\right)^{\deg_{\Gamma}(i)}\cdot n\cdot\ACFacDeg{i} \le \left(\frac{\arity_{\max}}{n}\right)^{|E|}n^{|R|}{\ACFacDeg{}}_{\max}^{|R|}.
	\]
	Combining this with the bound on total number of $f_L$ and $t_{f_R}$ gives us a bound of:
	\[
		n^{|L|-|P|} F^{|R|}\left(\frac{\arity_{\max}}{n}\right)^{|E|}n^{|R|}{\ACFacDeg{}}_{\max}^{|R|} = n^{|L|+|R|-|P|-|E|}F^{|R|}\arity_{\max}^{|E|}{\ACFacDeg{}}_{\max}^{|R|}.
	\]
	Since there are no isolated vertices, $|E|\ge|R|$ and hence the above is at most $n^{|L|+|R|-|P|-|E|}(F\arity_{\max}{\ACFacDeg{}}_{\max})^{|E|}$.
\end{proof}

\begin{definition}
	For a graph $\Gamma$ and a subset of its vertices $S$ we use $B_{\Gamma}(S,r)$ to denote the radius-$r$ ball around set $S$ within $\Gamma$.  We also abuse notation and use $B_{\bH}(S,r)$ to mean $B_{\Bip(\bH)}(S,r)$.
\end{definition}

\begin{lemma}	\label{lem:unlikely-subgraphs}
	Given a set of vertices $S$ in $\Bip(\bH)$, the probability that $|E(B_{\bH}(S,r))|-|V(B_{\bH}(S,r))|+|S|\ge t$ is at most $\left(\frac{(F\arity_{\max}{\ACFacDeg{}}_{\max})^{2(r+1)}(36t^3r^2)^5|S|}{n}\right)^t$.
\end{lemma}

In preparation to prove \pref{lem:unlikely-subgraphs} we will need the following statement about counts of trees with a bounded number of leaves.  The statement along with a proof can be found in \cite[Lemma 6.33]{BMR19}.
\begin{lemma}	\label{lem:count-trees-few-leaves}
	The number of nonisomorphic trees on $v$ vertices and $L$ leaves is bounded by $(4Lv)^{2L+1}$.
\end{lemma}

\begin{proof}[Proof of \pref{lem:unlikely-subgraphs}]
	Let us call a partially labeled subgraph $\Gamma = (L, R, P, p, E)$ a \emph{candidate witness} if
	\begin{itemize}
		\item $p(P) = S$,
		\item $\Gamma$ can be expressed as $F\cup B$ where $F$ is a forest and $B=\{\{u_1,v_1\},\dots,\{u_t,v_t\}\}$ is a set of $t$ additional edges,
		\item $F$ has $|P|$ connected components where each connected component contains exactly one $v\in P$ and has depth $r$ when rooted at $v$.
	\end{itemize}
	If $|E(B_{\bH}(S,r))|-|V(B_{\bH}(S,r))|+|S|\ge t$, then there must be an occurrence of some candidate witness $\Gamma = (L, R, P, p, E)$ within $\bH$.  We will first find a ``simple'' subgraph of $\Gamma=F\cup B$, which we call the \emph{trim} of $\Gamma$.  First let us augment $F$ to $\wt{F}$ by adding a single vertex $w$ and connecting it to all vertices in $P$ -- note that $\wt{F}$ is a tree.  Now let $\mathrm{Trim}(\wt{F})$ be the tree obtained by only choosing vertices that lie on paths from vertices in $L\coloneqq\{u_1,\dots,u_t,v_1,\dots,v_t\}$ to $w$.  Since the depth of $\mathrm{Trim}(\wt{F})$ is $r+1$ and has at most $2t$ leaves when rooted at $w$, the number of vertices in $\mathrm{Trim}(\wt{F})$ is at most $2tr+1$.  Let $\mathrm{Trim}(\Gamma) = (L', R', P, p, E')$ be the graph obtained by deleting $w$ from $\mathrm{Trim}(\wt{F})$, adding edges $\{u_1,v_1\},\dots,\{u_t,v_t\}$, and adding vertices in $P\setminus V(\mathrm{Trim}(\wt{F}))$.  Since $\mathrm{Trim}(\Gamma)$ is a subgraph of $\Gamma$ there must be an occurrence of $\mathrm{Trim}(\Gamma)$ in $\bH$.

	$\mathrm{Trim}(\Gamma)$ has at most $2tr$ vertices and $|E(\mathrm{Trim}(\Gamma))|-|L(\mathrm{Trim}(\Gamma))|-|R(\mathrm{Trim}(\Gamma))|+|P| \geq t$.  Thus, from \pref{lem:exp-occurrences} the probability that $\mathrm{Trim}(\Gamma)$ occurs in $\bH$ is bounded by $\left(\frac{(F\arity_{\max}{\ACFacDeg{}}_{\max})^{2(r+1)}}{n}\right)^{t}$.  Thus:
	\begin{align*}
		\Pr[|E(B_{\bH}(S,r))|-|V(B_{\bH}(S,r))|+|S|\ge t] &\le \Pr[\text{there is a candidate witness }\Gamma\text{ in $\bH$}]\\
		&\le \Pr[\text{there is a trim of a candidate witness }\mathrm{Trim}(\Gamma)\text{ in $\bH$}]\\
		&\le \sum_{\substack{\Gamma'\text{ trim of a candidate witness}}} \Pr[\Gamma'\text{ in $\bH$}]\\
		&\le \sum_{\substack{\Gamma'\text{ trim of a candidate witness}}} \left(\frac{(F\arity_{\max}{\ACFacDeg{}}_{\max})^{2(r+1)}}{n}\right)^{t}	\numberthis \label{eq:trim-sum}.
	\end{align*}

	Next we bound the number of terms in the above summation.  Since each $\Gamma'$ in the above sum can be specified by taking a tree on at most $2tr+1$ vertices and at most $2t+1$ leaves, deleting one vertex, and labeling each neighbor of this deleted vertex with an element of $P$, from \pref{lem:count-trees-few-leaves} and the fact the maximum degree in a tree is bounded by the number of leaves the number of terms is at most:
	\[
		(4(2t+1)(2tr+1))^{4t+3}\cdot(2tr+1)\cdot |P|^{2t+1} \le ((36t^3r^2)^5|P|)^t.
	\]
	Plugging this into \pref{eq:trim-sum} and using $|S|=|P|$ gives:
	\[
		\Pr[|E(B_{\bH}(S,r))|-|V(B_{\bH}(S,r))|+|S|\ge t] \le \left(\frac{(F\arity_{\max}{\ACFacDeg{}}_{\max})^{2(r+1)}(36t^3r^2)^5|S|}{n}\right)^t
	\]
\end{proof}

\begin{corollary}	\label{cor:bike-free}
	With probability $1-o_n(1)$, $\Bip(\bH)$ is $r$-bicycle free for $r = \frac{\log n}{\log \log n}$.
\end{corollary}
\begin{proof}
	This is a simple consequence of \pref{lem:unlikely-subgraphs}.  Indeed, by \pref{lem:unlikely-subgraphs} the probability that the radius-$(r+1)$ neighborhood of a single vertex $v\in[n]$ contains more than one cycle is at most $\frac{1}{n^{2-o_n(1)}}$, and hence by a union bound over all vertices the probability of any left vertex containing more than one cycle in its radius-$r+1$ neighborhood is bounded by $\frac{1}{n^{1-o_n(1)}}$.  Since every right vertex is incident to a left vertex, the statement we wish to prove follows.
\end{proof}

\subsection{Proof of \pref{lem:singleton-small-exp}}
We will need the following combinatorial lemma that appears in \cite[Lemma A.2]{FM17}.
\begin{lemma}	\label{lem:bic-frame-to-excess}
	If $e$ distinct edges of a graph $\Gamma$ belong to a $r$-bicycle frame, then $\Exc(\Gamma)\ge\frac{e}{r}$.
\end{lemma}

Our proof of the statement below follows the same strategy as the proof of a similar statement appearing in \cite{FM17}.
\begin{lemma}	\label{lem:singleton-small-exp}
	Suppose $S$ and $L$ are disjoint sets of right vertices of $\CompG_{n}$ of size at most $\log^2 n$, $\Ind_{\gamma}$ is the indicator random variable for whether $\gamma$ is in $\bH$, $\mu_{\gamma}$ is the probability that $\gamma$ is in $\bH$, and $\nobikes$ denotes the event that $\bH$ is $r$-bicycle free for $r = \frac{\log n}{\log\log n}$.  Then:
	\[
		\left|\E\left[\prod_{\gamma\in S}(\Ind_{\gamma}-\Mean_{\gamma})\prod_{\gamma\in L}\Ind_{\gamma} \Ind[\nobikes]\right]\right| \le \prod_{\gamma\in S\cup L}\Mean_{\gamma}\cdot 2^{|S|} \left(\frac{1}{n^{.5}}\right)^{\frac{|S|}{r}-\Exc(\Clos(S\cup L))}.
	\]
\end{lemma}
\begin{proof}
	\begin{align*}
		\left|\E\left[\prod_{\gamma\in S}(\Ind_{\gamma}-\Mean_{\gamma})\prod_{\gamma\in L}\Ind_{\gamma} \Ind[\nobikes]\right]\right| &= \left|\sum_{J\subseteq S} \E\left[ \prod_{\gamma\in J}\Ind_{\gamma}\cdot\prod_{\gamma\in S\setminus J} (-\Mean_{\gamma})\cdot\prod_{\gamma\in L}\Ind_{\gamma}\Ind[\nobikes]\right]\right|\\
		&= \left|\sum_{J\subseteq S} (-1)^{|S|-|J|}\prod_{\gamma\in S\setminus J}\Mean_{\gamma}\E\left[\prod_{\gamma\in J\cup L} \Ind_{\gamma} \Ind[\nobikes]\right]\right|\\
		&= \left|\sum_{J\subseteq S} (-1)^{|J|}\prod_{\gamma\in S\cup L}\Mean_{\gamma}\Pr\left[\nobikes|\gamma\in\bH~\forall\gamma\in J\cup L\right]\right| \\
		&= \prod_{\gamma\in S\cup L}\Mean_{\gamma}\left|\sum_{J\subseteq S} (-1)^{|J|}\Pr\left[\nobikes|\gamma\in\bH~\forall\gamma\in J\cup L\right]\right|	\numberthis \label{eq:exp-signed-sum}
	\end{align*}
	Now we focus our attention on understanding the quantity $\left|\sum_{J\subseteq S} (-1)^{|J|}\Pr\left[\nobikes|\gamma\in\bH~\forall\gamma\in J\cup L\right]\right|$.  Let $\bg_0$ be $\bH\setminus (S\cup L)$.
	\begin{align}
		\left|\sum_{J\subseteq S} (-1)^{|J|} \Pr\left[\nobikes|\gamma\in\bH~\forall\gamma\in J\cup L\right]\right|
		&= \left| \E_{\bg_0} \sum_{J\subseteq S} (-1)^{|J|}\Pr[\nobikes|\gamma\in\bH~\forall\gamma\in J\cup L, \bg_0] \right|	\label{eq:introducing-g0}
	\end{align}
	For $K\subseteq S$, define $f_{\bg_0}(K)$ as $1$ if $\Clos(\bg_0\cup K\cup L)$ has no $r$-bicycles and $0$ otherwise.  Suppose there is $s\in S$ that $f_{\bg_0}$ does not depend on -- that is, for any $K\subseteq S$, $f_{\bg_0}(K) = f_{\bg_0}(K\Delta \{s\})$, then for every $J$ which contains $s$:
	\[
		\Pr[\nobikes|\gamma\in\bH~\forall\gamma\in J\cup L, \bg_0] = \Pr[\nobikes|\gamma\in\bH~\forall\gamma\in J\cup L\setminus\{s\}, \bg_0].
	\]
	This means \pref{eq:introducing-g0} is equal to:
	\begin{align*}
		\pref{eq:introducing-g0} &= \left|\E_{\bg_0} \Ind[f_{\bg_0}\text{ depends on every $s\in S$}] \sum_{J\subseteq S} (-1)^{|J|}\Pr\left[\nobikes|\gamma\in\bH~\forall\gamma\in J\cup L, \bg_0\right]\right| \\
		&\le 2^{|S|}\cdot\Pr_{\bg_0}[f_{\bg_0}\text{ depends on every $s\in S$}].	\numberthis \label{eq:prob-bound-dep}
	\end{align*}
	Let $E_S$ be the set of all edges incident to $S$.  If $f_{\bg_0}$ depends on every $s\in S$, then the function $h_{\bg_0}$ defined on subsets of $E_S$ which is $1$ on input $K\subseteq E_S$ if $\Clos(\bg_0\cup L)\cup K$ has no $r$-bicycles depends on at least $|S|$ edges in $E_S$.  That means:
	\begin{align*}
		\pref{eq:prob-bound-dep} &\le 2^{|S|}\cdot\Pr_{\bg_0}[h_{\bg_0}\text{ depends on at least $|S|$ edges in $E_S$}]\\
		&\le 2^{|S|}\cdot\Pr_{\bg_0}[\text{At least $|S|$ edges in $E_S$ part of $r$-bicycle frame in $\Clos(\bg_0\cup S\cup L)$}]
		\intertext{which, via \pref{lem:bic-frame-to-excess}, can be bounded by}
		&\le 2^{|S|}\cdot\Pr_{\bg_0}\left[\Exc(B_{\Clos(\bg_0\cup S\cup L)}(\Clos(S), r)) \ge \frac{|S|}{r} \right] \\
		&\le 2^{|S|}\cdot\Pr_{\bg_0}\left[\Exc(B_{\Clos(\bg_0)}(\Clos(S\cup L), r)) \ge \frac{|S|}{r}-\Exc(\Clos(S\cup L))\right] \\
		&\le 2^{|S|}\cdot\Pr_{\bH}\left[\Exc(B_{\bH}(\Clos(S\cup L), r)) \ge \frac{|S|}{r} - \Exc(\Clos(S\cup L)) \right]
		\intertext{By \pref{lem:unlikely-subgraphs}, the bounds on size of $|S|$ and $|L|$, and the value of $r$, we can conclude that the above is at most:}
		&\le 2^{|S|} \min\left\{\left(\frac{1}{n^{.5}}\right)^{\frac{|S|}{r}-\Exc(\Clos(S\cup L))}, 1\right\}\\
		&\le 2^{|S|} \left(\frac{1}{n^{.5}}\right)^{\frac{|S|}{r}-\Exc(\Clos(S\cup L))}.
	\end{align*}
	Plugging this back into \pref{eq:exp-signed-sum} gives us the desired statement.
\end{proof}

\subsection{Proof of \pref{clm:shape-count}}
\begin{proof}
Consider any $\Sh\in\mathcal{U}_{s,x,\Profl}$. Recall that $S(\Sh)$ is the set of singleton vertices in $R(\Sh)$ and $\Dup(\Sh)$ the set of duplicative vertices in $R(\Sh)$. $\Dup^*(\Sh)$ is the maximum weight subset of $\Dup(W)$ that makes $\Sh$ $r$-bicycle free, and $\Profl(\Sh) = w(\Dup(\Sh)) - w(\Dup^*(\Sh))$. Thus we deduce that $\Profl(\Sh) \geq \abs{\Dup(\Sh)\setminus\Dup^*(\Sh)} $ We apply the following procedures to the walk $\Sh$.
\begin{enumerate}
\item Break the walk $\Sh$ into $\leq s + \Profl + \abs{\Dup(\Sh)\setminus\Dup^*(\Sh)} + k$ segments by first removing the vertices $S(\Sh)$ and $\Dup(\Sh)\setminus\Dup^*(\Sh)$ from $\Sh$ and second breaking the remaining segments at endpoints of the $k$ links in $\Sh$. Denote the new union of walks $\Sh_1$.
\item Since $S(\Sh)$ and $\Dup(\Sh)\setminus\Dup^*(\Sh)$ are removed from $\Sh$, $\Sh_1$ is singleton free, and the graph on $\Sh_1$, denoted by $\bG(\Sh_1)$, is $r$-bicycle free. 
\item We contract the graph $\bG(\Sh_1)$ by merging all adjacent edges that share a degree-2 vertex. We denote the resulting graph $\bG(\Sh_1)_c$, and we note that the vertices left in $\bG(\Sh_1)_c$ are those with degree $\geq 3$ in $\bG(\Sh_1)$. 
\end{enumerate}

We make the following observations on the size of $\bG(\Sh_1)$ and $\bG(\Sh_1)_c$. 

The number of vertices in $\bG(\Sh_1)$ is $\abs{V(\Sh)} - \abs{S(\Sh)} - \abs{\Dup(\Sh)\setminus\Dup^*(\Sh)}$. The number of edges in $\bG(\Sh_1)$ is $\leq \abs{E(\Sh)} - 2\abs{S(\Sh)} - 2\abs{\Dup(\Sh)\setminus\Dup^*(\Sh)}$.

To bound the number of vertices in $\bG(\Sh_1)_c$ is we apply the following lemma from \cite{MOP20a}. 

\begin{lemma}[{Lemma 6.18 in \cite{MOP20a}}]
Let $C$ be a $(k,2\ell)$-nonbacktracking, internally $2\ell$-bicycle-free linkage. Assume $\log{k\ell} = o(\ell)$. Then $\bG(C)$ has at most $O(k\log{k\ell})$ vertices of degree exceeding $2$.
\end{lemma}

Applying the lemma to the walk $\Sh_1$, we obtain that the number of degree $\geq 3$ vertice in $\Sh_1$ is $O(k\log{k\ell})$. Thus the number of vertices in $\bG(\Sh_1)_c$ is $O(k\log{k\ell})$. The number of edges in $\bG(\Sh_1)_c$ is 

\begin{align*}
&~~~\abs{E(\bG(\Sh_1))} - \left(\abs{V(\bG(\Sh_1))} - \abs{V(\bG(\Sh_1)_c)}\right) \\
&\le \abs{E(\Sh)} - 2\abs{S(\Sh)} - 2\abs{\Dup(\Sh)\setminus\Dup^*(\Sh)} - \abs{V(\Sh)} + \abs{S(\Sh)} + \abs{\Dup(\Sh)\setminus\Dup^*(\Sh)} + O(k\log{k\ell}) \\
&= (\abs{E(\Sh)} - \abs{V(\Sh)})  - \abs{S(\Sh)} - \abs{\Dup(\Sh)\setminus\Dup^*(\Sh)} + O(k\log{k\ell}) \\
&\le  x + O(k\log{k\ell}).
\end{align*}

Now to count the number of distinct $\Sh\in\mathcal{U}_{s,x,\Profl}$, it suffices to count 1. the number of distinct sets of breaking points ($S(\Sh),\Dup(\Sh)\setminus\Dup^*(\Sh),$ and $k$ link endpoints), 2. the number of distinct graphs $\bG(\Sh_1)$, 3. given the breaking points and $\bG(\Sh_1)$, the number of distinct walk segments in $\bG(\Sh_1)$ with those breaking points. We count each of the three quantities separately and multiply them together to obtain an upper bound on $\abs{\mathcal{U}_{s,x,\Profl}}$.

The number of distinct sets of breaking points: these breaking points breaks $\Sh$ into at most $s + 2\Profl + k$ segments. So there are $(2k\ell)^{s + 2\Profl + k}$ ways to choose these breaking points. 

The number of distinct graphs $\bG(\Sh_1)$: $\bG(\Sh_1)$ can be contracted to a graph $\bG(\Sh_1)_c$ on $O(k\log{k\ell})$ vertices and $ x + O(k\log{k\ell})$ edges. Each edge in $\bG(\Sh_1)_c$ represents a length-$\leq 2k\ell$ simple path in $\bG(\Sh_1)$. Thus there are $O(k\log{k\ell})^{2(x + O(k\log{k\ell}))}\cdot (2k\ell)^{x + O(k\log{k\ell})}$ distinct graphs $\bG(\Sh_1)$.

Given the breaking points and $\bG(\Sh_1)$, the number of distinct walk segments in $\bG(\Sh_1)$ with those breaking points: since $\bG(\Sh_1)$ is $r$-bicycle free with $r \geq 2\ell$ and each segment is of length $\leq 2\ell$, there are only $2$ distinct length $\leq 2\ell$ walk between any two vertices in $\bG(\Sh_1)$. Thus the number of distinct walk segments are $2^{s + 2\Profl + k}$.

Combine the three bound together we obtain that 
\[\abs{\mathcal{U}_{s,x,\Profl}} \leq (2k\ell)^{s + 2\Profl + x + O(k\log{k\ell})} O(k\log{k\ell})^{2(x + O(k\log{k\ell}))}  2^{s + 2\Profl + k} .\]

From this bound we quickly derive that 
\begin{align*}
 &~\sum_{\Sh\in\mathcal{U}_{s,x,\Profl}}
\left(\frac{\beta^r}{n^{.5}}\right)^{\frac{|S(\Sh)|}{r}}
	\left(\frac{\beta}{n^{.5}}\right)^{\Exc_{\Sh}}
	\left(\frac{\beta}{n}\right)^{\Profl(\Sh)}  \\
&=\abs{\mathcal{U}_{s,x,\Profl}} \cdot \left(\frac{\beta^r}{n^{.5}}\right)^{\frac{s}{r}}
	\left(\frac{\beta}{n^{.5}}\right)^{x} 
	\left(\frac{\beta}{n}\right)^{\Profl} \\
&= \left(\frac{\beta^r\cdot(4k\ell)^r}{n^{.5}}\right)^{\frac{s}{r}}
	\left(\frac{\beta\cdot 2k\ell \cdot O(k\log{k\ell})^2}{n^{.5}}\right)^{x} 
	\left(\frac{4\beta\cdot(2k\ell)^2}{n}\right)^{\Profl} 2^k(2k\ell \cdot O(k\log{k\ell})^2)^{O(k\log{k\ell})} \\
&\leq \left(\frac{\beta^r\cdot(4k\ell)^r}{n^{.5}}\right)^{\frac{s}{r}}
	\left(\frac{\beta\cdot 2(k\ell)^3}{n^{.5}}\right)^{x} 
	\left(\frac{4\beta\cdot(2k\ell)^2}{n}\right)^{\Profl}(2k\ell)^{O(k\log{k\ell})} 
\end{align*}
The last inequality follows since we pick $\ell$ such that $\log{k\ell} = o(\ell)$.
\end{proof}

\end{document}